\documentclass[sigconf, nonacm]{acmart}

\newcommand\vldbdoi{XX.XX/XXX.XX}
\newcommand\vldbpages{XXX-XXX}
\newcommand\vldbvolume{14}
\newcommand\vldbissue{1}
\newcommand\vldbyear{2025}
\newcommand\vldbauthors{\authors}
\newcommand\vldbtitle{\shorttitle}
\newcommand\vldbavailabilityurl{URL_TO_YOUR_ARTIFACTS}
\newcommand\vldbpagestyle{plain}

\AtBeginDocument{%
  \providecommand\BibTeX{{%
    \normalfont B\kern-0.5em{\scshape i\kern-0.25em b}\kern-0.8em\TeX}}}

\usepackage{amsmath,amsfonts}
\usepackage{algorithm}
\usepackage{algorithmicx}
\usepackage[noend]{algpseudocode}
\usepackage{xfrac}

\definecolor{white}{rgb}{1,1,1}
\definecolor{black}{rgb}{0,0,0}
\definecolor{grey}{rgb}{0.7,0.7,0.7}
\definecolor{blue}{rgb}{0,0,1}
\definecolor{green}{rgb}{0,1,0}
\definecolor{red}{rgb}{1,0,0}
\definecolor{yellow}{rgb}{1.0, 1.0, 0.0}

\definecolor{lightgrey}{rgb}{0.88,0.88,0.88}
\definecolor{lgrey}{rgb}{0.9,0.9,0.9}
\definecolor{llgrey}{rgb}{0.93,0.93,0.93}
\definecolor{lllgrey}{rgb}{0.96,0.96,0.96}
\definecolor{tableHeadGray}{rgb}{0.85,0.85,0.85}
\definecolor{oddRowGrey}{rgb}{0.95,0.95,0.95}
\definecolor{evenRowGrey}{rgb}{0.85,0.85,0.85}
\definecolor{lightyellow}{rgb}{1.0, 1.0, 0.88}
\definecolor{shadered}{rgb}{1,0.85,0.85}
\definecolor{shadegreen}{rgb}{0.95,1,0.95}
\definecolor{shadeblue}{rgb}{0.95,0.95,1}
\definecolor{selectiveyellow}{rgb}{1.0, 0.73, 0.0}

\definecolor{darkred}{rgb}{0.5,0,0}
\definecolor{darkgrey}{rgb}{0.5,0.5,0.5}
\definecolor{darkgreen}{rgb}{0,0.5,0}
\definecolor{darkblue}{rgb}{0,0,0.5}
\definecolor{darkpurple}{rgb}{0.5,0,0.5}
\definecolor{darkdarkpurple}{rgb}{0.3,0,0.3}

\algrenewcommand{\algorithmiccomment}[1]{\hskip0.5em$\triangleright$ \textcolor{darkpurple}{\footnotesize \textit{#1}}}

\usepackage[acronym]{glossaries}
\usepackage{glossaries-extra}

\setabbreviationstyle[acronym]{long-short}

\newacronym{algff}{FF}{Full Cluster Filtering}
\newacronym{algrp}{RP}{Cluster Range Pruning}
\newacronym{algbf}{BF}{Brute Force}
\newglossaryentry{erica}{name=Erica,description={}}
\newglossaryentry{acrp}{name=aggregate constraint repair problem,description={}}
\newacronym{SPD}{SPD}{statistical parity difference}
\newacronym{SPJ}{SPJ}{select-project-join}
\newacronym{AC}{AC}{aggregate constraint}
\newacronym{NCE}{NCE}{number of candidates evaluated}
\newacronym{NCA}{NCA}{number of clusters accessed}
\newacronym{ED}{ED}{exploration distance}
\newacronym{dshealth}{Healthcare}{Healthcare dataset}
\newacronym{dsacs}{Income}{Adult Census Income}
\newacronym{OOM}{OOM}{order of magnitude}
\newglossaryentry{tpch}{name=TPC-H,description={}}
\newglossaryentry{repairproblem}{name=Aggregate Constraint Repair Problem,description={}}

\usepackage{titlesec}

\titleclass{\subsubsubsection}{straight}[\subsubsection]
\newcounter{subsubsubsection}[subsubsection]
\renewcommand{\thesubsubsubsection}{\thesubsubsection.\arabic{subsubsubsection}}
\titleformat{\subsubsubsection}
  {\normalfont\normalsize\itshape}{\thesubsubsubsection}{1em}{}
\titlespacing*{\subsubsubsection}{0pt}{1.5ex plus 1ex minus .2ex}{0.5ex plus .2ex}

\usepackage{graphicx}
\usepackage{textcomp}
\usepackage{bmpsize}
\usepackage{xcolor}
\usepackage{lipsum}
\usepackage{multirow}
\usepackage{comment}
\usepackage{titlesec}
\titlespacing*{\subsubsection}{0pt}{5pt}{3pt}
\titlespacing*{\subsection}{0pt}{5pt}{3pt}
\usepackage{listings}
\usepackage{xspace}
\usepackage{float}
\usepackage{balance}
\usepackage{tabularx}
\usepackage{adjustbox}
\usepackage[most]{tcolorbox}
\usepackage{caption}
\usepackage{subcaption}
\usepackage{booktabs}
\usepackage{array}
\def\BibTeX{{\rm B\kern-.05em{\sc i\kern-.025em b}\kern-.08em \TeX}}
\usepackage{hyperref}
\usepackage{cleveref}
\usepackage{enumitem}
\setlist[itemize]{topsep=3pt, partopsep=0pt, leftmargin=*}

\usepackage{tcolorbox}
\usepackage{pifont}

\newtheorem{example}{Example}

\newbool{CompileTechReport}
\booltrue{CompileTechReport}
\newcommand{\redx}{\textcolor{red}{\ding{55}}\xspace}
\newcommand{\greenok}{\textcolor{darkgreen}{\ding{51}}\xspace}

\newcommand{\query}{\ensuremath{Q}\xspace}
\newcommand{\Refquery}{\ensuremath{\query_{fix}}\xspace}
\newcommand{\repairspace}[1]{\ensuremath{\textsc{Cand}_{#1}}\xspace}
\newcommand{\qdist}{\ensuremath{d}\xspace}
\newcommand{\dweight}{w}
\newcommand{\db}{\ensuremath{D}\xspace}

\newcommand{\propconstr}{\ensuremath{\omega}\xspace}
\newcommand{\acperc}[2]{\ensuremath{\omega}_{#1}^{d=#2}\xspace}
\newcommand{\propconstrset}{\ensuremath{\Omega}\xspace}
\newcommand{\defas}{\mathop{:=}}
\newcommand{\propthresh}{\ensuremath{\tau}\xspace}

\newcommand{\aggq}[1]{\query^{\omega}_{#1}}
\newcommand{\aaggq}{\aggq{}}
\newcommand{\aggpos}[1]{\query^{\omega\,+}_{#1}}
\newcommand{\aggneg}[1]{\query^{\omega\,-}_{#1}}

\newcommand{\expr}{\Phi}
\newcommand{\aggregation}{\gamma}
\newcommand{\f}[1]{\mathbf{#1}}
\newcommand{\fcall}[1]{\text{\textsc{#1}}}
\newcommand{\fullcoveringclusterset}{\fcall{FullCoverClusterSet}\xspace}
\newcommand{\parcoveringclusterset}{\fcall{ParCoverClusterSet}\xspace}
\newcommand{\rangedivide}{\fcall{RangeDivide}\xspace}
\newcommand{\hascandidates}{\fcall{hasCandidates}\xspace}
\newcommand{\topkconcretecand}{\fcall{topkConcreteCand}\xspace}

\newcommand{\branchfactor}{\ensuremath{\mathcal{B}}\xspace}

\newcommand{\bucketsize}{\ensuremath{\mathcal{S}}\xspace}
\newcommand{\pqueue}{queue}
\newcommand{\rsetcur}{\Repqueryset_{cur}}
\newcommand{\rsetnext}{\Repqueryset_{next}}
\newcommand{\rsetnew}{\Repqueryset_{new}}
\newcommand{\rsettopk}{{rcand}}
\newcommand{\settopk}{\queryset_{top-k}}
\newcommand{\acrangeevalall}{\fcall{aceval}_{\forall}}
\newcommand{\acrangeevalexists}{\fcall{aceval}_{\exists}}

\newcommand{\selection}{\sigma}
\newcommand{\projection}{\pi}
\newcommand{\join}{\bowtie}
\newcommand{\aggf}{\ensuremath{f}\xspace}

\newcommand{\constant}{\ensuremath{c}\xspace}

\newcommand{\attr}{\ensuremath{a}\xspace}

\newcommand{\numcond}{\ensuremath{m}\xspace}
\newcommand{\numtables}{\ensuremath{l}\xspace}
\newcommand{\numaggs}{\ensuremath{n}}
\newcommand{\numrelations}{\ensuremath{z}}
\newcommand{\numdoma}[1]{\ensuremath{N_{#1}}}

\newcommand{\cluster}{\ensuremath{C}\xspace}
\newcommand{\clusterset}{\ensuremath{\mathbf{C}}\xspace}
\newcommand{\fullclusterset}{\ensuremath{\mathbf{C}_{full}}\xspace}
\newcommand{\parclusterset}{\ensuremath{\mathbf{C}_{partial}}\xspace}
\newcommand{\boundsa}[1]{\ensuremath{\text{\textsc{bounds}}_{#1}}\xspace}

\newcommand{\rel}{\ensuremath{R}\xspace}
\newcommand{\ccount}{\ensuremath{\f{count}}\xspace}
\newcommand{\btrue}{\mathbf{true}}
\newcommand{\bfalse}{\mathbf{false}}

\newcommand{\evalall}{\ensuremath{\f{eval}_{\forall}}\xspace}
\newcommand{\evalallrange}{\ensuremath{\f{reval}_{\forall}}\xspace}
\newcommand{\evalsomerange}{\ensuremath{\f{reval}_{\exists}}\xspace}
\newcommand{\lb}[1]{\ensuremath{\underline{#1}}\xspace}
\newcommand{\ub}[1]{\ensuremath{\overline{#1}}\xspace}
\newcommand{\Repqueryset}{\ensuremath{\mathbb{Q}}\xspace}
\newcommand{\queryset}{\ensuremath{\mathcal{Q}}\xspace}

\newcommand{\tree}{kd-tree\xspace}

\newcommand{\card}[1]{\ensuremath{|{#1}|}\xspace}
\newcommand{\mylbag}{\{\!\!\{}
\newcommand{\myrbag}{\}\!\!\}}

\newcommand{\mypar}[1]{\smallskip\noindent{\textbf{#1}.}}
\DeclareMathOperator{\op}{op}

\DeclareMathOperator*{\topk}{top-k}

\newcommand{\detailedproof}[2]{\ifbool{CompileTechReport}{
\begin{proof}
#2
\end{proof}
}{\noindent\textit{Proof Sketch:}#1\qed}\smallskip}
\newcommand{\ifnottechreport}[1]{\ifbool{CompileTechReport}{}{#1}}
\newcommand{\iftechreport}[1]{\ifbool{CompileTechReport}{#1}{}}

\definecolor{lstpurple}{rgb}{0.5,0,0.5}
\definecolor{lstred}{rgb}{1,0,0}
\definecolor{lstreddark}{rgb}{0.7,0,0}
\definecolor{lstredl}{rgb}{0.64,0.08,0.08}
\definecolor{lstmildblue}{rgb}{0.66,0.72,0.78}
\definecolor{lstblue}{rgb}{0,0,1}
\definecolor{lstmildgreen}{rgb}{0.42,0.53,0.39}
\definecolor{lstgreen}{rgb}{0,0.5,0}
\definecolor{lstorangedark}{rgb}{0.6,0.3,0}
\definecolor{lstorange}{rgb}{0.75,0.52,0.005}
\definecolor{lstorangelight}{rgb}{0.89,0.81,0.67}
\definecolor{lstbeige}{rgb}{0.90,0.86,0.45}

\DeclareFontShape{OT1}{cmtt}{bx}{n}{<5><6><7><8><9><10><10.95><12><14.4><17.28><20.74><24.88>cmttb10}{}

\lstdefinestyle{psql}
{
tabsize=2,
basicstyle=\small\upshape\ttfamily,
language=SQL,
morekeywords={PROVENANCE,BASERELATION,INFLUENCE,COPY,ON,TRANSPROV,TRANSSQL,TRANSXML,CONTRIBUTION,COMPLETE,TRANSITIVE,NONTRANSITIVE,EXPLAIN,SQLTEXT,GRAPH,IS,ANNOT,THIS,XSLT,MAPPROV,cxpath,OF,TRANSACTION,SERIALIZABLE,COMMITTED,INSERT,INTO,WITH,SCN,UPDATED,PARTITION,BY,PRECEDING,FOLLOWING,CURRENT,ROW,ROWS,RANGE,GROUPING,SETS,CUBE,ROLL,UP},
extendedchars=false,
keywordstyle=\bfseries,
mathescape=true,
escapechar=@,
sensitive=true
}

\lstdefinestyle{psqlcolor}
{
tabsize=2,
basicstyle=\small\upshape\ttfamily,
language=SQL,
morekeywords={PROVENANCE,BASERELATION,INFLUENCE,COPY,ON,TRANSPROV,TRANSSQL,TRANSXML,CONTRIBUTION,COMPLETE,TRANSITIVE,NONTRANSITIVE,EXPLAIN,SQLTEXT,GRAPH,IS,ANNOT,THIS,XSLT,MAPPROV,OF,TRANSACTION,SERIALIZABLE,COMMITTED,INSERT,INTO,WITH,SCN,UPDATED},
extendedchars=false,
keywordstyle=\bfseries\color{lstpurple},
deletekeywords={count,min,max,avg,sum},
keywords=[2]{count,min,max,avg,sum,first,last,lead,lag,cxpath},
keywordstyle=[2]\color{lstblue},
stringstyle=\color{lstreddark},
commentstyle=\color{lstgreen},
mathescape=true,
escapechar=@,
sensitive=true
}

\lstdefinestyle{datalog}
{
basicstyle=\footnotesize\upshape\ttfamily,
language=prolog
}

\lstdefinestyle{pseudocode}
{
  tabsize=3,
  basicstyle=\small,
  language=c,
  morekeywords={if,else,foreach,case,return,in,or},
  extendedchars=true,
  mathescape=true,
  literate={:=}{{$\gets$}}1 {<=}{{$\leq$}}1 {!=}{{$\neq$}}1 {append}{{$\listconcat$}}1 {calP}{{$\cal P$}}{2},
  keywordstyle=\color{lstpurple},
  escapechar=&,
  numbers=left,
  numberstyle=\color{lstgreen}\small\bfseries,
  stepnumber=1,
  numbersep=5pt,
}

\lstdefinestyle{xmlstyle}
{
  tabsize=3,
  basicstyle=\small\upshape\ttfamily,
  language=xml,
  extendedchars=true,
  mathescape=true,
  escapechar=£,
  tagstyle=\bfseries,
  usekeywordsintag=true,
  morekeywords={alias,name,id},
  keywordstyle=\color{lstred}
}

\lstdefinestyle{xmlstyle-color}
{
  tabsize=3,
  basicstyle=\small\upshape\ttfamily,
  language=xml,
  extendedchars=true,
  mathescape=true,
  escapechar=£,
  tagstyle=\color{keywordpurple},
  usekeywordsintag=true,
  morekeywords={alias,name,id},
  keywordstyle=\color{lstred}
}

\lstdefinestyle{psql}
{
tabsize=2,
basicstyle=\small\upshape\ttfamily,
language=SQL,
morekeywords={PROVENANCE,BASERELATION,PRECEDING,CURRENT,ROW,INFLUENCE,PARTITION,OVER,COPY,ON,TRANSPROV,TRANSSQL,TRANSXML,CONTRIBUTION,COMPLETE,TRANSITIVE,NONTRANSITIVE,EXPLAIN,SQLTEXT,GRAPH,IS,ANNOT,THIS,XSLT,MAPPROV,cxpath,OF,TRANSACTION,SERIALIZABLE,COMMITTED,INSERT,INTO,WITH,SCN,UPDATED},
extendedchars=false,
keywordstyle=\bfseries,
mathescape=true,
escapechar=@,
sensitive=true
}

\usepackage[disable]{todonotes}

\DeclareRobustCommand{\BGI}[1]{{\todo[inline,color=red!40,size=\footnotesize]{\textbf{Boris says: }{#1}}}}
\DeclareRobustCommand{\BGDel}[2]{{\todo[inline,color=red!40,size=\footnotesize]{\textbf{Boris deleted:}{#1}\textbf{because
      {#2}}}}}

\definecolor{revblue}{rgb}{0,0,1}       
\definecolor{revgreen}{rgb}{0,0.5,0}    
\definecolor{revpurple}{rgb}{0.5,0,0.5} 

\newrobustcmd{\revone}[1]{\textcolor{revblue}{#1}}   
\newrobustcmd{\revtwo}[1]{\textcolor{revgreen}{#1}}  
\newrobustcmd{\revthree}[1]{\textcolor{revpurple}{#1}} 
\newrobustcmd{\revall}[1]{\textcolor{brown}{#1}}   

\setcopyright{acmlicensed}
\copyrightyear{2024}
\acmYear{2024}
\acmDOI{XXXXXXX.XXXXXXX}

\acmConference[Conference acronym 'XX]{Make sure to enter the correct
  conference title from your rights confirmation email}{June 03--05,
  2018}{Woodstock, NY}
\acmISBN{978-1-4503-XXXX-X/18/06}

\graphicspath{{./images/}}

\begin{document}

\title{Efficient Query Repair for Aggregate Constraints}
\subtitle{(extended version)}

\author{Shatha Algarni}
\orcid{}
\affiliation{%
  \institution{University of Southampton}
  \streetaddress{}
  \city{}
  \state{}
  \country{}
  \postcode{}
}
\affiliation{%
  \institution{University of Jeddah}
  \streetaddress{}
  \city{}
  \state{}
  \country{}
  \postcode{}
}
\email{s.s.algarni@soton.ac.uk}

\author{Boris Glavic}
\affiliation{%
  \institution{University of Illinois, Chicago}
  \city{}
  \country{}
  }
\email{bglavic@uic.edu}

\author{Seokki Lee}
\affiliation{%
  \institution{University of Cincinnati}
  \city{}
  \country{}
}
\email{lee5sk@ucmail.uc.edu}

\author{Adriane Chapman}
\affiliation{%
 \institution{University of Southampton}
 \city{}
 \country{}
 }
 \email{adriane.chapman@soton.ac.uk}

\renewcommand{\shortauthors}{Algarni et al.}

\begin{abstract}
In many real-world scenarios, query results must satisfy domain-specific constraints. For instance, a minimum percentage of interview candidates selected based on their qualifications should be female. These requirements can be expressed as constraints over an arithmetic combination of aggregates evaluated on the result of the query. In this work, we study how to repair a query to fulfill such constraints by modifying the filter predicates of the query.
  We introduce a novel query repair technique that leverages bounds on sets of candidate solutions and interval arithmetic to efficiently prune the search space. 
We demonstrate experimentally, that our technique significantly outperforms baselines that consider a single candidate at a time.
\end{abstract}
\maketitle

\pagestyle{\vldbpagestyle}
\begingroup\small\noindent\raggedright\textbf{PVLDB Reference Format:}\\
\vldbauthors. \vldbtitle. PVLDB, \vldbvolume(\vldbissue): \vldbpages, \vldbyear.\\
\href{https://doi.org/\vldbdoi}{doi:\vldbdoi}
\endgroup
\begingroup
\renewcommand\thefootnote{}\footnote{\noindent
This work is licensed under the Creative Commons BY-NC-ND 4.0 International License. Visit \url{https://creativecommons.org/licenses/by-nc-nd/4.0/} to view a copy of this license. For any use beyond those covered by this license, obtain permission by emailing \href{mailto:info@vldb.org}{info@vldb.org}. Copyright is held by the owner/author(s). Publication rights licensed to the VLDB Endowment. \\
\raggedright Proceedings of the VLDB Endowment, Vol. \vldbvolume, No. \vldbissue\ %
ISSN 2150-8097. \\
\href{https://doi.org/\vldbdoi}{doi:\vldbdoi} \\
}\addtocounter{footnote}{-1}\endgroup

\ifdefempty{\vldbavailabilityurl}{}{
\vspace{.3cm}
\begingroup\small\noindent\raggedright\textbf{PVLDB Artifact Availability:}\\
The source code, data, and/or other artifacts have been made available at \url{https://github.com/ShathaSaad/Query-Refinement-of-Complex-Constraints}.
\endgroup
}

\section{Introduction} \label{sec:Introducton}
Analysts are typically well versed in writing queries that return data based on obvious conditions, e.g., only return applicants with a master's degree. However,  a query result often has to fulfill additional constraints, e.g. fairness, that do not naturally translate into conditions. While for some applications it is possible to filter the results of the query to fulfill such constraints this is not always viable, e.g., because the same selection criterion has to be used for all job applicants. Thus, the query has to be repaired such that the result set of the fixed query satisfies all constraints.
Prior work in this area, including query-based explanations~\cite{tran2010conquer, ChapmanJ09} and repairs~\cite{BidoitHT16} for missing answers, work on answering why-not questions~\cite{ChapmanJ09, Benjelloun2006ULDBsDW} as well as query refinement / relaxation approaches~\cite{VW97, MK09, LM23} determine why specific tuples are not in the query's result or how to fix the query to return such tuples. In this work, we study a more general problem where the \emph{entire result set} of the query has to fulfill some constraint. The constraints we study in this work are expressive enough to guarantee query results adhere to legal and ethical regulations, such as fairness. 
\BGDel{Typically, it is challenging to express such constraints through the selection conditions of a query.
Consider the following example.}{space}



\begin{example}[\iftechreport{Fairness }Motivating Example]\label{ex:motivating-example}
Consider a job applicant dataset $D$ for a tech-company that contains six attributes: \texttt{ID, Gender, Field, GPA, TestScore,} and \texttt{OfferInterview}. The attribute \texttt{OfferInterview} was generated by an external ML model suggesting which candidates should receive an interview. The employer uses the query shown below to prescreen candidates: every candidate should be a CS graduate and should have a high GPA and test score.


\begin{lstlisting}[language=SQL]
@\small\upshape\bf\ttfamily{Q1}@: SELECT * FROM D WHERE Major = 'CS'
          AND TestScore $\geq$ 33 AND GPA $\geq$ 3.80
\end{lstlisting}

\mypar{Aggregate Constraint}
The employer 
wants to ensure that interview decisions are not biased against a specific gender using \gls{SPD}~\cite{BellamyDHHHKLMM19,MehrabiMSLG21}. given two groups (e.g., male and female) and a binary outcome attribute $Y$ where $Y=1$ is assumed to be a positive outcome (\texttt{OfferInterview}=1 in our case), the \gls{SPD} 
is the difference between the probability for individuals from the two groups to receive a positive outcome. 
In our example, the \gls{SPD} can be computed as shown below ($G$ is \texttt{Gender}  and $Y$ is \texttt{OfferInterview}). 
We use 
\( \f{count}(\theta) \) to denote the number of query results satisfying condition $\theta$. For example, \( \f{count}(G = M \land Y = 1) \) counts male applicants with a positive label.
\begin{align*}
\textsc{SPD} = \frac{\f{count}(G=M \land Y=1)}{\f{count}(G=M)} \nonumber -  \frac{\f{count}(G=F \land Y=1)}{\f{count}(G=F)} \nonumber
\end{align*}
The employer would like to ensure that the \gls{SPD} between male and female is below 0.2.
The model generating the \texttt{OfferInterview} attribute is trusted by the company, but is provided by an external service and, thus, cannot be fine-tuned to improve fairness. However, the employer is willing to change their prescreening criteria by expressing their fairness requirement as an aggregate constraint $\textsc{SPD} \leq 0.2$ as long as the same criteria are applied to judge every applicant to ensure individual fairness. That is, the employer desires a repair of the query whose selection conditions are used to filter applicants.
\ifnottechreport{ We present additional motivating examples beyond fairness in~\cite{AG25techrepo}.}


\end{example}

\begin{figure*}[htbp]
    \centering
    \includegraphics[width=0.9\textwidth]{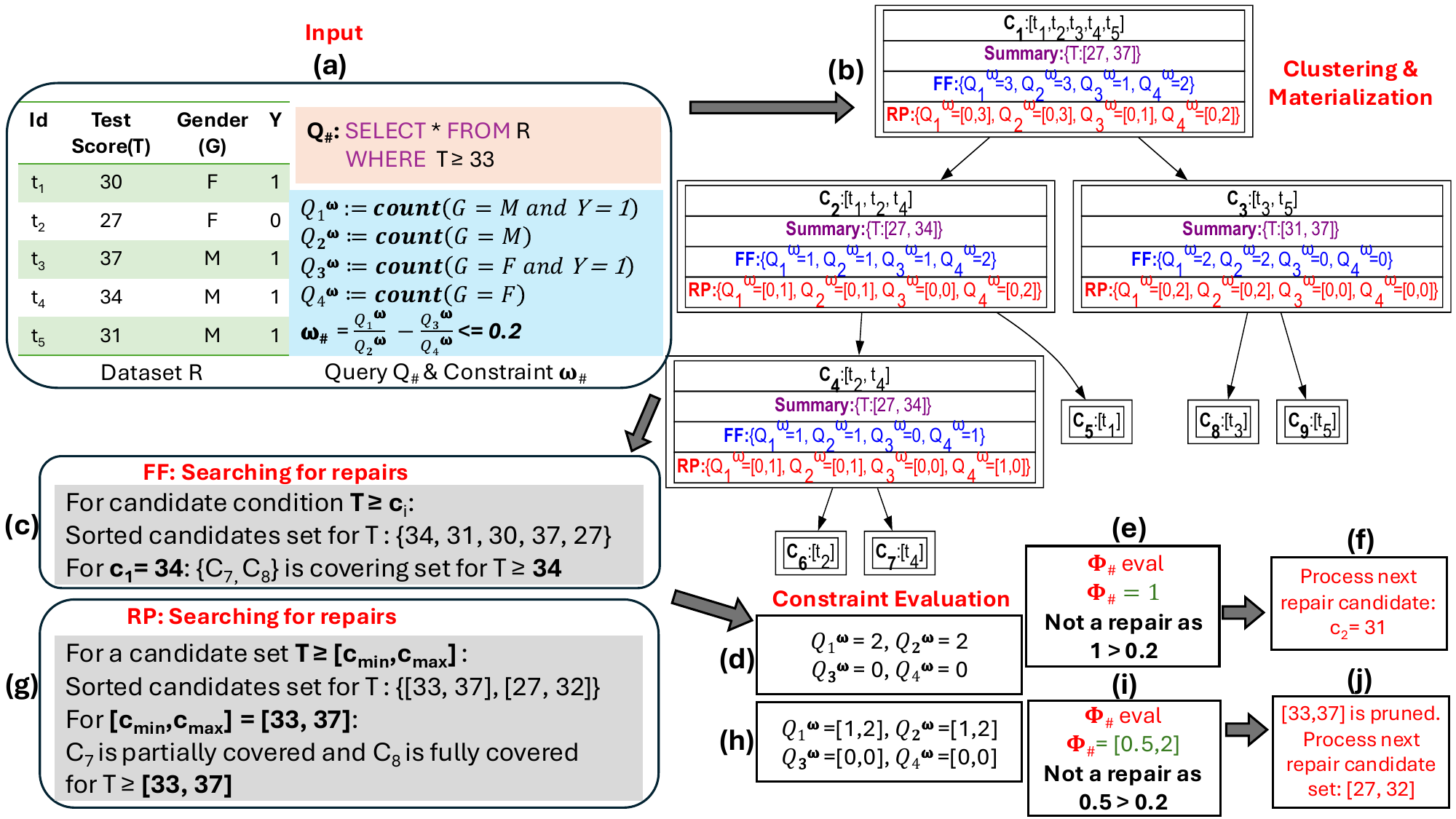} 
    \vspace{-2mm} 
    \caption{Overview of query repair with aggregate constraints using range-based pruning.}
    \label{fig:flowchart}
\end{figure*}

In this work, we model constraints on the query result as arithmetic expressions involving aggregate queries evaluated over the output of a \emph{user query}. When the result of the user query fails to adhere to such an \emph{\gls{AC}}, we would like the system to fix the violation by \emph{repairing} the query by adjusting its selection conditions, similar to~\cite{kalashnikov2018fastqre, MK09}. Specifically, we are interested in computing the top-$k$ repairs with respect to their distance to the user query. The rationale is that we would like to preserve the original semantics of the user's query as much as possible. Moreover, instead of assuming a single best repair, we consider returning $k$ repairs ranked by their distance to the original query to allow users to choose the one that best matches their intent.
\glspl{AC} significantly generalize the cardinality constraints supported in prior work on query repair for fairness~\cite{LM23,campbell24,campbell-24-r} and on query relaxation \& refinement~\cite{KL06a, MK09}. By allowing arithmetic combinations of aggregation results we support common fairness measures such as \gls{SPD} that cannot be expressed as cardinality constraints. Our work has applications beyond fairness, when applying  uniform criteria to select a set of entities subject to additional constraints, e.g., a government agency has to solicit contractors, 20\% of which should be local (See \cite{AG25techrepo} for a detailed example).
  New challenges arise from the generality of \glspl{AC} as \glspl{AC} are typically not monotone, invalidating most optimizations proposed in related work.

\ifCompileTechReport
Repairing queries based on aggregate constraints has additional applications beyond fairness including solving the \emph{empty answer} problem where the conditions of a query that returns an empty result should be relaxed to return an answer~\cite{mottin-16-hpapemanp, MM13d} and the converse where a query that returns to many answers should be refined~\cite{VW97}. This is the case, both types of use cases can be expressed as aggregate constraints: $\f{count}() > 0$ for the empty answer problem and $\f{count}() < c$ for the query refinement problem. Other types of cardinality constraints can also be expressed naturally as aggregate constraints. In addition to SPD, other common fairness metrics can also be expressed as arithmetic combinations of aggregate results. Furthermore, the expressive power gained by allowing arithmetic operations is not limited to fairness as illustrated in the following example.

\begin{example}
Consider a government agency responsible for contracting out supply of school meals to vendors. Due to anti-corruption legislation, the agency has to decide on a fixed set of criteria based on which vendors are prefiltered and can only hire vendors that fulfill the criteria. Such criteria can be naturally expressed as a query over a table storing vendors. Factors that could be used considered include the vendor's price per meal, the calories of the food they supply, their distance in miles to the school district, since when the vendor has been in business, and others. Based on these attributes, the agency may determine a set of criteria as shown below (low enough price, sufficient amount of calories, withing 100 miles of the school, and the vendor should have been in business for at least 5 years).

\begin{lstlisting}[language=SQL]
@\small\upshape\bf\ttfamily{Q3}@: SELECT * FROM vendor
    WHERE price < 20
      AND calories > 500
      AND distance < 100
      AND established < 2020
      ...
\end{lstlisting}

Using query repairs with aggregate constraints, the agency can refine their criteria to fulfill additional requirements such as:
  \begin{itemize}
  \item \textbf{Cardinality constraints}: The agency needs to ensure that they have a sufficiently large pool of vendors that fulfill their criteria to choose from to ensure all schools are covered as not all contract negotiations with vendors will be successful. Furthermore, as the agency has to justify which vendors that match their criteria they select, returning too many vendors is also undesirable. This can be expressed as a conjunction of two aggregate constraints:
    \[
      \propconstr_{numvendors} \defas \f{count}() \geq k \land \f{count}() < l
    \]

for some $k < l$.
  \item \textbf{Reducing carbon footprint}: A new regulation is proposed to reduce the carbon footprint of supplying meals requiring that at least 30\% of vendors have to be within a distance of 15 miles of the school they supply. This can be expressed as an aggregate constraint:
    \[
      \propconstr_{carbon} \defas \frac{\f{count}(distance < 15)}{\f{count}()} \geq 0.3
    \]
  \item \textbf{Budget constraint}: Assume that the database used by the agency has a certain total number of meals $M$ that need to be supplied and each vendor has a maximum capacity of meals they can supply. To ensure that vendors are selected such that at least the required capacity is met while ensuring that the total cost is below the budget $B$, the agency can use the following constraint:
    \[
      \propconstr_{budget} \defas \f{sum}(capacity) \geq  M \land \f{sum}(price) \leq B
    \]
  \end{itemize}

\end{example}

  \begin{example}[Company Product Management]\label{subsec:generalUse}


A retail company aims to support inventory planning by retrieving data on parts of type “Large Brushed” with a size greater than 10 that are supplied by suppliers located in Europe.
The company uses the following query to retrieve this information:

\begin{lstlisting}
@\small\upshape\bf\ttfamily{Q2}@: SELECT *
    FROM part, supplier, partsupp, nation, region
    WHERE p_partkey = ps_partkey AND
        s_suppkey = ps_suppkey AND p_size >= 10
        AND s_nationkey = n_nationkey
        AND n_regionkey = r_regionkey
        AND p_type = 'LARGE BRUSHED'
        AND r_name = 'EUROPE'
\end{lstlisting}

\mypar{Aggregate Constraint}
In order to minimize the impact of supply change disruption, the company wants only a certain amount of expected revenue to be from countries with import/export issues.
The constraint requires that products from UK contribute less than 10\% of the total revenue of the result set in order to minimize supply chain disruptions. Formally, the constraint is defined as follows:
\[
\frac{\sum \text{Revenue}_{\text{ProductsSelectedFromUK}}}{\sum \text{Revenue}_{\text{Selected Products}}} \leq 0.1
\]


Prior work on query repair ~\cite{AlbarrakS17} only supports constraints on a single aggregation result while the constraint shown above is an arithmetic combination of aggregation results as supported in our framework. 
\end{example}
\fi



A brute force approach for solving the query repair problem is to enumerate all possible candidate repairs in order of their distance to the user query. 
Each candidate is evaluated by running the modified query and checking whether it fulfills the aggregate constraint. The algorithm terminates once $k$ have been found. The main problem with this approach is that the number of repair candidates is exponential in the number of predicates in the user query. Furthermore,  for each repair candidate we have to evaluate the modified user query and one or more aggregate queries on top of its result. Given that the repair problem is NP-hard we cannot hope to avoid this cost in general.

\vspace{-2mm}
\mypar{Reusing aggregation results}
Nonetheless, we identify two opportunities for optimizing this process. When two repair candidates are similar (in terms of the constants they use in selection conditions), then typically there will be overlap between the aggregate constraint computations for the two candidates.
To exploit this observation, we use a \tree~\cite{bentley-75-mbstusass} to partition the input dataset. For each cluster (node in the \tree)
we materialize the result of evaluating the aggregation functions needed for a constraint on the set of tuples contained in the cluster as well as store bounds for the attribute values within the cluster (as is done in, e.g.,  zonemaps~\cite{M98a, ZW17}).
Then to calculate the result of an aggregation function for a repair candidate, we use the bounds for each cluster to determine whether all tuples in the cluster fulfill the selection conditions of the repair candidate (in this case the materialized aggregates for the cluster will be added to the result), none of the tuples in the cluster fulfill the condition (in this case the whole cluster will be skipped), or if some of the tuples in the cluster fulfill the condition (in this case we apply the same test to the children of the cluster in the \tree).
We refer to this approach as \emph{\gls{algff}}.
In contrast to the brute force approach \gls{algff} reuses aggregation results materialized for clusters. 
Continuing with~\Cref{ex:motivating-example},
consider the \tree in~\Cref{fig:flowchart}(b) which partitions the input dataset $R$ in~\Cref{fig:flowchart}(a) into a set of clusters. Here, we simplify \( Q_1 \) from~\Cref{ex:motivating-example} by considering only a single condition, $TestScore (T) \geq 33$, but use the same aggregate constraint \( \propconstr_\# \). Consider cluster \( \cluster_2 \) in~\Cref{fig:flowchart}(b), where the values of attribute \textit{TestScore} (\( T \)) are bounded by \([27, 34]\). For a repair candidate with a condition $T \geq 37$ the entire cluster can be skipped as no tuples in $\cluster_2$ can fulfill the condition. In contrast, for condition \( T \geq 30 \), all tuples in \( \cluster_3 \) satisfy the condition, since the values of attribute \( T \)  are bounded by \([31, 37]\) (see~\Cref{fig:flowchart}(b)).
\mypar{Evaluating multiple candidate repairs at once}
We 
extend this idea to bound the aggregation constraint result for sets of repair candidates at once. We refer to this approach as \emph{\gls{algrp}}. A set of repair candidates is encoded as intervals of values for the constant $\constant_i$ of each predicate $\attr_i \op \constant_i$ of the user query, e.g., $\constant_1 \in [33,37]$
as shown in ~\Cref{fig:flowchart}(g). 
We again reason about whether all / none of the tuples in a cluster fulfill the condition for \emph{every} repair candidate from the set.
The result are valid bounds on the aggregation constraint result for any candidate repair within the candidate set. Using these bounds we validate or disqualify complete candidate sets at once.

We make the following contributions in this work:
\begin{itemize}
    \item A formal definition of query repair under constraints involving arithmetic combinations of aggregate functions in \Cref{sec:problem Overview}. 
    \item We present an optimized algorithm for the \gls{acrp} problem that  reuses aggregation results when evaluating repair candidates (\Cref{sec:baseline}) and evaluates multiple repair candidates by exploiting sound bounds that hold for all repair candidates in a set (\Cref{sec:Range}).
     \item A comprehensive experimental evaluation over multiple datasets, queries and constraints in \Cref{sec:expermint}. Compared to the state of the art~\cite{LM23}, we cover significantly more complex constraints. 
\end{itemize}




\section{Problem Definition} \label{sec:problem Overview}


We consider a dataset $\db = { R_1, \cdots, R_\numrelations}$ 
consisting of one or more relations $R_i$, an input query $\query$ that should be repaired, and an \glsfmtlong{AC}. The goal is to find the $k$ queries that fulfill the constraint and minimally differ from \query. 

\mypar{User Query} \label{su.bsubsec: user query}
A user query \query is an \gls{SPJ} query, i.e., a relational algebra expression of the form:
\[
  \projection_{A}(\selection_{\theta}(R_1 \join \ldots \join R_\numtables))
\]
We assume that the selection predicate $\theta$ of such a query is a conjunction $
\theta = \theta_1 \land \ldots \land \theta_{\numcond}
$ of comparisons of the form $\attr_i \op_i \constant_i$. For numerical attributes $\attr_i$, we allow $\op_i \in \{ <, >, \leq, \geq, =, \neq\}$ and for categorical attributes $\attr_i$ we only allow $\op_i \in \{=, \neq \}$.
We use $\query(\db)$ to denote the result of evaluating $\query$
over \db.

\mypar{Aggregate Constraints (\gls{AC})} \label{subsubsec: arithmatic constraints}
The user specifies requirements on the result of their query as an \gls{AC}. An \gls{AC} is a comparison between a threshold and an arithmetic expression over the result of filter-aggregation queries.
Such queries are of the form $\aggregation_{\aggf(a)}(\selection_{\theta}(\query(\db))$ where $\aggf$ is an aggregate function -- one of $\f{count}, \f{sum}, \f{min}, \f{max}, \f{avg}$ -- and $\theta$ is a selection condition. We use $\aaggq$ to denote such a filter-aggregation query. 
These queries are evaluated over the user query's result $\query(\db)$.
An aggregate constraint $\propconstr$ is of the form:
\[
\propconstr \defas \propthresh\, 
\op\,\, \expr(\aggq{1}, \ldots, \aggq{\numaggs}).
\]
Here, $\expr$ is an arithmetic expression using operators $(+,-,\ast, /)$ over $\{\aggq{i}\}$, $op$ is a comparison operator, and \propthresh is a threshold.
%
Aggregate constraints are non-monotone in general due to (i) non-monotone arithmetic operators like division, (ii) non-monotone aggregation function, e.g., $\f{sum}$ over the integers $\mathbb{Z}$, and (iii) combination of monotonically increasing and  decreasing aggregation functions, e.g., 
$\f{max}(A) + \f{min}(B)$.

\mypar{Query Repair} \label{subsubsec: query refinement}
%
Given a user query $\query$, database $\db$, and constraint $\propconstr$ that is violated on $\query(\db)$, we want to generate a repaired version $\Refquery$ of $\query$ such that $\Refquery(\db)$ fulfills $\propconstr$. We restrict repairs to changes of the selection condition $\theta$ of 
$\query$. For ease of presentation, we consider a single \gls{AC}, but our algorithms can also handle a conjunction of multiple \glspl{AC}, e.g., the cardinality for one group should be above a threshold $\propthresh_1$ and for another group below a threshold $\propthresh_2$.
Given the user query  $\query$ with a condition
$
\theta = \defas \bigwedge_{i=1}^{\numcond} \attr_i \op_i {\constant_i}
$, a \emph{repair candidate} is a query $\Refquery$ that differs from $\query$ only in the constants used in selection conditions, i.e., $\Refquery$ uses a condition:
$
\theta' \defas \bigwedge_{i=1}^{\numcond} \attr_i \op_i {\constant_i}'$.
For convenience, we will often use the vector of constants $\vec{c} = [\constant_1', \ldots, \constant_{\numcond}']$ to denote a repair candidate
and use $\repairspace{\query}$ to denote the set of all candidates.
A candidate is a \emph{repair} if
$\Refquery(\db) \models \propconstr$.

\mypar{Repair Distance} \label{subsubsec: refinment distance}
Ideally, we would want to achieve a repair that minimizes the changes to the user's original query to preserve the intent of the user's query as much as possible. 
We measure similarity using a linear weighted combination of distances between the constants used in selection conditions of the user query and the repair, similar to~\cite{campbell24, KL06a}.\footnote{
    We discuss other possible optimality criteria used in prior work in \Cref{sec:relwork}. Other options include returning all repairs that are pareto optimal regarding predicate-level distances or to minimize the change to the query's result. Our algorithms can be extended to optimize for any distance metric which can be interval-bounded based on bounds for attribute values of a set of tuples.}
Consider the user query $\query$ with selection condition $\theta_1 \land \ldots \land \theta_{\numcond}$ and repair $\Refquery$ with selection condition ${\theta_1}' \land \ldots \land {\theta_{\numcond}}'$. Then the distance $\qdist(\query,\Refquery)$ is defined as:
\vspace{-2.1mm}
\[
\qdist(\query,\Refquery) = \sum_{i=1}^{\numcond} \dweight_i \cdot \qdist(\theta_i,{\theta_i}')
\]
where $\dweight_i$ is a weight in $[0,1]$ such that $\sum_i \dweight_i = 1$ and the distance between two predicates $\theta_i = \attr_i \op_i \constant_i$ and ${\theta_i}' = \attr_i \op_i {\constant_i}'$ for numeric attributes $\attr_i$ is: $\frac{|{\constant_i}' - \constant_i|}{|\constant_i|}$.
For categorical attributes, the distance is $1$ if $\constant_i \neq {\constant_i}'$ and $0$ otherwise.
For example, for~\Cref{ex:motivating-example}, the repair candidate with conditions \texttt{Major} = EE, $\texttt{Testscore}\geq 33$, and $\texttt{GPA} \geq 3.9$ has a distance of $1+ \frac{\f33-33}{\f33} + \frac{\f3.9-3.8}{\f3.8} = 1.026$.

We are now ready to formulate the problem studied in this work, computing the $k$ repairs with the smallest distance to the user query.
Among these $k$ repairs, the user can then select the repair that best aligns with their preferences. Here $\topk_{x \in X} f(x)$ returns the $k$ elements from set $X$ with the smallest $f(x)$ values.

\smallskip
\noindent\begin{tcolorbox}[rounded corners,left=1mm,top=0mm,bottom=-0.5mm, width=0.48\textwidth]
\textsc{\Gls{acrp}}:
  \begin{itemize}
    \item \textbf{Input}: user query $\query$, database $\db$, constraint $\propconstr$, threshold $k$
    \item \textbf{Output}:
      \[
        \topk_{\Refquery \in \repairspace{\query}:\,\,\,\, \Refquery(\db)\, \models\, \propconstr} \qdist(\query, \Refquery)
        \]

    \end{itemize}
\end{tcolorbox}


\mypar{Hardness}\label{Search Space}
To generate a repair $\Refquery$ of $\query$,
we must explore the combinatorially large search space of possible candidate repairs. For a single predicate over an attribute $\attr_i$ with $\numdoma{i}$ distinct values there are $O(\numdoma{i})$ possible repairs. Thus, the size of the candidate set $\repairspace{\query}$ is in $O(\prod_{i=1}^{\numcond} \numdoma{i})$,  exponential in $\numcond$, the number of conditions in the user query. Unsurprisingly, the \gls{acrp} is NP-hard in the schema size.\ifnottechreport{ In \cite{AG25techrepo} we provide more details on the number of repairs for specific predicates.}

\iftechreport{Consider a query with a conjunction of conditions of the form $\attr_i \op_i \constant_i$ for $i \in [1,\numcond]$. Again
$\numdoma{i}$ denote the number of values in the active domain of $\attr_i$.
 Each candidate repair corresponds to choosing constants  $[\constant_1', \ldots, \constant_{\numcond}']$. 
 The  number of candidate repairs depends on which comparison operators are used, e.g., for $\leq$ there are at most $\numdoma{i} + 1$ possible values
 that lead to a different result in terms of which of the input tuples will fulfill the condition. To see why this is the case assume that the values in $\attr_i$ sorted based on $\leq$ are $a_1, \ldots, a_p$. Then for any constant $\constant$, the condition $\attr_i \leq \constant$ includes tuples with values in $\{ a_i \mid a_i \leq \constant \}$ and this filtered set of $\attr_i$ values is always a prefix of $a_1, \ldots, a_p$. Thus, there are $\numdoma{i} + 1 = p+1$ for choosing the length of this sequence ($0$ to $p$).
 The size of the search space is $O(\prod_{i=1}^{\numcond} \numdoma{i})$,  exponential in $\numcond$, the number of conditions in the user query.  
Unsurprisingly, the \gls{acrp} is NP-hard in the schema size.}  

\section{The \Acrlong{algff} Algorithm} \label{sec:baseline}
We now present \emph{\acrfull{algff}}, our first algorithm for the \gls{acrp} that materializes results of each aggregate-filter query $\aggq{i}$ for subsets of the input database $\db$ and combines these aggregation results to compute the result of $\aggq{i}$ for a repair candidate \Refquery and then use it to evaluate the \glsfmtfull{AC} $\propconstr$,for  \Refquery. \Cref{fig:flowchart} shows the example of applying this
 algorithm: (b) building a \tree 
and materializing statistics, (c) searching for candidate repairs, and (d)-(e)  evaluating constraints for repair candidates.

\subsection{Clustering and Materializing Aggregations} \label{subsec:clustering}
For ease of presentation, we consider a database consisting of a single table $\rel$ from now on. However, our approach can be generalized to queries involving joins by materializing the join output and treating it as a single table. As repairs only change the selection conditions of the user query, there is no need to reevaluate joins when checking repairs.
We use a \tree to partition $\rel$ into subsets (\emph{clusters}) based on attributes that appear in the selection condition ($\theta$) of the user query. The rationale is that the selection conditions of a repair candidate filter data along these attributes.


To evaluate the \gls{AC} \propconstr for a candidate $\Refquery=[\constant_1', \ldots, \constant_{\numcond}']$, we determine a set of clusters (nodes in the \tree) that cover exactly the subset of \db that fulfills the selection condition of the candidate. We can then merge the materialized aggregation results for these clusters to compute the
results of the aggregation queries $\aggq{i}$ used in \propconstr for $\Refquery(\db)$.
To do that, we record the following information for each cluster $\cluster \subseteq \db$ that can be computed by a single scan over the tuples in the cluster, or by combining results from previously generated clusters if we generate clusters bottom up.


\begin{itemize}[leftmargin=1em]
\item \textbf{Selection attribute bounds:} For each attribute $\attr_i$ used in the condition $\theta$, we store $\boundsa{\attr_i} \defas [\min(\pi_{\attr_i}(\cluster)), \max(\pi_{\attr_i}(\cluster))]$.
\item \textbf{Count}: The total number of tuples $\ccount(\cluster) \defas \card{\cluster}$ in the cluster.
\item \textbf{Aggregation results}: For each filter-aggregation query $\aaggq$ in constraint $\propconstr$, we store $\aaggq(\cluster)$.
\end{itemize}

An example \tree is shown in~\Cref{fig:flowchart}(b).
The user query filters on attribute $TestScore (T)$. The root of the \tree represents the full dataset. At each level, the clusters from the previous level are split into a number of sub-clusters 
(this is a configuration parameter \branchfactor called the branching factor), two in the example, along one of the attributes in $\theta$. For instance, the root cluster $\cluster_1$ is split into two clusters $\cluster_2$ and $\cluster_3$ by partitioning the rows in $\cluster_1$ based on their values in attribute $T$. For cluster $\cluster_2$ containing three tuples $t_1$, $t_2$, and $t_4$, we have $\boundsa{T} = [27,34]$ as the lowest $T$ value is $27$ (from tuple $t_2$) and the highest value is $34$ (tuple $t_4$). The value of $\aggq{2} = \f{count}(Gender(G) = M)$ for $\cluster_2$ is $1$ as there is one male in the cluster. Consider a repair candidate with the condition \( T \geq 37 \). Based on the bounds \( \boundsa{T} = [27, 34] \), we know that none of the 
tuples satisfy this condition.
Thus, this cluster and the whole subtree rooted at the cluster can be ignored for computing the \gls{AC} $\propconstr_{\#}$ for the candidate.

For ease of presentation we assume that the leaf nodes of the \tree contain a single tuple each. As this would lead to very large trees, in our implementation we do not further divide clusters $\cluster$ if that contain less tuples than a threshold \bucketsize (i.e. $\card{\cluster} \leq \bucketsize$). We refer to this parameter as the \emph{bucket size}.

\newcommand{\fullcovered}{in}
\newcommand{\notcovered}{notin}
\newcommand{\partialcovered}{pin}

\begin{algorithm}[t]
    \small
  \caption{\fcall{\fullcoveringclusterset}}
  \label{alg:filter_fully}
  \begin{algorithmic}[1]
    \Require \tree with root $\cluster_{root}$, condition $\theta' = \theta_1' \land \ldots \land \theta_\numcond'$, relation $\rel$.
    \Ensure Set of clusters $\clusterset$ such that $\bigcup_{\cluster \in \clusterset} C = \selection_{\theta'}(\rel)$.
    \State $stack \gets [\cluster_{root}]$
    \State $\clusterset \gets \emptyset$ \Comment{Initialize result set}

    \While{$stack \neq \emptyset$ }
      \State $\cluster_{cur} \gets \f{pop}(stack)$1
      \State $\fullcovered \gets \btrue, \notcovered \gets \bfalse$
      \ForAll{$\theta_i' = (\attr_i \op_i \constant_i') \in \theta'$}
        \State $\fullcovered \gets \fullcovered \land \evalall(\theta_i', \boundsa{\attr_i}(\cluster_{cur}))$ \Comment{All tuples fulfill $\theta_i'$?}
        \State $\notcovered \gets \notcovered \lor \evalall(\neg \theta_i', \boundsa{\attr_i}(\cluster_{cur}))$
      \EndFor

      \If{$\fullcovered$} \Comment{All tuple in $\cluster$ fulfill $\theta'$}
        \State $\clusterset \gets \clusterset \cup \{\cluster_{cur} \}$
      \ElsIf{$\neg\,\notcovered$} \Comment{Some tuples in $\cluster$ may fulfill $\theta'$}
        \ForAll{$\cluster \in \f{children}(\cluster_{cur})$} \Comment{Process children}
          \State $stack \gets stack \cup \{\cluster\}$
        \EndFor
      \EndIf
    \EndWhile
    \State \Return $\clusterset$
  \end{algorithmic}
\end{algorithm}






\subsection{Constraint Evaluation for Candidates} \label{subsec: baseline}
The \gls{algff} algorithm (\Cref{alg:filter_fully}) takes as input the condition $\theta'$ of a repair candidate,
the root node of the \tree $\cluster_{root}$, and returns a set of disjoint clusters $\clusterset$ such that the union of these clusters is precisely the subset of the relation $\rel$ that fulfills $\theta'$:
\begin{align}\label{eq:covering-cluster-set}
  \bigcup_{\cluster \in \clusterset} = \selection_{\theta'}(\rel)
\end{align}
The statistics materialized for this cluster set $\clusterset$ are then used to evaluate the \gls{AC} for the repair candidate.

\subsubsection{Determining a Covering Set of Clusters}
\label{sec:determ-cover-set}
%
The algorithm maintains a $stack$ of clusters to be examined that is initialized with the root cluster $\cluster_{root}$ (line 1). It then processes one cluster at a time until a set of clusters $\clusterset$ fulfilling~\Cref{eq:covering-cluster-set} has been determined (lines 3-14).
For each cluster $\cluster$, we distinguish 3 cases (lines 6-8):
(i) we can use the bounds on the selection attributes recorded for the cluster to show that all tuples in the cluster fulfill $\theta'$, i.e., $\selection_{\theta'}(\cluster) = \cluster$ (line 7). In this case, the cluster will be added to $\clusterset$ (lines 9-10); (ii) based on the bounds, we can determine that none of the tuples in the cluster fulfill the condition (line 8). Then this cluster can be ignored; (iii) either a non-empty subset of $\cluster$ 
fulfills $\theta'$ or based on the bounds $\boundsa{\attr_i}(\cluster)$ we cannot demonstrate that $\selection_{\theta'}(\cluster) = \emptyset$ or $\selection_{\theta'}(\cluster) = \cluster$ hold. In this case, we add the children of $\cluster$ to the stack to be evaluated in future iterations (lines 11-13). The algorithm uses the function \evalall shown in \Cref{tab:clusterOperators} to determine based on the bounds of the cluster $\cluster$, the comparison condition $\theta'_i$ is guaranteed to be true for all $t \in \cluster$. Additionally, it checks whether case (ii) holds by applying $\evalall$ to the negation $\theta'_i$. Note that to negate a comparison we simply push the negation to the comparison operator, e.g., $\neg (\attr < \constant) = (\attr \geq \constant)$. As the selection condition of any repair candidate is a conjunction of comparisons $\theta'_1 \land \ldots \land \theta'_\numcond$, the cluster is \emph{fully covered} (case (i)) if $\evalall$ returns true for all $\theta'_i$ and \emph{not covered at all} (case (ii)) if $\evalall$ returns true for at least one comparison $\neg \theta'_i$.


\begin{table}[t]
    \centering
    \caption{Given the bounds
       $[\lb{\attr},\ub{\attr}]$ for the attribute $\attr$ of
       a condition $\attr \op \constant$ or $\attr \in [\constant_1, \constant_2]$,
       function \evalall does return true if the condition evaluates to true for all values in $[\lb{\attr},\ub{\attr}]$.
       For \gls{algrp}, we 
      consider a range $[\lb{\constant},\ub{\constant}]$ (corresponding to a set
      of candidates) or two ranges $[\lb{\constant_1}, \ub{\constant_1}]$ and $[\lb{\constant_2}, \ub{\constant_2}]$ for operator $\in$.
      \evalallrange determines whether for every
      $\constant \in [\lb{\constant},\ub{\constant}]$, the condition is
      guaranteed to evaluate to true for every $\attr \in [\lb{\attr},\ub{\attr}]$ while
      \evalsomerange determines whether for some
      $\constant \in [\lb{\constant},\ub{\constant}]$, the condition may evaluate to true for
      $\attr \in [\lb{\attr},\ub{\attr}]$. 
    }
    \vspace{-1.5mm}
\resizebox{1\linewidth}{!}{
  \begin{tabular}{|c|c|c|c|c|c|}
        \hline
        \textbf{Op.}                  & \evalall                                                          & \evalallrange                                                               & \evalsomerange                                                                      \\ \hline
        $>, \geq$                     & $\lb{\attr} > \constant$,\hspace{3mm} $\lb{\attr} \geq \constant$ & $\lb{\attr} > \ub{\constant}$,\hspace{3mm} $\lb{\attr} \geq \ub{\constant}$ & $\ub{\attr} > \lb{\constant}$,\hspace{3mm} $\ub{\attr} \geq \lb{\constant}$         \\ \hline
        $<, \leq$                     & $\ub{\attr} < \constant$,\hspace{3mm} $\ub{\attr} \leq \constant$ & $\ub{\attr} < \lb{\constant}$,\hspace{3mm} $\ub{\attr} \leq \lb{\constant}$ & $\lb{\attr} < \ub{\constant}$,\hspace{3mm} $\lb{\attr} \leq \ub{\constant}$         \\ \hline
        $=$                           & $\lb{\attr} = \ub{\attr} = \constant$                             & $\lb{\attr} = \lb{\constant} = \ub{\attr} = \ub{\constant}$                 & $[\lb{\attr},\ub{\attr}] \cap [\lb{\constant},\ub{\constant}] \neq \emptyset$       \\ \hline
        $\neq$                        & $\constant \not\in [\lb{\attr},\ub{\attr}]$                       & $[\lb{\attr},\ub{\attr}] \cap [\lb{\constant}, \ub{\constant}] = \emptyset$ & $\neg(\lb{\attr} = \lb{\constant} = \ub{\constant} = \ub{\attr})$                   \\ \hline
      $\in [\constant_1,\constant_2]$ & $\constant_1 \leq \lb{\attr} \land \ub{\attr} \leq \constant_2$   & $\ub{\constant_1} \leq \lb{\attr} \land \ub{\attr} \leq \lb{\constant_2}$   & $[\lb{\attr}, \ub{\attr}] \cap [\lb{\constant_1}, \ub{\constant_2}] \neq \emptyset$ \\ \hline
    \end{tabular}
}
    \label{tab:clusterOperators}
\end{table}
\subsubsection{Determining Coverage}\label{sec:cluster-coverage}
In \Cref{tab:clusterOperators}, we define the function \evalall which takes a condition $\attr \op \constant$ and bounds $\boundsa{\attr}(\cluster)$ for attribute $\attr$ in cluster $\cluster$ and returns true if it is guaranteed that all tuples $t \in \cluster$ fulfill the condition. Ignore \evalallrange for now, this function will be used in \Cref{sec:Range}. An inequality $>$ (or $\geq$) is true for all tuples if the lower bound
$\lb{\attr}$ of $\attr$ is larger (larger equal) than the threshold $\constant$. The case for $<$ and $\leq$ is symmetric: the upper bound $\ub{\attr}$ has to be smaller (smaller equals) than $\constant$. For an equality, 
we can only guarantee that the condition is true if $\lb{\attr} = \ub{\attr} = \constant$. For $\neq$, all tuples fulfill the inequality if $\constant$ does not belong to the interval  $[ \lb{\attr}, \ub{\attr} ]$.

For the running example in \Cref{fig:flowchart}, consider a repair candidate with the condition \( T \geq 34 \), where \( \constant_1 = 34 \). The algorithm maintains a stack of clusters initialized to  \([\cluster_1]\), the root node of the \tree. In each iteration it takes on cluster form the stack. The root cluster \( \cluster_1 \),  has \( \boundsa{T}(\cluster_1) = [27, 37] \). The algorithm evaluates whether all or none of the tuples satisfy the condition. Since it neither is the case,
we proceed to the children of $\cluster_1$: \( \cluster_2 \) and \( \cluster_3 \). The same situation occurs for \( \cluster_2 \) and \( \cluster_3 \) leading to further exploration of their child \( \{\cluster_4 \) and \( \cluster_5 \}\) for \( \cluster_2 \) and \(\{ \cluster_8 \) and \( \cluster_9 \}\) for \( \cluster_3 \).
Since the coverage for \( \cluster_4 \) cannot be determined, the algorithm proceeds to process \( \cluster_6 \) and \( \cluster_7 \).
Clusters \( \cluster_5 \), \( \cluster_6 \) and \( \cluster_9 \) are determined to not satisfy the condition 
while \( \cluster_7 \) and \( \cluster_8 \) are confirmed to meet the condition and are
added to \( \clusterset \). In this example, we had to explore all of the leaf clusters, but often we will be able to prune or confirm clusters covering multiple tuples. For instance, for \(T \geq 37\), $\cluster_2$ with bounds $[27,34]$ with all of its descendents can be skipped as $T \geq 37$ is false for any $T \in [27,34]$.

\subsubsection{Constraint Evaluation}\label{sec:clusterset-constraint-eval}
After identifying the covering set of clusters \clusterset for a repair candidate \Refquery, our approach evaluates the \gls{AC} $\propconstr$ over \clusterset.
Recall that for each cluster $\cluster$ we materialize the result of each filter aggregate query $\aggq{i}$ used in $\propconstr$. For aggregate function $\f{avg}$ that is not decomposable, we apply the standard approach of storing $\f{count}$ and $\f{sum}$ instead. We then compute $\aggq{i}(\query(\db))$ over the materialized aggregation results for the clusters. Concretely, for such an aggregate query $\aaggq \defas \aggregation_{\f{f}(\attr)}(\selection_{\theta'}(\query(\db))$ we compute its result as follows using \clusterset:
$
  \aggregation_{\f{f'}(\attr)}\left( \bigcup_{\cluster \in \clusterset} \{\aaggq(\cluster)\} \right).
$
Here $\f{f'}$ is the function we use to merge aggregation results for multiple subsets of the database. This function depends on $\f{f}$, e.g., for both $\f{count}$ and $\f{sum}$ we have $\f{f'} = \f{sum}$, for $\f{min}$ we use $\f{f'} = \f{min}$, and for  $\f{max}$ we use $\f{f'} = \f{max}$.
We then substitute these aggregation results into \propconstr and evaluate the resulting expression to determine whether \Refquery fulfills the constraints and is a repair or not.

In the example from \Cref{fig:flowchart}(c), the covering set of clusters for the repair candidate with $\constant_1 = 34$ is \( \clusterset = \{\cluster_7, \cluster_8\} \). Evaluating \( \aggq{1} = \f{count}(Gender(G) = M \land Y = 1) \) over \( \clusterset \), we sum the counts:
$\aggq{1} = \aggq{1\cluster_7} + \aggq{1\cluster_8} = 1+1=2$. Similarly,
$
\aggq{2} = \aggq{2\cluster_7} + \aggq{2\cluster_8} = 1+1=2,
$
$
\aggq{3} = \aggq{3\cluster_7} + \aggq{3\cluster_8} = 0+0=0, \quad
\aggq{4} = \aggq{4\cluster_7} + \aggq{4\cluster8} = 0+0=0
$ as shown in \Cref{fig:flowchart}(d).
Substituting these values into $\propconstr_{\#}$, we obtain $1 \leq 0.2 = \bfalse$ as shown in \Cref{fig:flowchart}(e). Since the candidate \( T \geq 34 \) does not satisfy the constraint it is not a valid repair.

\subsection{Computing Top-$k$ Repairs} \label{subsubsec: top-k}

To compute the top-$k$ repairs, we enumerate all repair candidates in increasing order of their distance to the user query using the distance measure from \Cref{sec:problem Overview}. For each candidate \Refquery we apply the \gls{algff} to determine a covering clusterset, evaluate the constraint \propconstr, and output \Refquery if it fulfills the constraint. Once we have found $k$ results, the algorithm terminates.

\begin{algorithm}[t]
\caption{\fcall{\parcoveringclusterset}}
\label{alg:optimized_filter_ranges}
\begin{algorithmic}[1]
  \Require \tree with root $\cluster_{root}$, repair candidate set $\Repqueryset =  [ [\lb{\constant_1}, \ub{\constant_1}], \ldots, [\lb{\constant_\numcond}, \ub{\constant_\numcond}]]$, condition $\theta$
  \Ensure Partially covering cluster set $(\fullclusterset, \parclusterset)$
  \State $stack \gets [\cluster_{root}]$
  \State $\fullclusterset \gets \emptyset, \parclusterset \gets \emptyset$ \Comment{Initialize cluster sets}
  \While{$stack \neq \emptyset$}
    \State $\cluster_{cur} \gets \f{pop}(stack)$
    \State $\fullcovered \gets \btrue, \partialcovered \gets \btrue$
    \ForAll{$\theta_i = (\attr_i \op_i \constant_i) \in \theta$} \Comment{$\cluster_{cur}$ fully / part. covered?}
      \State $\fullcovered \gets \fullcovered \land \evalallrange(\theta_i, [\lb{\constant_i}, \ub{\constant_i}], \boundsa{\attr_i}(\cluster_{cur}))$
      \State $\partialcovered \gets \partialcovered \land \evalsomerange(\theta_i, [\lb{\constant_i}, \ub{\constant_i}], \boundsa{\attr_i}(\cluster_{cur}))$
    \EndFor
    \If{$\fullcovered$} \Comment{Add fully covered cluster to the result}
      \State $\fullclusterset \gets \fullclusterset \cup \{ \cluster_{cur} \}$
    \ElsIf{$\partialcovered$}
      \If{$\f{isleaf}(\cluster_{cur})$} \Comment{Partially covered leaf cluster}
        \State $\parclusterset \gets \parclusterset \cup \{ \cluster_{cur} \}$
      \Else \Comment{Process children of partial cluster}
        \ForAll{$\cluster \in \f{children}(\cluster_{cur})$}
          \State $stack \gets stack \cup \{\cluster\}$
        \EndFor
      \EndIf
    \EndIf
  \EndWhile
  \State \Return $(\fullclusterset, \parclusterset)$
\end{algorithmic}
\end{algorithm}

\section{\Acrlong{algrp} (\gls{algrp})}
\label{sec:Range}

While algorithm \gls{algff} reduces the effort needed to evaluate aggregation constraints for repair candidates, it has the drawback that we still have to evaluate each repair candidate individually. We now present an enhanced approach that reasons about sets of repair candidates. For a user query condition $\theta_1 \land \ldots \land \theta_\numcond$ where $\theta_i \defas \attr_i \op_i \constant_i$, we use ranges of constant values instead of constants to represent such a set of repairs $\Repqueryset$:
$
  [ [\lb{\constant_1}, \ub{\constant_1}], \ldots, [\lb{\constant_\numcond}, \ub{\constant_\numcond}] ].
$
  Such a list of ranges $\Repqueryset$ represents a set of a repair candidates:
  \[
    \{ [\constant_1, \ldots, \constant_\numcond] \mid \forall i \in [1,m]: \constant_i \in [\lb{\constant_i}, \ub{\constant_i}] \}
  \]
Consider an aggregation constraint $\propconstr \defas \propthresh \op \expr(\aggq{1}, \ldots, \aggq{n})$.
  Our enhanced approach \gls{algrp} uses a modified version of the \tree from \gls{algff} to compute conservative bounds of the arithmetic expression $\lb{\expr}$ and $\ub{\expr}$ on the possible values for $\expr$ that hold for all repair candidates in \Repqueryset. Based on such bounds, if (i) $\propthresh \op \constant$ holds for every $\constant \in [\lb{\expr},\ub{\expr}]$, then every $\Refquery \in \Repqueryset$ is a valid repair, if (ii) $\propthresh \op \constant$ is violated for every $\constant \in [\lb{\expr},\ub{\expr}]$, then no $\Refquery \in \Repqueryset$ is a valid repair and we can skip the whole set. Otherwise, (iii) there may or may not exist some candidates in \Repqueryset that are repairs. In this case, our algorithm partitions \Repqueryset into multiple subsets and applies the same test to each partition. Following, we introduce our algorithm that utilizes such repair candidate sets and bounds on the aggregate constraint results and then explain how to use the \tree to compute bounds.

\begin{algorithm}[t]
\caption{Top-k Repairs
w. Range-based Pruning of Candidates
}
\label{alg:range_filtering}
\begin{algorithmic}[1]
    \Require \tree with root $\cluster_{root}$,
    constraint \gls{AC} $\propconstr$, repair candidate set $\Repqueryset =  [ [\lb{\constant_1}, \ub{\constant_1}], \ldots, [\lb{\constant_\numcond}, \ub{\constant_\numcond}]]$, user query condition $\theta = \theta_1 \land \ldots \land \theta_\numcond$, user query $\query$
    \Ensure Top-$k$ repairs $\settopk$ 
    \State $\settopk \gets \emptyset$ \Comment{Queue of repairs $\query'$ sorted on $\qdist(\query,\query')$}
    \State $\rsettopk \gets \emptyset$ \Comment{Queue of repair sets $\Repqueryset'$ sorted on  $\lb{\qdist(\query,\Repqueryset')}$}
    \State $\pqueue \gets [\Repqueryset]$ \Comment{Queue of repair candidate sets $\Repqueryset'$ sorted on $\lb{\qdist(\query,\Repqueryset')}$}
    \While{$\pqueue \neq \emptyset$}
     \State $\rsetcur \gets \fcall{pop}(\pqueue)$
     \State $\rsetnext \gets \fcall{peek}(\pqueue)$ \Comment{Peek at next item in queue}
     \State $(\fullclusterset, \parclusterset) \gets \parcoveringclusterset(\rsetcur, \cluster_{root},\theta)$ 


     \If{$\acrangeevalall
     (\propconstr, \fullclusterset, \parclusterset)$}\label{algline:acevalall}\Comment{All $\query' \in \rsetcur$ are repairs?}
         \State $\rsettopk \gets \fcall{insert}(\rsettopk, \rsetcur)$
     \ElsIf{$\acrangeevalexists(\propconstr, \fullclusterset, \parclusterset)$} \Comment{Some  repairs?} \label{algline:acevalexists}
       \For{$\rsetnew \in \rangedivide(\rsetcur)$}  \Comment{divide ranges}
         \If{$\hascandidates(\rsetnew)$}
           \State $\pqueue \gets \fcall{insert}(\pqueue, \rsetnew)$
         \EndIf
       \EndFor
     \EndIf
     \State $\settopk \gets \topkconcretecand(\rsettopk,k)$ \Comment{Top $k$ repairs}
     \If{$\card{\settopk} \geq k$} \Comment{Have $k$ repairs?}
       \If{$\lb{d(\query,\rsetnext)} > d(\query,\settopk[k])$} \Comment{Rest inferior?
       }
         \State \textbf{break}
       \EndIf
     \EndIf
    \EndWhile
    \State \Return $\settopk$
  \end{algorithmic}
\end{algorithm}

  \subsection{Computing Top-$k$ Repairs}
  \label{sec:computing-top-k}

\gls{algrp} (\Cref{alg:range_filtering}) takes as input a \tree with root $\cluster_{root}$, a user query condition $\theta$, a \gls{AC} $\propconstr$,  a candidate set
$
\Repqueryset = [ [\lb{\constant_1}, \ub{\constant_1}], \ldots,$ $[\lb{\constant_\numcond}, \ub{\constant_\numcond}]],
$
 and user query \query and returns the set of top-$k$ repairs $\settopk$. 

The algorithm maintains three priority queues: (i) $\settopk$ is a queue of individual repairs that eventually will store the top-k repairs. This queue is sorted on $\qdist(\query, \Refquery)$ where $\Refquery$ is a repair in the queue; (ii) $\rsettopk$ is a queue where each element is a repair candidate set $\Repqueryset$ encoded as ranges as shown above. For each $\Repqueryset$ we have established that for all $\Refquery \in \Repqueryset$, $\Refquery$ is a repair. This query is sorted on the lower bound $\lb{\qdist(\query,\Refquery \in \Repqueryset)}$ of the distance of any repair in $\Repqueryset$ to the user query. Finally, (iii) $\pqueue$ is a queue where each element is a repair candidate set $\Repqueryset$. This queue is also sorted on $\lb{\qdist(\query,\Refquery \in \Repqueryset)}$. In each iteration of the main loop of the algorithm, one repair candidate set from $\pqueue$ is processed.

The algorithm initializes $\pqueue$ to the input parameter repair candidate set $\Repqueryset$. We call the algorithm with a repair candidate set that covers the whole search space
(line 1-3).
The algorithm's main loop processes one repair candidate $\rsetcur$ at a time (line 5) while keeping track of the next candidate $\rsetnext$ (line 6) until a set of top-k repairs fulfilling \gls{AC} $\propconstr$ has been determined (lines 4–17).
For the current repair candidate set $\rsetcur$, we use function \parcoveringclusterset (\Cref{alg:optimized_filter_ranges}) to determine two sets of clusters $\fullclusterset$ and $\parclusterset$ (line 7). For every cluster  $\cluster \in \fullclusterset$, all tuples in $\cluster$ fulfill the condition of every repair candidate $\Refquery \in \rsetcur$ and for every cluster  $\cluster \in \parclusterset$, there may exist some tuples in $\cluster$ such that for some repair candidates $\Refquery \in \rsetcur$, the tuples fulfill the condition of $\Refquery$. We use these two sets of clusters to determine bounds on the arithmetic expression $[\lb{\expr},\ub{\expr}]$ of the \gls{AC} $\propconstr$.
The algorithm then distinguishes between three cases (line 8-13):
(i) function \( \acrangeevalall \) uses $\fullclusterset$ and $\parclusterset$ to determine whether \( \propconstr \) is guaranteed to hold for every  \( \Refquery \in \rsetcur \). For that we compute bounds $[\lb{\expr},\ub{\expr}]$ on $\expr$ that hold for every $\Refquery \in \rsetcur$.
If this is the case then all \( \Refquery \in \rsetcur \) are repairs and we add $\rsetcur$ to
\( \rsettopk \) (lines 8–9);
(ii) function \( \acrangeevalexists \) determines 
$\fullclusterset$ and $\parclusterset$ to check whether some repair candidates \( \Refquery \in \rsetcur \) may fulfill the \gls{AC} and needs to be further examined (lines 10–13); (iii) if both $\acrangeevalall$ and $\acrangeevalexists$ return false, then it is guaranteed that no $\Refquery \in \rsetcur$ is a repair and we can discard $\rsetcur$.
We will discuss these functions in depth in \Cref{subsec: interval arithmatic}.

For example, if \( \propconstr \defas 0.7 \leq \expr \) and we compute bounds $[\lb{\expr},\ub{\expr}] = [0.5,1]$ that hold for all $\Refquery \in \rsetcur$, then $\acrangeevalall$ returns false as some $\Refquery \in \rsetcur$ may not fulfill the constraint. However, $\acrangeevalexists$ return true as some $\Refquery \in \rsetcur$ may fulfill the constraint. In this case, the algorithm partitions \( \rsetcur \) into smaller sub-ranges \( \rsetnew \) using the function \( \rangedivide(\rsetcur) \) (line 11). Assume that $\rsetcur = [ [\lb{\constant_1}, \ub{\constant_1}], \ldots, [\lb{\constant_\numcond}, \ub{\constant_\numcond}]]$. \rangedivide splits each range $[\lb{\constant_i}, \ub{\constant_i}]$ into a fixed number of fragments $\{ [\lb{\constant_{i_1}}, \ub{\constant_{i_1}}], \ldots, [\lb{\constant_{i_l}}, \ub{\constant_{i_l}}] \}$ such that each $[\lb{\constant_{i_j}}, \ub{\constant_{i_j}}]$ is roughly of the same size  and returns the following set of repair candidate sets:
\[
\{ [[\lb{\constant_{1_{j_1}}}, \ub{\constant_{i_{j_1}}}], \ldots, [\lb{\constant_{\numcond_{j_{\numcond}}}}, \ub{\constant_{\numcond_{j_\numcond}}}]] \mid [j_1, \ldots, j_{\numcond}] \in [1,l]^\numcond \}
  \]
Each \rangedivide forces further evaluation down to the leaf clusters, making recursive partitioning the dominant factor in the overall runtime. In the worst-case scenario, where every split results in another partitions case, the algorithm must descend to the finest-grained leaf clusters, causing its runtime to approach that of the brute-force approach.

That is, each $\rsetnew$ has one of the fragments for each $[\lb{\constant_i}, \ub{\constant_i}]$ and the union of all repair candidates in these repair candidate sets is $\rsetcur$. We use $l=2$ in our implementation.
The function $\hascandidates$ (line 12-13) checks whether each range in \( \rsetnew \) contains at least one value that exists in the data. This restricts the search space to only include candidates that actually appear in the data. Recall from our discussion of the search space at the end of \Cref{sec:problem Overview} that we only 
consider values from the active domain of an attribute as constants for repair candidates. That is, we can skip candidate repair sets $\rsetnew$ that do not contain any such values.
For example, if the dataset contains only values 8 and 10 for a given attribute, then applying a filter  \( \attr \leq 9 \) would yield the same result as \( \attr \leq 8 \), since no data points lie between 8 and 10. 
If this condition is satisfied, \( \rsetnew \) is inserted into the priority queue \( \pqueue \) to be processed in future iterations of the main loop.
%
%
In each iteration we use function \( \fcall{topkConcreteCand}\) (line 14) to determine the $k$ repairs $\query_i$ across all $\Repqueryset \in \rsettopk$ with the lowest distance to the user query $\query$. If we can find $k$ such candidates (line 15), then we test whether no repair candidate from the next repair candidate set $\rsetnext$ may be closer to $\query$ then the $k$th candidate $\settopk[k]$ from $\settopk$ (line 16). This is the case if the lower bound on the distance of any candidate in $\rsetnext$ is larger than the distance of $\settopk[k]$. Furthermore, the same holds for all the remaining repair candidate sets in $\rsettopk$, because $\rsettopk$ is sorted on the lower bound of the distance to the user query.
That is, $\settopk$ contains exactly the top-k repairs and the algorithm returns $\settopk$.

\subsection{Determining Covering Cluster Sets} \label{subsec:optimized_evaluation}
\label{sec:determ-cover-clust}

Similar to \gls{algff}, we can use the \tree to determine a covering cluster set \clusterset. However, as we now deal with a set of candidate repairs \Repqueryset, we would have to find a \clusterset such that for all $\Refquery \in \Repqueryset$ we have:
\(
  \Refquery(\db) = \bigcup_{\cluster \in \clusterset} \cluster
\).
Such a covering cluster set is unlikely to exist as for any two $\Refquery \neq \Refquery' \in \Repqueryset$ it is likely that $\Refquery(\db) \neq \Refquery'(\db)$. Instead we relax the condition and allow clusters $\cluster$ that are \emph{partially covered}, i.e., for which some tuples in $\cluster$  may be in the result of some candidates in \Repqueryset. We modify \Cref{alg:filter_fully} to take a repair candidate set as an input and to return two sets of clusters: $\fullclusterset$ which contains clusters for which all tuples fulfill the selection condition of all $\Refquery \in \Repqueryset$ and $\parclusterset$ which contains clusters that are only partially covered.

Analogous to \Cref{alg:filter_fully}, the updated algorithm (\Cref{alg:optimized_filter_ranges}) maintains a stack of clusters to be processed that is initialized with the root node of the \tree (line 1). In each iteration of the main loop (line 3-16), the algorithm determines whether all tuples of the current cluster $\cluster_{cur}$ fulfill the conditions $\theta_i$ for all repair candidates $\Refquery \in \Repqueryset$. This is done using function $\evalallrange$
(line 7).
Additionally, we check whether it is possible that at least one tuple fulfills the condition of at least one repair candidate $\Refquery \in \Repqueryset$. This is done using a function $\evalsomerange$ (line 8). If the cluster is fully covered we add it to the result set $\fullclusterset$ (line 10).
If it is partially covered, then we distinguish between two cases (line 11- 16). Either the cluster is a leaf node (line 12-13) or it is an inner node (line 14-16). If the cluster is a leaf, then we cannot further divide the cluster and add it to $\parclusterset$. If the cluster is an inner node, then we process its children as we may be able to determine that some of its children are fully covered or not covered at all.

\Cref{tab:clusterOperators} shows how conditions are evaluated by \evalallrange and \evalsomerange.
For a condition $\attr > \constant$, if the lower bound of attribute $\lb{\attr}$ is larger than the upper bound $\ub{\constant}$, then all tuples in the cluster fulfill the condition for all $\Refquery \in \Repqueryset$. The cluster is partially covered if $\ub{\attr} > \lb{\constant}$ as then there exists at least one value in the range of $\attr$ and constant $\constant$ in $[\lb{\constant},  \ub{\constant}]$ for which the condition is true.

In the example from ~\Cref{fig:flowchart}, a repair candidate \( [[33, 37]] \) is evaluated. 
Recall that the single condition in this example is $T \geq \constant$.
 $\cluster_{root}$ has $\boundsa{T} = [27, 37]$. The algorithm first applies \evalallrange to check if all tuples in $\cluster_{root}$ satisfy the condition. Since \( 27 \not\ge 37 \), the algorithm proceeds to evaluate the condition for partial coverage using \evalsomerange.
 Since \( \cluster_1 \) is partially covered and not a leaf, the algorithm continues by processing \( \cluster_1 \)'s children, \( \cluster_2 \) and \( \cluster_3 \).
 For \( \cluster_3 \), a similar situation occurs: the lower bound of the attribute, \( \lb{\attr} = 31 \), is not greater than the upper bound of the constant, \( \ub{\constant} = 37 \) and we have to process additional clusters, \( \cluster_8 \) and \( \cluster_9 \). 
The same holds for \( \cluster_2 \) and we process its children: \( \cluster_4 \) and \( \cluster_5 \). Additionally, \( \cluster_4 \) fails \evalallrange but satisfies partial coverage with \evalsomerange, necessitating  evaluation of its children, \( \cluster_6 \) and \( \cluster_7 \).
%
%
Finally, the algorithm applies \evalallrange and \evalsomerange if necessary to the clusters \( \cluster_5 \), \( \cluster_6 \), \( \cluster_7 \), \( \cluster_8 \), and \( \cluster_9 \), confirming that \( \cluster_8 \in \fullclusterset\) and \( \cluster_7 \in \parclusterset  \), as \( t_3.T = 37 \geq \constant \) is true for all $\constant \in [33,37]$ and \( t_4.T = 34 \geq \constant \) is only true for some $\constant \in [33,37]$. 

\subsection{Computing Bounds on Constraints}\label{subsec: interval arithmatic}
Given the cluster sets $(\fullclusterset, \parclusterset)$ computed by \Cref{alg:optimized_filter_ranges}, we next (i) compute bounds on the results of the aggregation queries $\aggq{i}$ used in the constraint, then (ii) use these bounds to compute bounds
$[\lb{\expr},\ub{\expr}]$ on the result of the arithmetic expression $\expr$ of the \gls{AC} \propconstr over repair candidates in \Repqueryset.
These bounds are conservative in the sense that all possible results are guaranteed to be included in these bounds. Then, finally, (iii)
 function $\acrangeevalall$ uses the computed bounds to determine whether all candidates in \Repqueryset fulfill the constraint by applying \evalallrange from \Cref{tab:clusterOperators}. For a constraint $\propconstr \defas \propthresh\,
\op \expr$, $\acrangeevalall$ calls \evalallrange with $[\lb{\expr},\ub{\expr}]$ and $\propthresh$.
$\acrangeevalexists$ uses $\evalsomerange$ instead to determine whether some candidates in \Repqueryset may fulfill the constraint.
This requires techniques for computing bounds on the possible results of arithmetic expressions and aggregation functions when the values of each input of the computation are known to be bounded by some interval.

\subsubsection{Bounding Aggregation Results} \label{sub:BoundsPhase1}
We now discuss how to compute bounds for the results of the filter-aggregation queries $\aggq{i}$ of an aggregate constraint $\propconstr$ based on the cluster sets $(\fullclusterset, \parclusterset)$ returned by \Cref{alg:optimized_filter_ranges}.
As every cluster \cluster in \fullclusterset is fully covered for all repair candidates in \Repqueryset, i.e., all tuples in the cluster fulfill the conditions of each $\Refquery \in \Repqueryset$, the materialized aggregation results $\aggq{i}(\cluster)$ of \cluster contribute to both the lower bound $\lb{\aggq{i}}$ and upper bound $\ub{\aggq{i}}$ as for \gls{algff}. For partially covered clusters ($\parclusterset$), we have to make worst case assumptions to derive valid lower and upper bounds. For the lower bound, we have to consider the minimum across two options: (i) no tuples from the cluster will fulfill the condition of at least one \Refquery in \Repqueryset. In this case, the cluster is ignored for computing the lower bound e.g., in case for $\f{max}$; (ii) based on the bounds of the input attribute for the aggregation within the cluster, there are values in the cluster that if added to the current aggregation result further lowers the result, e.g., a negative number for $\f{sum}$ or a value smaller than the current minimum for $\f{min}$. For example, for $\f{min}(\attr)$ we have to reason about two cases: (i) we can add $\lb{\attr}$ to the aggregation in case of negative numbers; (ii) otherwise should ignore this cluster for computing lower bounds.
For $\f{sum}$ we have the two cases: (i) the attribute for the aggregation has negative numbers. In this case we sum the negative numbers for the lower bound. (ii) otherwise should ignore this cluster for computing lower bounds.
For the upper bound we have the symmetric two cases: (i) 
if including no tuples from the cluster would result in a larger aggregation result, e.g., for $\f{sum}$ when all values in attribute $\attr$ in the cluster are negative then including any tuple from the cluster would lower the aggregation result and (ii) if the upper bound of values for the aggregation input attribute within the cluster increases the aggregation result,
we include the aggregation bounds in the computation for the upper bound.

\subsubsection{Bounding Results of Arithmetic Expressions}
\label{sec:bound-results-arith}
Given the bounds on aggregate-filter queries, we use \emph{interval arithmetic}~\cite{stolfi2003introduction, de2004affine} which computes sound bounds for the result of arithmetic operations when the inputs are bound by intervals. In our case, the bounds on the results of aggregate queries $\aggq{i}$ are the input and bounds $[\lb{\expr}, \ub{\expr}]$  on $\expr$ are the result.   
The notation we use is similar to~\cite{zhang2007efficient}.  \Cref{tab:operation_bounds} shows the definitions for arithmetic operators we support in aggregate constraints. Here, $\lb{E}$ and $\ub{E}$ denote the lower and upper bound on the values of expression $E$, respectively. For example, for addition the lower bound for the result of addition $\lb{E_1 + E_2}$ of two expressions $E_1$ and $E_2$ is $\lb{E_1} + \lb{E_2}$.
\subsubsection{Bounding Aggregate Constraint Results} \label{sub:BoundsPhase2}
Consider a constraint $\propconstr \defas \propthresh \op \expr$.
There are three possible outcomes for a repair candidate set: (i) $\propthresh \op \expr$ is true for all $[\lb{\expr}, \ub{\expr}]$ which $\acrangeevalall$ determines using \evalallrange and bounds $[\propthresh,\propthresh]$; (ii) some of the candidate in \Repqueryset may fulfill the condition, which $\acrangeevalexists$ determines using \evalsomerange;
(iii) none of the candidates in \Repqueryset fulfill the condition (both (i) and (ii) are false). 

In the running example from \Cref{fig:flowchart}(g), the covering set of clusters for repair candidate set $\Repqueryset \defas [[33,37]]$ are \( \fullclusterset = \{\cluster_8\} \) and $\parclusterset = \{\cluster_7\}$.
To evaluate \( \aggq{1} = \f{count}(G = M \land Y = 1) \) over these clusters, the algorithm include the materialized aggregation results for $\cluster_8$ for both the lower bound $\lb{\aggq{i}}$ and upper bound $\ub{\aggq{i}}$.
For the partially covered $\cluster_7$, the lower bound of $\aggq{1\,\cluster_7}$ is $0$ for this cluster (the lowest count is achieved by excluding all tuples from the cluster), while the upper bound is $1$, as there exists a male in the cluster satisfying $Y = 1$. Thus, we get the following bounds for $\aggq{1\,\cluster_7} = [0,1]$.
Similarly, we compute the remaining aggregation bounds:
$
\aggq{1\,\cluster_8} = [1,1], \quad \aggq{2\,\cluster_7} = [0,1], \quad \aggq{2\,\cluster_8} = [1,1], \quad
\aggq{3\,\cluster_7} = [0,0], \quad \aggq{3\,\cluster_8} = [0,0], \quad \aggq{4\,\cluster_7} = [0,0], \quad \aggq{4\,\cluster_8} = [0,0]
$.

Next, in \Cref{fig:flowchart}(h) we sum the lower and upper bounds for each aggregation $\aggq{i}$ across all clusters in $\clusterset$:
$
\aggq{1} = \aggq{1\,\cluster_7} + \aggq{1\,\cluster_8} = [1,2]
$,
$
\aggq{2} = \aggq{2\,\cluster_7} + \aggq{2\,\cluster_8} = [1,2]
$,
$
\aggq{3} = \aggq{3\,\cluster_7} + \aggq{3\,\cluster_8} = [0,0], \quad
\aggq{4} = \aggq{4\,\cluster_7} + \aggq{4\,\cluster_8} = [0,0]
$.
We then substitute the computed values $\{\aggq{1}, \aggq{2}, \aggq{3}, \aggq{4}\}$ into $\propconstr_{\#}$ and evaluate the resulting expression using interval arithmetic (\Cref{tab:operation_bounds}). Given:
$
\propconstr_{\#} = \sfrac{\f{\aggq{1}}}{\f{\aggq{2}}} -  \sfrac{\f{\aggq{3}}}{\f{\aggq{4}}}
$
the lower and upper bounds for the first term $\sfrac{\f{\aggq{1}}}{\f{\aggq{2}}}$ are computed as:
[$\lb{\sfrac{{\f{E_1}}}{\f{E_2}}}, \quad \ub{\sfrac{{\f{E_1}}}{\f{E_2}}}$] = $[\sfrac{1}{2},2]$. Similarly, for the second term:
$\sfrac{\f{\aggq{3}}}{\f{\aggq{4}}} = [0,0]$.
Applying interval arithmetic to compute the subtraction we get:
$
\lb{{\f{E_1}} - {\f{E_2}}}, \quad \ub{{\f{E_1}} - {\f{E_2}}}
$.
Thus, we obtain bounds $[\lb{\expr_{\#}}, \ub{\expr_{\#}}] = [\sfrac{1}{2},2]$  (\Cref{fig:flowchart}(i)). Since $\lb{\expr_{\#}} = \sfrac{1}{2} > 0.2$, none of the candidates in $\Repqueryset = [[33,37]]$ can be repairs and we can prune $\Repqueryset$.

In practice, \gls{algrp} performs best when \tree clusters are homogeneous with respect to the predicate attributes $\attr_i$ in $\theta$ i.e., when most cluster bounds, $\boundsa{\attr_i}$, lie entirely above or below the repair intervals $[\lb{\constant_i}, \ub{\constant_i}]$. This enables efficient pruning of infeasible or fully satisfying candidate sets. The effect is especially strong when large regions of the search space can be ruled out or accepted entirely based on the aggregation constraint $\propconstr$. If predicate attributes are strongly correlated with those in the arithmetic expression $\expr$, cluster inclusion often predicts the outcome of $\expr$, allowing entire repair sets to be evaluated at once. Conversely, when many clusters partially overlap the predicate ranges, the algorithm must recursively partition $\Repqueryset$ and evaluate finer-grained cluster levels, eventually matching the cost of brute-force in the worst case.

\ifnottechreport{
      \begin{theorem}[Correctness of \gls{algff} and \gls{algrp}]\label{theo:correctness}
    Given an instance $(\query, \db, \propconstr, k)$ of the \emph{aggregate constraint repair problem}, \Cref{alg:range_filtering} computes the solution for this problem instance.
\end{theorem}        
\begin{proof}
        We present the proof in \cite{AG25techrepo}.
      \end{proof}
    }
\iftechreport{
\subsubsection{Correctness}\label{sec:proof}
  \begin{theorem}[Correctness of \gls{algff} and \gls{algrp}]\label{theo:correctness}
Given an instance $(\query, \db, \propconstr, k)$ of the \emph{aggregate constraint repair problem}, \gls{algff} and \gls{algrp} (\Cref{alg:range_filtering}) compute the solution for this problem instance.
  \end{theorem}
Before presenting the proof of \Cref{theo:correctness} we first establish several auxiliary results used in the proof. First we demonstrate that given the conditions of a repair candidate \Refquery, \Cref{alg:filter_fully} returns a \textit{covering cluster set} \clusterset which is a set of clusters that cover exactly the tuples from the input database $\db$ that fulfill the selection condition of \Refquery. Then we proceed to show that evaluating any filter-aggregate query over the result of \Refquery using the materialized aggregation results for a covering \clusterset yields the same result as computing this aggregation on the result of \Refquery (\Cref{lem:fully-covering-clusterset}). As aggregate constraints are evaluated over the results of filter-aggregate queries (\Cref{lem:alg-filter-fully-correct}), this then immediately implies that checking whether an aggregate constraint holds for repair candidate \Refquery can be determined using the pre-aggregated results stored for clusters using a fully covering cluster set for \Refquery. Thus, algorithm \gls{algff} described in \Cref{sec:baseline} solves the \gls{repairproblem}.

Next we prove several auxiliary lemmas that will help us establish the correctness for algorithm \gls{algrp}. \gls{algrp} determines for a set of repair candidates \Repqueryset (described through ranges for each predicate) at once whether all or none of the candidates in the set fulfill the given aggregate constraint. Recall that we use \Cref{alg:range_filtering} to compute a partially covering cluster set for \Repqueryset which is a pair $(\fullclusterset, \parclusterset)$ where $\fullclusterset$ contains clusters for which all tuples fulfill the condition of all repair candidates in \Repqueryset and $\parclusterset$ contains clusters for which some tuples may fulfill the condition of some repair candidates in \Repqueryset. We show that for any $\Refquery \in \Repqueryset$, the set of tuples returned by $\Refquery$ over the input database $\db$ is a superset of the tuples in $\fullclusterset$ and a subset of the tuples in  $\fullclusterset \cup \parclusterset$ (\Cref{lem:partial-covering-cluster-}). Then we show that aggregation over partially covering cluster sets yields \textit{sound} bounds for the result of a filter-aggregate query for any candidate $\Refquery \in \Repqueryset$. That is, for every $\Refquery \in \Repqueryset$ and filter-aggregate query $\aggq{i} \defas \aggregation{\f(\attr)}(\selection_{\theta_i}(\rel))$ for $\rel = \Refquery(\db)$ we have that the bounds computed for $\aggq{i}$ contain the result of $\aggq{i}$ evaluated over $\rel$ (\Cref{lem:bounds-on-aggregation-res}).
Finally, we demonstrate that for a partially covering cluster set $(\fullclusterset, \parclusterset)$ and constraint \propconstr, $\acrangeevalall$ returns true ($\acrangeevalexists$ returns false) returns true if all repair candidates in \Repqueryset are guaranteed to fulfill (not fulfill) the constraint (\Cref{lem:universal-and-existenical}). Using these lemmas we then prove \Cref{theo:correctness}.

 \Cref{alg:filter_fully}  takes as input the kd-tree and a condition $\theta'$ and returns a set of clusters \clusterset. We first show that \clusterset cover precisely the set of tuples from $\db$ which fulfill the condition $\theta'$.

\begin{lemma}[\Cref{alg:filter_fully} Returns Covering Clustersets]\label{lem:fully-covering-clusterset}
  Consider a repair candidate $\Refquery$ with selection condition $\theta'$ and input relation $\rel$.\footnote{If the user query contains joins, we treat the join result as the input relation $\rel$.} Let $\clusterset$ be the clusterset returned by \Cref{alg:filter_fully}. We have
  \[
    \selection_{\theta'}(\rel) = \bigcup_{\cluster \in \clusterset} \cluster
  \]
\end{lemma}
\begin{proof}
  The input condition $\theta'$ is a conjunction of the form:
  \[
    \bigwedge_{i \in [1,m]} \theta_i'
  \]
  where $\theta_i'$ is a comparison of the form $\attr_i \op_i \constant_i$ and $\op_i$ is one of the supported comparison operators.
  Recall that \Cref{alg:filter_fully} traverses the kd-tree from the root maintaining a list of clusters that are fully covered, i.e., all tuples in the cluster fulfill condition $\theta'$. For each cluster $\cluster$ and attribute $\attr$ we store $\boundsa{\attr}(\cluster) = [\lb{\attr},\ub{\attr}]$ where $\lb{\attr}$ is the smallest $\attr$ value in $\cluster$ and $\ub{\attr}$ is the largest $\attr$ value:
  \begin{align*}
    \lb{\attr} &= \min \left( \{ t.\attr \mid t \in \cluster \}\right)\\
    \ub{\attr} &= \max \left( \{ t.\attr \mid t \in \cluster \}\right)\\
  \end{align*}
  For each cluster the algorithm uses the attribute ranges stored for each cluster $\cluster$ to distinguish three cases: (i) all tuples in \cluster the fulfill the condition $\theta'$: $\forall t \in \cluster: t \models \theta'$, (ii) no tuples in $\cluster$ fulfill the condition (the cluster can be ignored): $\forall t \in \cluster: t \not\models \theta'$, and (iii) some tuples in \cluster may fulfill the condition.
  Clusters for which case (i) applies are added to the result, clusters for which case (ii) applies are ignored, and for case (iii) the algorithms proceeds to the children of the cluster in the kd-tree. Note that the algorithm is sound, but not complete, in the sense that it may fail to determine that case (i) or (ii) applies and classify a cluster as case (iii) instead. This does not affect the correctness of the algorithm as for case (iii) the algorithm processes all children of the cluster.
  As the children of a cluster cover exactly the same tuples as the cluster itself, this approach is guaranteed to return a set of clusters that exactly cover the tuples in $\selection_{\theta'}(\rel)$ as long as the algorithm can correctly determine for a given cluster which case applies. For that the algorithm uses function $\evalall$ that takes as input a comparison $\attr \op \constant$ and bounds $\boundsa{\attr}(\cluster)$ for the values of $\attr$ in the cluster \cluster.
  Thus, what remains to be shown is that for each  condition $\theta_i'$ if $\evalall(\theta_i',\boundsa{\attr}(\cluster))$  returns true given the bounds $\boundsa{\attr}(\cluster)$ for the attribute $\attr$ used in $\theta_i'$, then $\forall t \in \cluster: t \models \theta_i'$. We prove this through a case distinction over the supported comparison operators. Recall that the definition of  $\evalall$ is shown in \Cref{tab:clusterOperators}.

  \mypar{Order relations ($\attr <, \leq, \geq, > c$)} We prove the claim for $<$. The remaining cases are symmetric or trivial extensions. $\evalall$ returns true if $\ub{\attr} < c$. In this case for every tuple $t \in \cluster$ we have $t.\attr \leq \ub{\attr} < c$ which implies $t.\attr < c$ and in turn $t \models \theta_i'$.

  \mypar{Equality ($\attr = c$)} $\evalall$ return true if $\lb{\attr} = \ub{\attr} = c$. In this case we have $t.\attr = c$ for every tuple $t \in \cluster$. Thus, $t \models \theta_i'$.

  \mypar{Inequality ($\attr \neq c$)} $\evalall$ return true if $c \not \in [\lb{\attr}, \ub{\attr}]$. Consider a tuple $t \in \cluster$. We know $t.\attr \in [\lb{\attr}, \ub{\attr}]$. Using $c \not \in [\lb{\attr}, \ub{\attr}]$ it follows that $t.\attr \neq c$.

  \mypar{Interval membership ($\attr \in [c_1,c_2]$)} $\evalall$ return true if $\constant_1 \leq \lb{\attr} \land \ub{\attr} \leq \constant_2$. For every tuple $t \in \cluster$ we have $t.\attr \in [\lb{\attr},\ub{\attr}]$ and, thus, also $t.\attr \in [c_1,c_2]$. Thus,  $t \models \theta_i'$.
\end{proof}

Using \Cref{lem:fully-covering-clusterset}, we demonstrate that for a given repair candidate \Refquery, any filter-aggregate query evaluated using the materialized aggregation results for the clusters returned by \Cref{alg:filter_fully} for \Refquery are the same as aggregation function results computed over the result of evaluating \Refquery on the input database $\db$. Recall that for $\f{avg}$ we compute $\f{sum}$ and $\f{count}$ and then calculate the average as $\frac{\f{sum}}{\f{count}}$.

\begin{lemma}[Aggregate Results on Fully Covering Clustersets]\label{lem:alg-filter-fully-correct}
  Consider a repair candidate \Refquery with condition $\theta'$, database $\db$, and filter-aggregate query $\aaggq \defas \aggregation_{\f(\attr)}(\selection_{\theta''}(\query(\db)))$. Let $\clusterset$ be the result returned for \Refquery by \Cref{alg:filter_fully} as shown in \Cref{lem:fully-covering-clusterset}. We have:
\[
  \aggregation_{\f{f'}(\attr)}\left( \bigcup_{\cluster \in \clusterset} \{\aaggq(\cluster)\} \right)
\]
where $\f{f'}$ is chosen based on $\f{f}$:
\begin{itemize}
\item $\f{f'} = \f{sum}$ for $\f{f} = \f{sum}$ and $\f{f} = \f{count}$
\item $\f{f'} = \f{min}$ for $\f{f} = \f{min}$
\item $\f{f'} = \f{max}$ for $\f{f} = \f{max}$
\end{itemize}
\end{lemma}
\begin{proof}
We demonstrated in \Cref{lem:fully-covering-clusterset} that $\clusterset$ contains exactly the set of tuples returned by $\selection_{\theta'}(\db)$. Recall that for each $\aaggq$ we materialize the result $\aaggq(\cluster)$. It is well-known that aggregation functions $\f{sum}$, $\f{min}$, and $\f{max}$ are associative and commutative. Thus, computing these aggregation functions over $\selection_{\theta''}(\query(\db))$ can be decomposed into computing the aggregation function results over the results for each cluster in the covering cluster set returned by \Cref{alg:filter_fully}. Thus, the lemma holds for these aggregation functions. Aggregation function $\f{count}()$ can be expressed alternatively as a sum over a constant value $1$ per tuple. Thus, the same decomposition applies. As $\f{avg} = \frac{\f{sum}}{\f{count}}$, we compute the sum and count and use these to compute the average.
\end{proof}

Next, we demonstrate that \Cref{alg:optimized_filter_ranges} returns a partially covering cluster set.

\begin{lemma}[Partially Covering Cluster Sets]\label{lem:partial-covering-cluster-}
  Consider a set of repair candidates \Repqueryset and let $(\fullclusterset, \parclusterset)$ be the result returned by \Cref{alg:optimized_filter_ranges} for \Repqueryset. The following holds for every $\Refquery \in \Repqueryset$ with condition $\theta'$:
  \[
    \bigcup_{\cluster \in \fullclusterset} \cluster \,\,\,\subseteq\,\,\, \selection_{\theta'}(\db) \,\,\,\subseteq\,\,\, \bigcup_{\cluster \in \fullclusterset \cup \parclusterset} \cluster
  \]
\end{lemma}
\begin{proof}
  \Cref{alg:optimized_filter_ranges} takes as input a set of repair candidates $\Repqueryset$ encoded as a list of ranges, one for each condition $\theta_i$ in the selection condition $\theta$ of the input query $\query$.
It
keeps a stack of clusters to be processed that is initialized with the root of the kd-tree $\cluster_{root}$. In each iteration, a single cluster $\cluster_{cur}$ is checked.
The algorithm uses variables $in$ and $pin$ to distinguish between 3 cases for $\cluster_{cur}$. Either (i) all tuples in $\cluster_{cur}$ are guaranteed to fulfill the selection conditions of every repair candidate $\Refquery \in \Repqueryset$; (ii) none of the tuples in $\cluster_{cur}$ fulfill the condition of any repair candidate $\Refquery \in \Repqueryset$; (iii) some tuples may fulfill the condition of some repair candidate $\Refquery \in \Repqueryset$.

Analog to determining fully covering cluster sets, the tests used to determine these cases are only sound, in that for some clusters for which case (i) or (ii) applies, the algorithm may fail to detect that and treat the cluster as case (iii). Again this is safe, as it is sufficient to compute lower and upper bound on aggregation results and, then, constraints.

In case (i) the cluster is added to $\cluster_{full}$ as all tuples from the cluster will be in $\selection_{\theta'}(\rel)$ for every $\Refquery \in \Repqueryset$. For case (ii) the cluster can be discarded as it none of the tuples in $\cluster_{cur}$ are in $\selection_{\theta'}(\db)$ for any $\Refquery \in \Repqueryset$. Finally, in case (ii) the algorithm processes the children of $\cluster_{cur}$ to potentially identify smaller clusters covering parts of the tuples from $\cluster_{cur}$ that are fully included (case (i)), or excluded (case (ii)). The exception are partially covered (case (iii)) leaf clusters which are included in $\parclusterset$. As long as the tests for case (i) and (ii) are sound, the claim trivially holds as $\selection_{\theta'}(\db)$ contains $\bigcup_{\cluster \in \fullclusterset} \cluster$ for any repair candidate $\Refquery \in \Repqueryset$ where $\theta'$ is the condition of \Refquery and only tuple in $\bigcup_{\cluster \in \fullclusterset \cup \parclusterset} \cluster$ can be in $\selection_{\theta'}(\db)$.

The algorithm maintains two variables $in$ and $pin$ to determine which case applies for $\cluster_{cur}$. We will prove that if
$in = \btrue$ and $pin = \btrue$ then case (i) applies, if $in = \bfalse$ and $pin = \bfalse$ then case (ii) applies, and otherwise (which will be the case if $in \bfalse$ and $pin = \btrue$) then case (iii) applies. The algorithm sets $in =\btrue$ if for all conditions $\theta_i = \attr_i \op_i \constant_i$ in the condition $\theta$ of the user query $\query$, function \evalallrange returns true. Analog $pin = \btrue$ if the same holds using function \evalsomerange. Thus, it remains to be shown that given the bounds $\boundsa{\attr_i}(\cluster_{cur})$ for the values of $\attr_i$ of all tuples in $\cluster_{cur}$ and ranges $[\lb{\constant_i}, \ub{\constant_i}]$ of constants for repair candidates, that  (a) if \evalallrange returns true, then for every $\Refquery \in \Repqueryset$ with constant $\constant_i \in [\lb{\constant_i}, \ub{\constant_i}]$ and tuple $t \in \cluster_{cur}$ we have that $t.\attr_i \op_i \constant_i$ evaluates to true and (b) if \evalsomerange return false then for every $\Refquery \in \Repqueryset$ and every $\constant \in [\lb{\constant_i}, \ub{\constant_i}]$ and tuple $t \in \cluster_{cur}$ we have that $t.\attr_i \op_i \constant_i$ evaluates to false. As clearly under these conditions $in \Rightarrow pin$ ($\neg pin \Rightarrow \neg in$ as (a) is a more restrictive condition than (b), the claim then immediately follows.
Recall that \Cref{tab:clusterOperators} shows the definitions of \evalallrange and \evalsomerange.

\mypar{\underline{Claim (a): $\evalallrange$}}
We prove the claim for each of the supported comparison operators.

  \mypar{Order relations ($\attr <, \leq, \geq, > [\lb{\constant},\ub{\constant}]$)} We prove the claim for $<$. The remaining cases are symmetric or trivial extensions. \evalallrange returns true if $\ub{\attr} < \lb{\constant}$. In this case for every tuple $t \in \cluster_{cur}$ and repair candidate with constant $\constant \in [\lb{\constant},\ub{\constant}]$, i.e., condition $\theta_i' \defas \attr < \constant$, we have $t.\attr \leq \ub{\attr} < \lb{\constant} \leq \constant$ which implies $t.\attr < \constant$ and in turn $t \models \theta_i'$.

  \mypar{Equality ($\attr = c$)} $\evalallrange$ return true if $\lb{\attr} = \lb{\constant} = \ub{\constant} = \ub{\attr}$. This then implies $\constant = \lb{\constant} = \ub{\constant}$ as $\constant \in [\lb{\constant},\ub{\constant}]$.
We get $t.\attr = \constant$ for every tuple $t \in \cluster$. Thus, $t \models \theta_i'$.

  \mypar{Inequality ($\attr \neq c$)} $\evalallrange$ returns true if $c \not \in [\lb{\attr}, \ub{\attr}] \cap [\lb{\constant}, \ub{\constant}]$. Consider a tuple $t \in \cluster$ and repair candidate with constant $\constant$. We know $t.\attr \in [\lb{\attr}, \ub{\attr}]$ and $\constant \in [\lb{\constant}, \ub{\constant}]$. Using $c \not \in [\lb{\attr}, \ub{\attr}] \cap [\lb{\constant}, \ub{\constant}]$ it follows that $t.\attr \neq \constant$.

  \mypar{Interval membership ($\attr \in [[\lb{\constant_1},\ub{\constant_1}], [\lb{\constant_2},\ub{\constant_2}]]$)} $\evalallrange$ return true if
  \begin{align}
    \label{eq:range-membership-cond}
    \ub{\constant_1} \leq \lb{\attr} \land \ub{\attr} \leq \lb{\constant_2}
  \end{align}

  For every tuple $t \in \cluster$ and repair candidate with constants $[\constant_1, \constant_2]$ we have $t.\attr \in [\lb{\attr},\ub{\attr}]$, $\constant_1 \in [\lb{\constant_1},\ub{\constant_1}]$, and $\constant_2 \in [\lb{\constant_2},\ub{\constant_2}]$.
Then \Cref{eq:range-membership-cond} implies:
  \[
  \constant_1 \leq \ub{\constant_1} \leq \lb{\attr} \leq t.\attr \leq \ub{\attr} \leq \lb{\constant_2} \leq \constant_2
    \]
    Thus, also $t.\attr \in [\constant_1,\constant_2]$ and  $t \models \theta_i'$.

    \mypar{\underline{Claim (b): $\evalsomerange$}}
The proof for \evalsomerange is analog replacing universal with existential conditions, i.e., there may exist $t \in \cluster_{cur}$ and $\Refquery \in \Repqueryset$ such that the condition $\theta_i'$ for \Refquery holds on $t$.
\end{proof}
Given such a partially covering cluster set, we can compute sound bounds on the results of a filter-aggregate query as shown below.

\begin{lemma}[Sound Bounds on Aggregation Results]\label{lem:bounds-on-aggregation-res}
  Consider a set of repair candidates \Repqueryset and filter-aggregate query $\aggq{i}$. Let $(\fullclusterset, \parclusterset)$ be the result returned by \Cref{alg:optimized_filter_ranges} for \Repqueryset. Furthermore, let $\lb{\aggq{i}}$ and $\ub{\aggq{i}}$ be the bounds computed for the result of $\aggq{i}$. For every $\Refquery \in \Repqueryset$ with condition $\theta'$ we have:
  \[
    \lb{\aggq{i}} \leq \aggq{i}(\Refquery(\db)) \geq \ub{\aggq{i}}
  \]
\end{lemma}
\begin{proof}
  We proof the claim through a case distinction over the supported aggregation functions $\f{sum}$, $\f{min}$, $\f{max}$, $\f{count}$ and $\f{avg}$ using \Cref{lem:partial-covering-cluster-}.
  For $\lb{\aggq{i}}$ we have to reason about the smallest aggregation result that can be achieved by including all tuples from clusters in \fullclusterset that fulfill the condition $\theta'$ of the filter-aggregate query $\aggq{i}$  and a subset (possibly empty) of tuples from cluster in \parclusterset. Analog for $\ub{\aggq{i}}$ we have to determine the maximal aggregation result achievable.
  We first determine $(\fullclusterset, \parclusterset)$ using the filter condition $\theta''$ of the repair candidate. Recall that we materialize aggregation results for each $\aggq{i}$ and each cluster $\cluster$.

\mypar{\underline{$\f{min}(\attr)$}}
For $\lb{\aggq{i}}$, the smallest possible result is bound from below by including the minimum lower bound of $\attr$ across \fullclusterset and \parclusterset as the value of any tuple that may contribute to the aggregation result is bound from below by
\[
\lb{\aggq{i}} = \f{min}\left(\{ \aggq{i}(\cluster) \mid \cluster \in (\fullclusterset \cup \parclusterset) \} \right)
\]
For $\ub{\aggq{i}}$, the maximum result is achieved by excluding all tuples from clusters in \parclusterset as $\f{min}$ is antimonotone and, thus, including more tuples can only lower the minimum:
\[
  \ub{\aggq{i}} = \f{min}\left(\{ \aggq{i}(\cluster) \mid \cluster \in \fullclusterset \} \right)
\]
\mypar{\underline{$\f{max}(\attr)$}}
Aggregation function $\f{max}$ is monotone and, thus, the computation is symmetric to the one for $\f{min}$:

\begin{align*}
\lb{\aggq{i}} &= \f{max}\left(\{ \aggq{i}(\cluster) \mid \cluster \in \fullclusterset \} \right)\\
  \ub{\aggq{i}} &= \f{max}\left(\{ \aggq{i}(\cluster) \mid \cluster \in (\fullclusterset \cup \parclusterset)  \} \right)
\end{align*}

\mypar{\underline{$\f{sum}(\attr)$}}
For $\f{sum}$ the lowest possible sum is bound from below by first computing the sum of the lower bound of each cluster multiplied by the number of tuples in the cluster over the clusters in \fullclusterset. Including tuples from a cluster \cluster in \parclusterset in the sum can only lower the sum if there may be tuples in the cluster that have negative values in attribute $\attr$ which is the case if $\lb{\attr} < 0$ for this cluster. To be able to determine the greatest negative contribution we store for each cluster $\aggneg{i}(\cluster)$ and $\aggpos{i}(\cluster)$. Recall that within the context of this proof we use $\theta'$ to denote the condition of $\aggq{i}$. We use $\mylbag \myrbag$ to denote multisets and $\uplus$ to denote multiset union. Then we define:
\begin{align*}
\aggneg{i}(\cluster) &= \f{sum}( \mylbag t.\attr \mid t \in \cluster t.\attr < 0 \land t \models \theta' \myrbag)\\
\aggpos{i}(\cluster) &= \f{sum}( \mylbag t.\attr \mid t \in \cluster t.\attr > 0 \land t \models \theta' \myrbag)\\
\end{align*}
Note that here we define $\f{sum}(\emptyset) = 0$.
Based on these observations we get:
\begin{align*}
  \lb{\aggq{i}} &= \f{sum}\bigg( \mylbag \aggq{i}(\cluster) \mid \cluster \in \fullclusterset \myrbag \uplus \mylbag \aggneg{i}(\cluster) \mid \cluster \in \parclusterset \myrbag \bigg)
\end{align*}
To upper bound the maximal achievable sum, we again include aggregation results for clusters in \fullclusterset and add aggregation results over the subset of tuples in a partially covered cluster that have  positive $\attr$ values.
\begin{align*}
  \ub{\aggq{i}} &= \f{sum}\bigg( \mylbag \aggq{i}(\cluster) \mid \cluster \in \fullclusterset \myrbag \uplus \mylbag \aggpos{i}(\cluster) \mid \cluster \in \parclusterset \myrbag \bigg)
\end{align*}
\mypar{\underline{$\f{count}$}}
Function $\f{count}$ is a special case of $\f{sum}$ where each tuple contributes $1$ to the sum. Hence for the lower bound no clusters from \parclusterset are included and for the upper bound all clusters from \parclusterset are included:
\begin{align*}
\lb{\aggq{i}} &= \f{sum}\bigg( \mylbag \aggq{i}(\cluster)  \mid \cluster \in \fullclusterset \myrbag \bigg)\\
\ub{\aggq{i}} &= \f{sum}\bigg( \mylbag
                \aggq{i}(\cluster)  \mid \cluster \in \parclusterset \myrbag \bigg) \\
\end{align*}

\mypar{\underline{$\f{avg}(\attr)$}}
For $\f{avg}$, we compute bounds on $\f{sum}(\attr)$ and $\f{count}$ and then compute bounds for $\frac{\f{sum}(\attr)}{\f{count}}$ using the rules for arithmetic over bounds that we prove correct in the proof of \Cref{lem:universal-and-existenical}.
\end{proof}

Based on the sound bounds on the results of filter-aggregate queries, we can derive sound bounds for arithmetic expressions over the results of such queries. Based on such bounds, we then derive bounds on the results of aggregate constraints.

\begin{lemma}[Universal and Existential Constraint Checking is Sound]\label{lem:universal-and-existenical}
  Let $\Repqueryset =  [ [\lb{\constant_1}, \ub{\constant_1}], \ldots, [\lb{\constant_\numcond}, \ub{\constant_\numcond}]]$ be a set of repair candidates, $\propconstr$ an aggregate constraint, and $(\fullclusterset, \parclusterset)$ be a partially covering cluster set for $\Repqueryset$. Then the following holds:
  \begin{align*}
    \acrangeevalall(\propconstr, \fullclusterset, \parclusterset) = \btrue &\Rightarrow \forall \query \in \Repqueryset: \query \text{\textbf{ is a repair}}\\
    \acrangeevalexists(\propconstr, \fullclusterset, \parclusterset) = \bfalse &\Rightarrow \forall \query \in \Repqueryset: \query \text{\textbf{ is not a repair}}\\
  \end{align*}
\end{lemma}
\begin{proof}
  Consider an aggregate constraints $\propconstr \defas \expr \op \propthresh$ where $\expr$ is an arithmetic expression over the results of the filter-aggregate queries used in the constraint.
  We have shown in \Cref{lem:bounds-on-aggregation-res} that the bounds we compute for the results of filter-aggregate queries are sound. Both $\acrangeevalall$ and $\acrangeevalexists$ use interval arithmetic~\cite{stolfi2003introduction, de2004affine} to compute ranges that bound all possible results of arithmetic expressions based on ranges bounding their inputs. The soundness of these rules has been proven in related work (e.g., \cite{feng2021efficient, zhang2007efficient}). It is sufficient to demonstrate that the semantics for individual operators / functions is sound as the composition of two functions that are sound is guaranteed to be sound. In summary, given the bounds on filter-aggregate query results, we can compute sound bounds $[\lb{\expr}, \ub{\expr}]$ for the result of the arithmetic part $\expr$ of the aggregate constraint $\propconstr$.

\mypar{\underline{$\acrangeevalall$}}
Function $\acrangeevalall$ first computes $[\lb{\expr}, \ub{\expr}]$ and then uses \evalallrange with the condition  $\expr \op \propthresh$ of the aggregate constraint to check whether the constraint is guaranteed to evaluate to true for all repair candidates $\Refquery \in \Repqueryset$. As we have already shown that if \evalallrange returns true, then the condition holds for any constants within the input bounds, we know that if $\acrangeevalall$ returns true, then every $\Refquery \in \Repqueryset$ is a repair.

\mypar{\underline{$\acrangeevalexists$}}
Function $\acrangeevalexists$ also computes $[\lb{\expr}, \ub{\expr}]$ and then uses \evalsomerange with the condition  $\expr \op \propthresh$ to check whether some repair candidate $\Refquery \in \Repqueryset$ may be a repair. As we have already shown that if \evalsomerange returns false, then for all constants within the input bounds the condition evaluates to false, the claim follows immediately.
\end{proof}

Using the preceding lemmas, we are now ready to prove \Cref{theo:correctness}.

\begin{proof}[Proof of \Cref{theo:correctness}]

  \mypar{Correctness of \glsfmtfull{algff}}
  Algorithm \gls{algff} enumerates all repair candidate in increasing order of their distance from the user query (breaking ties arbitrarily) and for each candidate evaluates the aggregate constraint using a covering cluster set computed using \Cref{alg:filter_fully} and then returns the first $k$ repair candidates that fulfill the aggregate constraint. As we have shown in \Cref{lem:alg-filter-fully-correct} that \Cref{alg:filter_fully} returns a set of clusters that precisely covers the tuples returned by the repair candidate and that computing the results of filter-aggregate queries computed using the materialized aggregation results for each cluster yields the same result as evaluating the filter-aggregate query on the result of the repair candidate, \gls{algff} correctly identifies repairs.

  \mypar{Correctness of \glsfmtfull{algrp} (\Cref{alg:optimized_filter_ranges})}
  \Cref{alg:optimized_filter_ranges} maintains a queue of repair candidate sets to be processed that is initialized with a single set covering all repair candidates which is sorted based on the lower bound on of the distance of repair candidates in the set to the user query.

  In each iteration, the algorithm pops one repair candidate set $\Repqueryset_{cur}$ from the queue and computes a partially covering cluster set for $\Repqueryset_{cur}$. If $\acrangeevalall$ returns true on this partially covering cluster set then all candidates in $\Repqueryset_{cur}$ are guaranteed to be repairs (\Cref{lem:universal-and-existenical}) and the set is added to list $rcand$. Otherwise, the algorithm checks whether some candidates in $\Repqueryset_{cur}$ may be repairs using $\acrangeevalexists$. From \Cref{lem:universal-and-existenical} we know that if $\acrangeevalexists$ returns false, then none of the candidates in $\Repqueryset_{cur}$ can be repairs and $\Repqueryset_{cur}$ can be discarded.

  Otherwise, the algorithm splits $\Repqueryset_{cur}$ into multiple subsets such that the union of these subsets is $\Repqueryset_{cur}$ and pushes them onto the work queue. Note that $\fcall{hasCandidates}$ is used to check whether at least one repair candidate in each new set $\Repqueryset_{new}$ that differs in terms of returned tuples from at least one repair candidate outside the set. This check is based on the observation that for certain comparison operators, different thresholds are guaranteed to yield the same result and for some thresholds we are guaranteed to get an empty result. For instance, $\attr = \constant$ is guaranteed to evaluate to false if $\constant$ does not appear in attribute $\attr$.

  In each iteration the algorithm updates the concrete set of repairs $\settopk$ that are closest to the user query. Furthermore, in each iteration, the algorithm peeks at the next repair candidate set $\Repqueryset_{next}$ in $queue$. If we have found at least $k$ repairs and the distance of the furthest repair in $\settopk$ is smaller than the lower bound of the distance for repair candidates in $\Repqueryset_{next}$ then we are guaranteed to have found the top-k repairs and the algorithm terminates. The correctness of the algorithm follows based on the following observation: (i) the algorithm starts with the set of all possible repair candidates, (ii) in each iteration sets of repair candidates are only discarded if they are guaranteed to not contain any repair, (iii) in each iteration sets of repair candidates are only included in $rcand$ if they are guaranteed to only contain repairs, and (iv) for sets that are neither added to $rcand$ or discarded, the set $\Repqueryset_{cur}$ is split into multiple subsets such that their union covers exactly $\Repqueryset_{cur}$, i.e., no repair candidates that may be repairs are discarded.
\end{proof}


}

\begin{table}[t]
    \centering
    \vspace{-2mm}
    \renewcommand{\arraystretch}{1.5}
    \setlength{\tabcolsep}{10pt}
    \caption{Bounds on applying an operator to the result of expressions $E_1$ and $E_2$ with interval bounds~\cite{zhang2007efficient}.} 
    \label{tab:operation_bounds} 
    \vspace{-2mm}
    \begin{tabular}{|c|l|}
        \hline
      \textbf{op} & \textbf{Bounds for the expression ($E_1$ op $E_2$)}                                                                               \\
        \hline
        +         & $\lb{E_1 + E_2} = \lb{E_1} + \lb{E_2} \hfill \ub{E_1 + E_2} = \ub{E_1} + \ub{E_2}$                                                \\
        \hline
      $-$         & $\lb{E_1 - E_2} = \ub{E_1} - \lb{E_2} \hfill \ub{E_1 - E_2} = \lb{E_1} - \ub{E_2}$                                                \\
      \hline
        $\times$  & $\lb{E_1 \times E_2} = \min(\lb{E_1} \times \lb{E_2}, \lb{E_1} \times \ub{E_2}, \ub{E_1} \times \lb{E_2}, \ub{E_1} \times \ub{E_2})$ \\
                  & $\ub{E_1 \times E_2} = \max(\lb{E_1} \times \lb{E_2}, \lb{E_1} \times \ub{E_2}, \ub{E_1} \times \lb{E_2}, \ub{E_1} \times \ub{E_2})$ \\
      \hline
      $/$         & $\lb{E_1 / E_2} = \min(\lb{E_1}/\lb{E_2}, \lb{E_1}/\ub{E_2}, \ub{E_1}/\lb{E_2}, \ub{E_1}/\ub{E_2})$                                  \\
                  & $\ub{E_1 / E_2} = \max(\lb{E_1}/\lb{E_2}, \lb{E_1}/\ub{E_2}, \ub{E_1}/\lb{E_2}, \ub{E_1}/\ub{E_2})$        \\
        \hline
    \end{tabular}
\end{table}
\section{Experiments} \label{sec:expermint}

We start by comparing
the brute force approach and the baseline \gls{algff} technique (\Cref{subsec: baseline}) against our \gls{algrp} algorithm (\Cref{sec:Range})  in~\Cref{subsec: compare proposed techniques}. 
We then investigate the impact of several  factors on performance in~\Cref{subsec: invistigation}, including dataset size and similarity of the top-k repairs to the user query. Finally, in~\Cref{subsec: compare with related work}, we compare
with \emph{\gls{erica}}~\cite{LM23} which targets group cardinality constraints.

\subsection{Experimental Setup}
\label{sec:setup}

\mypar{Datasets}
We choose two real-world datasets of size 500K, \emph{\gls{dsacs}}~\cite{FriedlerSVCHR19} and \emph{\gls{dshealth}}~\cite{grafberger2021mlinspect}, that are commonly used to evaluate fairness.
We also utilize the  \emph{\gls{tpch}}~\cite{tpc-hLink} benchmark,
to varying dataset size
from 25K to 500K.
\iftechreport{We converted categorical columns into numerical data as the algorithms are designed for numerical data.}

\mypar{Queries}
\Cref{tab:queries} shows the queries used in our experiments. For \gls{dshealth}, we use queries $Q_1$ and $Q_2$ from~\cite{LM23} and a new query $Q_3$.
For \gls{dsacs}, we use $Q_4$ from~\cite{LM23} and new queries $Q_5$ and $Q_6$.
$Q_7$ is a query with 3 predicates inspired by \gls{tpch}'s $Q_2$.

\mypar{Constraints}
For \gls{dshealth} and \gls{dsacs}, we enforce the \gls{SPD} between two demographic groups to be within a certain bound.
\Cref{tab:constraints} shows the details of the constraints used. In some experiments, we vary the bounds
$B_l$ and $B_u$. For a constraint $\propconstr_i$ we use $\acperc{i}{p}$ to denote a variant of $\propconstr_i$ where the bounds have been set such that the top-k repairs are within the first p\% of the repair candidates ordered by distance. An algorithm that explores the individual repair candidates in this order has to explore the first p\% of the search space.
\ifnottechreport{For the detailed settings see~\cite{AG25techrepo}.}
For \gls{dsacs} (\gls{dshealth}), we determine the groups for \gls{SPD} based on gender and race (race and age group).
For \gls{tpch}, we enforce the constraint $\propconstr_{5}$ \iftechreport{as described in~\Cref{subsec:generalUse}.}
\ifnottechreport{which minimizes the impact of supply change disruption, where the company wants only a certain amount of expected revenue to be from certain countries.}
We use $\propconstrset$ to denote a set of \glspl{AC}.
$\propconstrset_{6}$ through $\propconstrset_{8}$ are sets of cardinality constraints for comparison with \gls{erica}. As mentioned in \cref{sec:problem Overview}, we present our repair methods for a single \gls{AC}. However,  as we already mentioned in the original submission in the experimental section (top left of page 9 in the original submission), the methods can be trivially extended to find repairs for a set / conjunction of constraints, i.e., the repair fulfills $\bigwedge_{\propconstr \in \propconstrset} \propconstr$. For \gls{algrp} (\Cref{alg:range_filtering}), it is sufficient to replace the condition in Line \ref{algline:acevalall} with $\bigwedge_{\propconstr \in \propconstrset} \acrangeevalall(\propconstr, \clusterset_{full}, \clusterset_{partial})$ and in Line \ref{algline:acevalexists} with $\bigwedge_{\propconstr \in \propconstrset} \acrangeevalexists(\propconstr, \clusterset_{full}, \clusterset_{partial})$.


\mypar{Parameters}
There are three key tuning parameters that could impact the performance of our methods.
Recall that we use a kd-tree to perform the clustering described in Section \ref{subsec:clustering}.
We consider two tuning parameters for the tree:
  \textbf{branching factor} -  each node has \branchfactor children;
  \textbf{bucket size} -
  parameter \bucketsize
determines the minimum number of tuples in a cluster. We do not split nodes $\leq \bucketsize$ tuples. When one of our algorithms reaches such a leaf node we just evaluate computations on each tuple in the cluster, e.g., to determine which tuples fulfill a condition. 
%
%
We also control $k$, the number of repairs returned by our methods. 
The default settings are follows: $\branchfactor = 5$, $k=7$, and $\bucketsize = 15$. The default dataset size is 50K tuples.

All algorithms were implemented in Python. Experiments were conducted on a machine with 2 x 3.3Ghz AMD Opteron CPUs (12 cores) and 128GB RAM.
Each experiment was repeated five times and we report median runtimes as the variance is low ($\sim 3\%$).

\begin{table}
    \caption{Queries for Experimentation}
    \vspace{-2mm}
    \adjustbox{max width=\columnwidth}{
    \begin{tabular}{|c|l|}
    \hline
         \multicolumn{2}{|l|}{
\lstinline{SELECT * FROM Healthcare}
         } \\
         \cline{1-2}
         $Q_1$ &
\begin{lstlisting}
WHERE income >= 200K AND num-children >= 3
      AND county <= 3
\end{lstlisting} \\ \cline{1-2}
  $Q_2$ &
  \begin{lstlisting}
WHERE income <= 100K AND complications >= 5
      AND num-children >= 4
\end{lstlisting} \\
\cline{1-2}
  $Q_3$ &
\begin{lstlisting}
WHERE income >= 300K AND complications >= 5
      AND county == 1
\end{lstlisting} \\ \hline \hline
\multicolumn{2}{|l|}{
\lstinline{SELECT * FROM ACSIncome}
         } \\
         \cline{1-2}
  $Q_4$ &
  \begin{lstlisting}
WHERE working_hours >= 40 AND Educational_attainment >= 19
      AND Class_of_worker >= 3
\end{lstlisting} \\ \cline{1-2}
  $Q_5$ &
  \begin{lstlisting}
WHERE working_hours <= 40 AND Educational_attainment <= 19
      AND Class_of_worker <= 4
\end{lstlisting}  \\ \cline{1-2}
  $Q_6$ &
  \begin{lstlisting}
WHERE Age >= 35  AND Class_of_worker >= 2
      AND Educational_attainment <= 15
\end{lstlisting} \\ \hline \hline
$Q_7$ &
\begin{lstlisting}
SELECT * FROM part, supplier, partsupp, nation, region
WHERE p_partkey = ps_partkey AND s_suppkey = ps_suppkey
AND s_nationkey = n_nationkey AND n_regionkey=r_regionkey
AND p_size >= 10 AND p_type in ('LARGE BRUSHED')
AND r_name in ('EUROPE')
\end{lstlisting} \\  \hline
    \end{tabular}
    }
    \label{tab:queries}
\end{table}


\begin{table}[t] 
    \centering 
    \caption{Constraints for Experimentation}
    \vspace{-2mm}
    \begin{tabular}{|p{0.2cm}|p{7.6cm}|
    }
    \hline
    \textbf{ID} & \hspace{30mm}\textbf{Constraint}
    \\ \hline
    $\propconstr_{1}$ &
    $ \frac{\f{count}(\text{race} = 1 \land \text{label} = 1)}{\f{count}(\text{race} = 1)} - \frac{\f{count}(\text{race} = 2 \land \text{label} = 1)}{\f{count}(\text{race} = 2)} \in [B_l,B_u]$
    \\ [1mm]
    \hline
    $\propconstr_{2}$ &
    {\small $ \frac{\f{count}(\text{ageGroup} = 1 \land \text{label} = 1)}{\f{count}(\text{ageGroup} = 1)} - \frac{\f{count}(\text{ageGroup} = 2 \land \text{label} = 1)}{\f{count}(\text{ageGroup} = 2)} \in [B_l,B_u]$}
    \\ [1mm]
    \hline
    $\propconstr_{3}$ &
    $ \frac{\f{count}(\text{sex} = 1 \land \text{PINCP} \geq 20k)}{\f{count}(\text{sex} = 1)} - \frac{\f{count}(\text{sex} = 2 \land \text{PINCP} \geq 20k)}{\f{count}(\text{sex} = 2)} \in [B_l,B_u]$
    \\ [1mm]
    \hline
    $\propconstr_{4}$ &
    {\small $ \frac{\f{count}(\text{race} = 1 \land \text{PINCP} \geq 15k)}{\f{count}(\text{race} = 1)} - \frac{\f{count}(\text{race} = 2 \land \text{PINCP} \geq 15k)}{\f{count}(\text{race} = 2)} \in [B_l,B_u]$}
    \\ [1mm]
    \hline
    $\propconstr_{5}$ &
    $ \frac{\sum \text{Revenue}_{\text{ProductsSelectedFromUK}}}{\sum \text{Revenue}_{\text{Selected Products}}} \in [B_l,B_u]$
    \\ [1mm]
    \hline
     $\propconstrset_{6}$ &
                            $\propconstr_{61} \defas \f{count}(\text{race}= \text{race1} )  \leq B_{u_1}$ \\
      &$\propconstr_{62} \defas  \f{count}(\text{age}= \text{group1})  \leq B_{u_2}$
      \\
    \hline
     $\propconstrset_{7}$ &
                            $\propconstr_{71} \defas \f{count}(\text{race}= \text{race1})  \leq B_{u_1}$\\
      &$\propconstr_{72} \defas \f{count}(\text{age}= \text{group1})  \leq B_{u_2}$ \\
    & $\propconstr_{73} \defas  \f{count}(\text{age}= \text{group3})  \leq B_{u_3}$
    \\
    \hline
    $\propconstrset_{8}$ &
{\small $\propconstr_{81} \defas \f{count}(\text{Sex}= \text{Female})  \leq B_{u_1}$}\\
&{\small $\propconstr_{82} \defas \f{count}(\text{Race}= \text{Black})  \leq B_{u_2}$ } \\
& {\small $\propconstr_{83} \defas \f{count}(\text{Marital}= \text{Divorced})  \leq B_{u_3}$ }
     \\ \hline
    \end{tabular}
    \label{tab:constraints}
\end{table}

\iftechreport{
\begin{figure}[t]
    \centering
    \begin{minipage}{1\linewidth}
        \begin{subfigure}[t]{0.99\linewidth}
            \centering
            \includegraphics[width=\linewidth, trim=0 30 0 0]{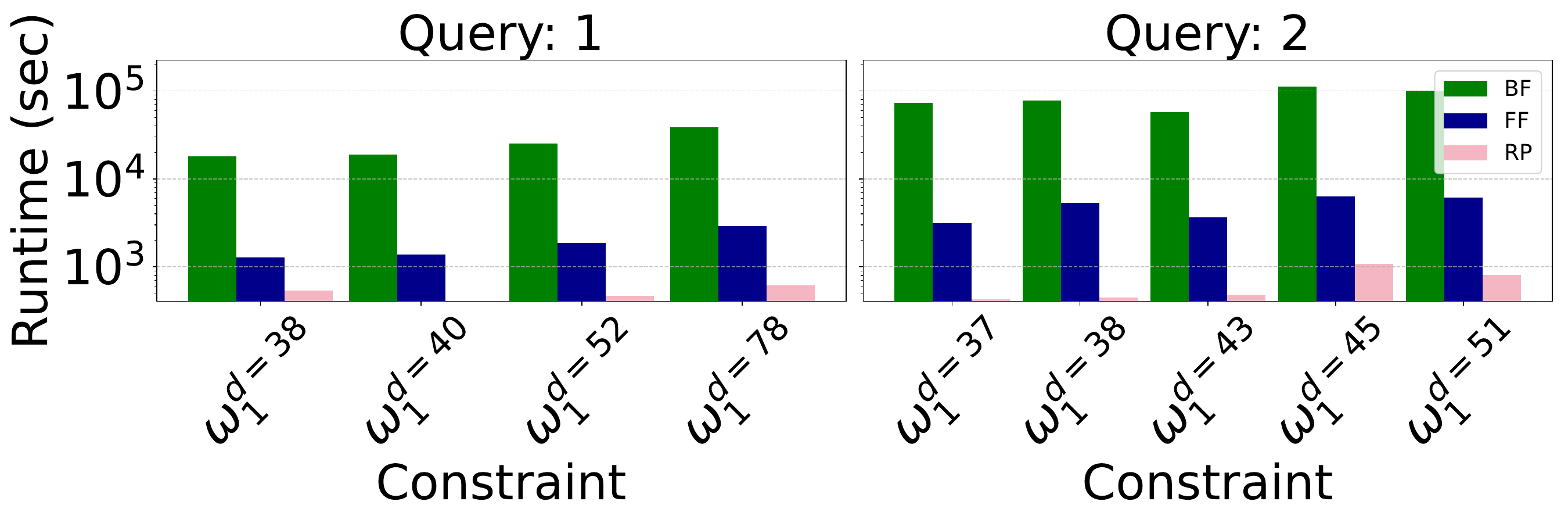}
            \caption{Runtime (sec).}
            \label{fig:runtime_comparison}
        \end{subfigure}
    \end{minipage}\\[-2mm]

    \begin{minipage}{1\linewidth}
        \begin{subfigure}[t]{0.99\linewidth}
            \centering
            \includegraphics[width=\linewidth, trim=0 30 0 0]{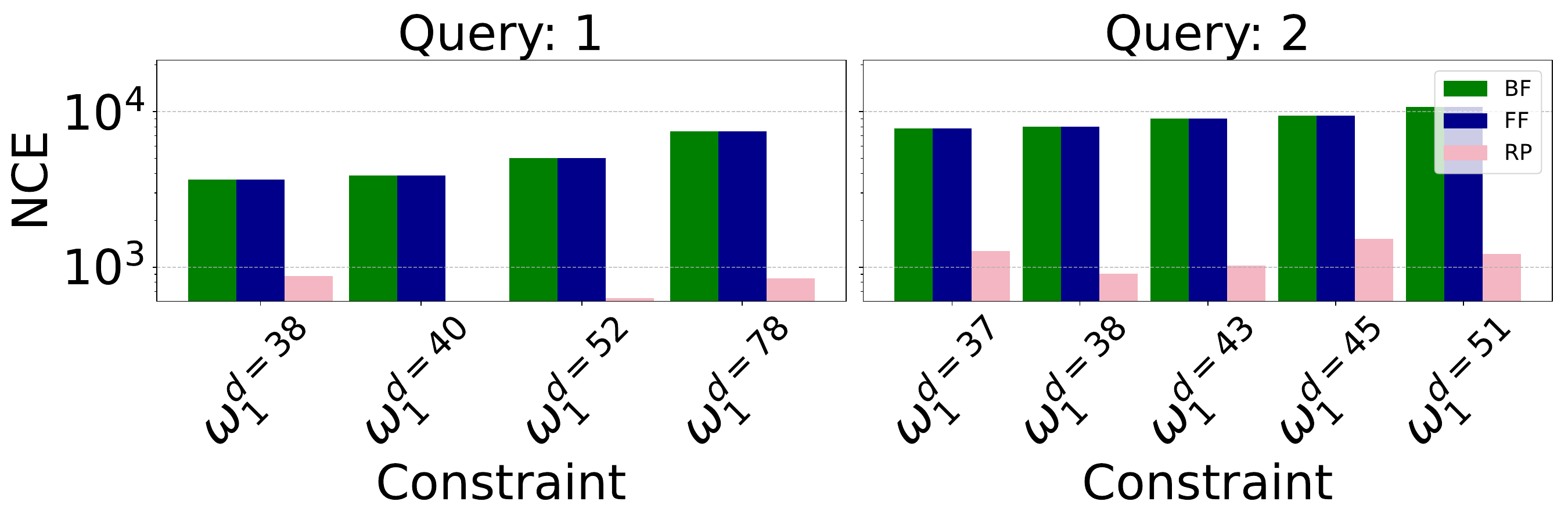}
            \caption{Total number of constraints evaluated (NCE).}
            \label{fig:checked_comparison}
        \end{subfigure}
    \end{minipage}\\[-2mm]

    \begin{minipage}{1\linewidth}
        \begin{subfigure}[t]{0.99\linewidth}
            \centering
            \includegraphics[width=\linewidth, trim=0 30 0 0]{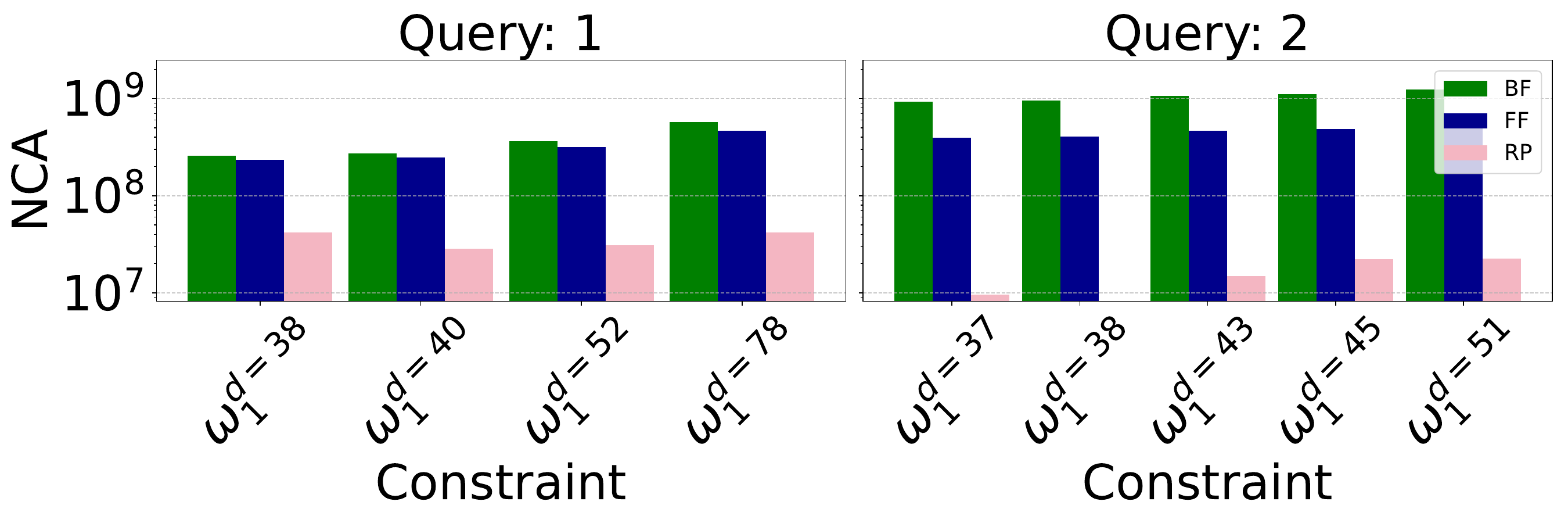}
            \caption{Total number of cluster accessed (NCA).}
            \label{fig:access_comparison}
        \end{subfigure}
    \end{minipage}
    \vspace{-4mm}
    \caption{Runtime, \gls{NCE}, and \gls{NCA} for \gls{algff}, \gls{algrp}, and brute force over the \gls{dshealth} dataset.}
    \label{fig:comparison_combined}
\end{figure}
}

\begin{figure*}[htbp]
    \centering
    \begin{minipage}{1\linewidth}
    \begin{subfigure}[t]{0.49\linewidth}
        \centering
        \includegraphics[width=\linewidth, trim=0 30 0 0]{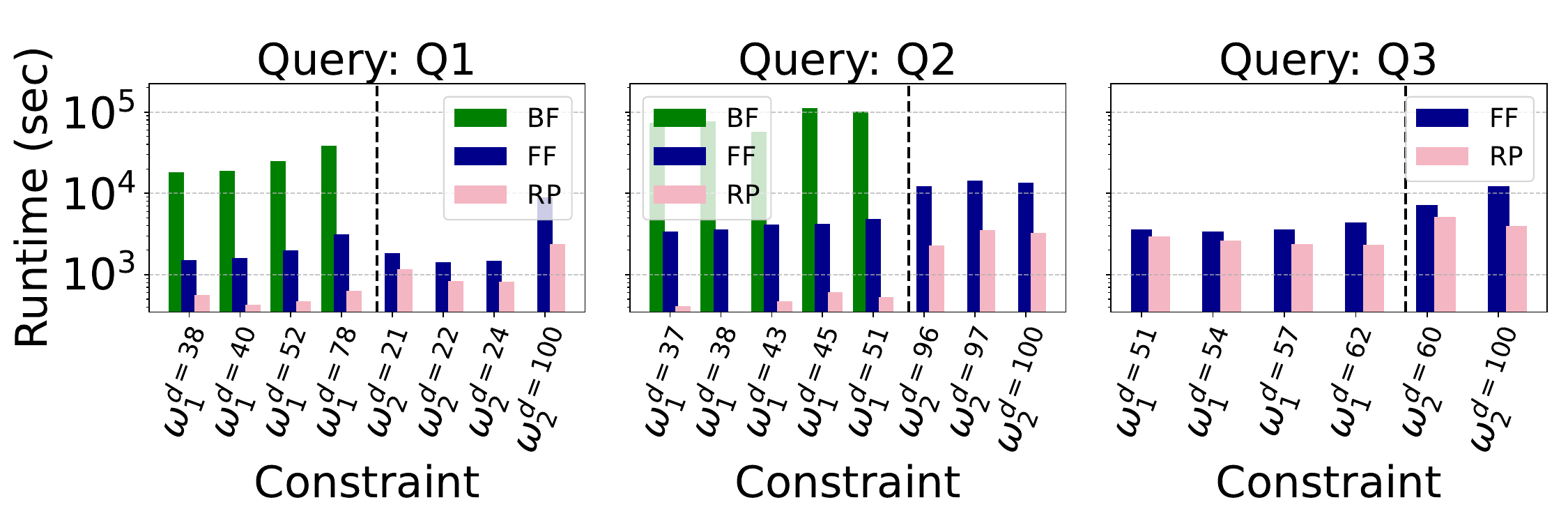}
        \caption{Runtime (sec) - \gls{dshealth} dataset.}
        \label{fig:ful_range_runtime_comparison_healthcare}
    \end{subfigure}
    \begin{subfigure}[t]{0.49\linewidth}
        \centering
        \includegraphics[width=\linewidth, trim=0 30 0 0]{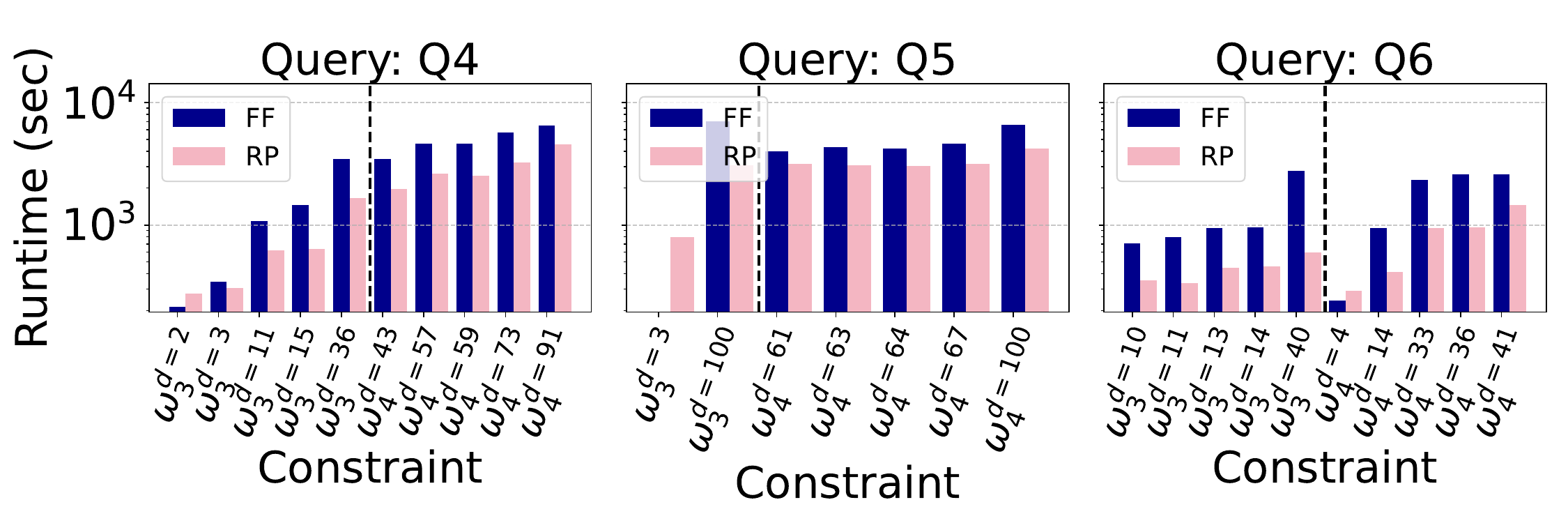}
        \caption{Runtime (sec) - \gls{dsacs} dataset.}
        \label{fig:ful_range_runtime_comparison_ACSIncome}
    \end{subfigure}
    \vspace{-3mm}
    \end{minipage}\\
    \begin{minipage}{1\linewidth}
    \begin{subfigure}[t]{0.49\linewidth}
        \centering
        \includegraphics[width=\linewidth, trim=0 30 0 0]{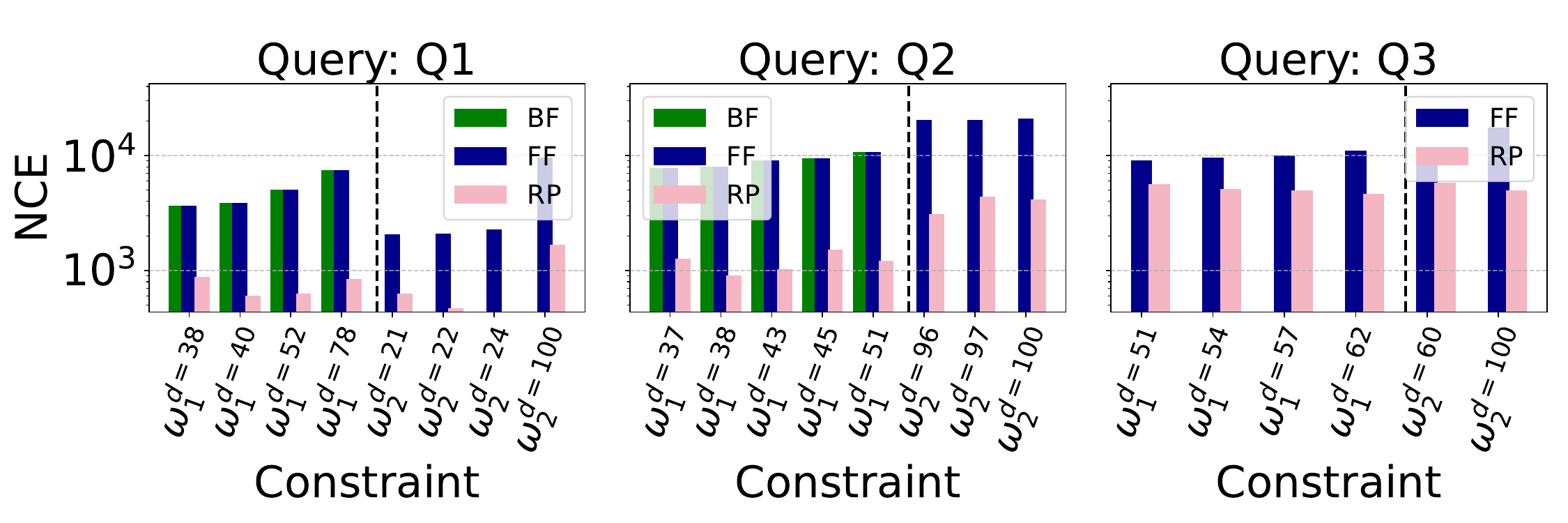}
        \caption{Total number of constraints evaluated (NCE) - \gls{dshealth} dataset.}
        \label{fig:ful_range_checked_comparison_healthcare}
    \end{subfigure}
    \hfill
    \begin{subfigure}[t]{0.49\linewidth}
        \centering
        \includegraphics[width=\linewidth, trim=0 30 0 0]{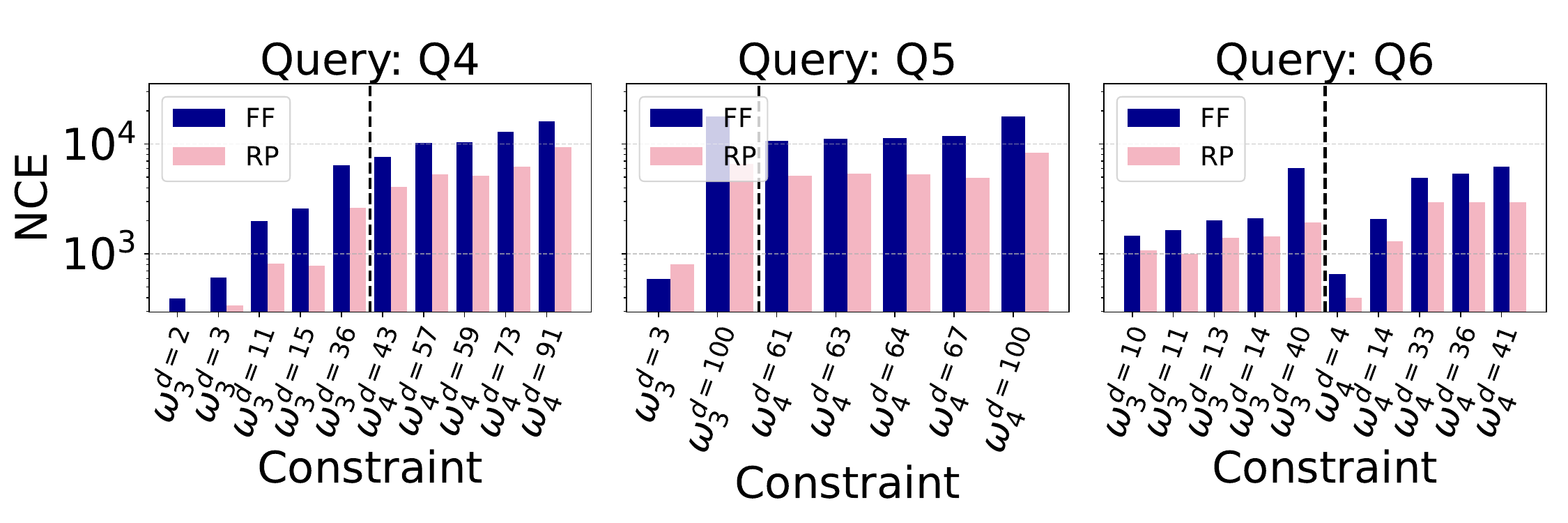}
        \caption{Total number of constraints evaluated (NCE) - \gls{dsacs} dataset.}
        \label{fig:ful_range_checked_comparison_ACSIncome}
    \end{subfigure}
    \vspace{-3mm}
    \end{minipage}\\
    \begin{minipage}{1\linewidth}
    \begin{subfigure}[t]{0.49\linewidth}
        \centering
        \includegraphics[width=\linewidth, trim=0 30 0 0]{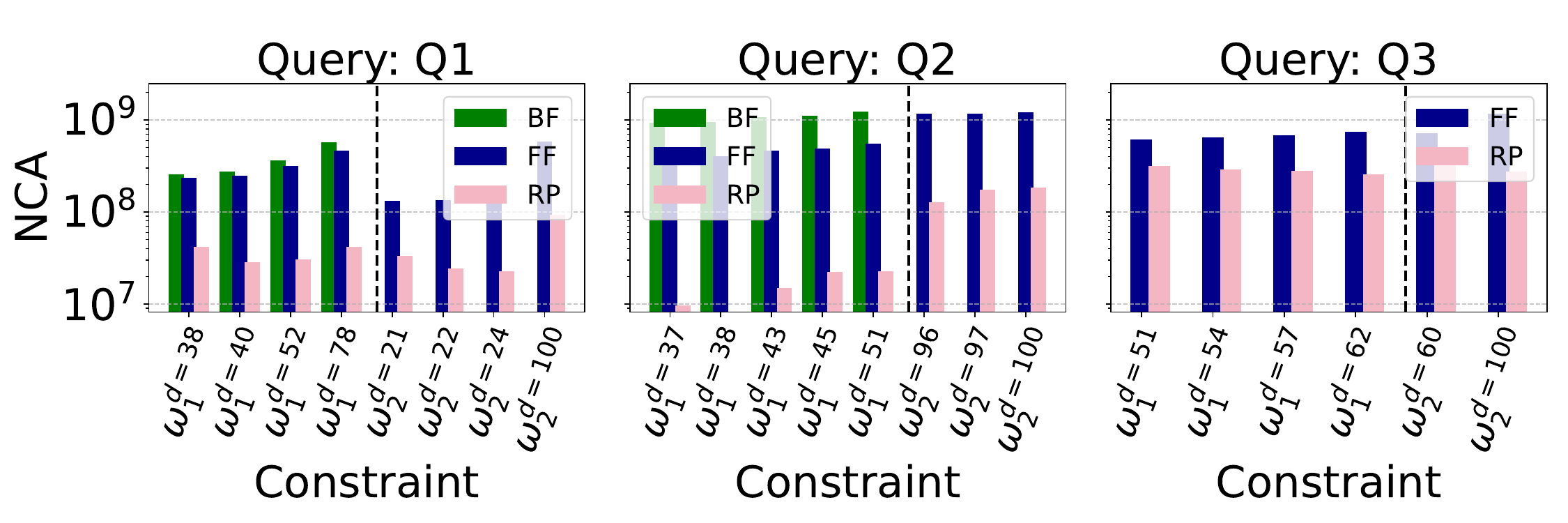}
        \caption{Total number of cluster accessed (NCA) - \gls{dshealth} dataset.}
        \label{fig:ful_range_access_comparison_healthcare}
    \end{subfigure}
    \hfill
    \begin{subfigure}[t]{0.49\linewidth}
        \centering
        \includegraphics[width=\linewidth, trim=0 30 0 0]{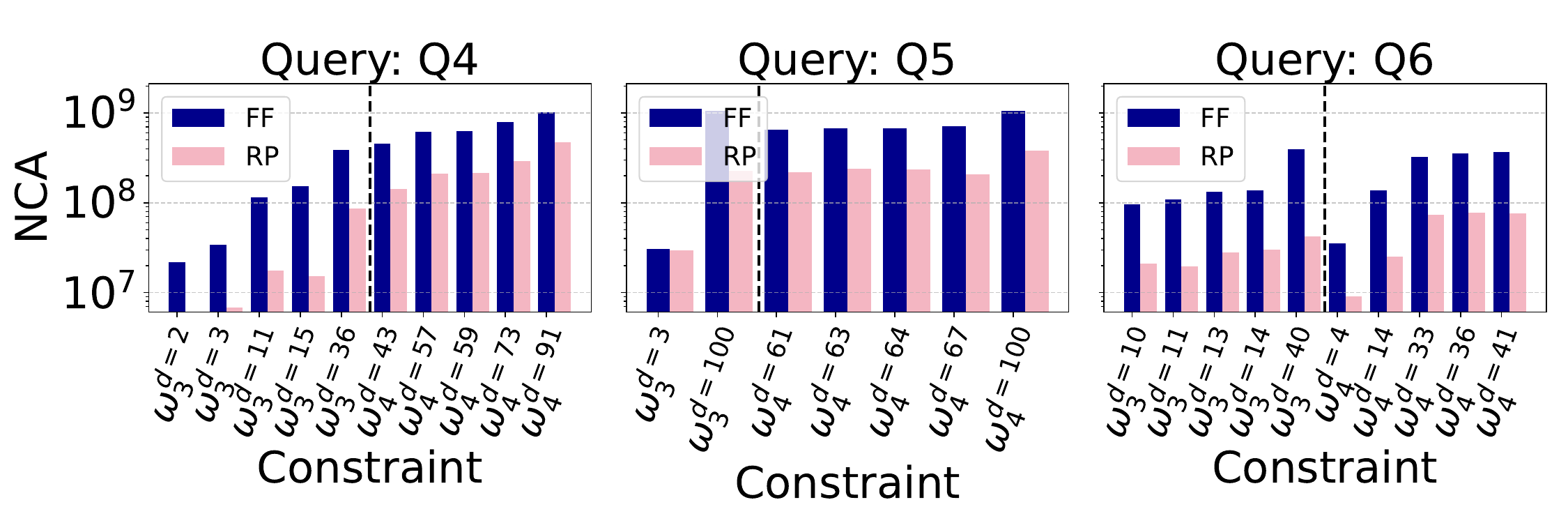}
        \caption{Total number of cluster accessed (NCA) - \gls{dsacs} dataset.}
        \label{fig:ful_range_access_comparison_ACSIncome}
    \end{subfigure}
    \end{minipage}
    \vspace{-4mm}
    \caption{
    Runtime, \gls{NCE}, and \gls{NCA} for \gls{algff} and \gls{algrp} over the \gls{dshealth} and \gls{dsacs} datasets using the queries from~\Cref{tab:queries}.
    }
    \label{fig:checked_comparison_proposed}
\end{figure*}

\subsection{Performance of \glsfmtshort{algff} and \glsfmtshort{algrp}}
\label{subsec: compare proposed techniques}
We compare \gls{algff} and \gls{algrp} using datasets \gls{dshealth} (\glspl{AC} $\propconstr_{1}$ and $\propconstr_{2}$ from \Cref{tab:constraints}) and \gls{dsacs} (\glspl{AC} $\propconstr_{3}$ and $\propconstr_{4}$) with the default parameter settings and queries $Q_1$, $Q_2$, and $Q_3$ (\Cref{tab:queries}). 
In addition to runtime, we also measure \glsfmtfull{NCE} which is the total of number of repair candidates for which we evaluate the \gls{AC} and \glsfmtfull{NCA} which is the total number of clusters accessed by an algorithm.

\iftechreport{
\mypar{Comparison with Brute Force} \label{subsec: cpmare with brute force}
We first compare \gls{algff} and \gls{algrp} with the \gls{algbf} method on the \gls{dshealth} dataset using
queries $Q_1$ and $Q_2$, \ifnottechreport{and} the constraint $\propconstr_{1}$ and default settings in~\Cref{sec:setup}.
As expected, both \gls{algff} and \gls{algrp} outperform \gls{algbf} by at least one order of magnitude in terms of runtime as shown in~\Cref{fig:runtime_comparison}.
The \gls{algrp} algorithm significantly reduces both \gls{NCE} and \gls{NCA}, while the \gls{algff} method maintains the same \gls{NCE} as \gls{algbf} but decreases the \gls{NCA} compared to \gls{algbf} (as \gls{algbf} does not use clusters we count tuple accesses) as in~\Cref{fig:access_comparison} and~\Cref{fig:checked_comparison}.
}

\mypar{Runtime}
\label{par:Run Time}
\Cref{fig:ful_range_runtime_comparison_healthcare,fig:ful_range_runtime_comparison_ACSIncome} show the runtime of the \gls{algff} and \gls{algrp} algorithms for \gls{dshealth} and \gls{dsacs}, respectively. For a subset of the experiments we also report results for the \gls{algbf} method.
For given constraint $\propconstr_i$ we vary the bounds $[B_l,B_u]$ to control what percentage of repair candidates have to be processed by the algorithms to determine the top-$k$ repairs as explained above. For example, $\acperc{1}{38}$ in~\Cref{fig:ful_range_runtime_comparison_healthcare} for $Q_1$ is the constraint $\propconstr_{1}$ from~\Cref{tab:constraints} with the bounds set such that 38\% of the candidate solutions have to be explored. We refer to this as the \gls{ED}.
As expected, both \gls{algff} and \gls{algrp} outperform \gls{algbf} by at least one order of magnitude in terms of runtime as shown in~\Cref{fig:ful_range_runtime_comparison_healthcare}.
The \gls{algrp} algorithm significantly reduces both \gls{NCE} and \gls{NCA}, while the \gls{algff} method maintains the same \gls{NCE} as \gls{algbf} but decreases the \gls{NCA} compared to \gls{algbf} (as \gls{algbf} does not use clusters we count tuple accesses).
\gls{algrp} (pink bars) generally outperforms \gls{algff} (blue bars) for most settings, demonstrating an additional improvement of up to an order of magnitude due to its capability of pruning and confirming sets of candidates at once.
The only exception are settings where the top-k repairs are found by exploring a very small portion of the search space, e.g., $Q_4$ with $\propconstr_{3}$.
\iftechreport{In~\Cref{fig:ful_range_runtime_comparison_ACSIncome}, both algorithms exhibit similar performance for
$Q_4$ with $\propconstr_{3}$, where solutions are found after exploring only 2\% and 3\% of the search space.
A similar pattern is observed
 for $Q_5$ with $\acperc{3}{3}$ and $Q_6$ with $\acperc{4}{4}$.}
\iftechreport{We further investigate the relationship between 
the \gls{ED} and runtime in~\Cref{subsec: invistigation}.}

\mypar{Total \glsfmtfull{NCE}}
We further analyze how \gls{NCE} affects the performance of our methods (\Cref{fig:ful_range_checked_comparison_healthcare,fig:ful_range_checked_comparison_ACSIncome}) 
\Gls{algrp} consistently checks fewer candidates compared to \gls{algff}.
As observed in the runtime evaluation, the difference between the two algorithms is more pronounced when larger parts of the search space have to be explored.
\iftechreport{For example, as shown in~\Cref{fig:ful_range_checked_comparison_ACSIncome} for $Q_4$ with $\acperc{3}{2}$ and $Q_5$ with $\acperc{3}{3}$.}

\mypar{Total Number of Cluster Accessed (NCA)}
The results for the number of clusters accessed are shown in ~\Cref{fig:ful_range_access_comparison_healthcare,fig:ful_range_access_comparison_ACSIncome} for \gls{dshealth} and \gls{dsacs}, respectively. Similar to the result for \gls{NCE}, \gls{algrp} accesses significantly fewer clusters than \gls{algff}. 


\subsection{
Performance-Impacting Factors}
\label{subsec: invistigation}

To gain deeper insights into the behavior observed in~\Cref{subsec: compare proposed techniques}, we  investigate the relationship between the \glsfmtfull{ED} and performance.
%
We also evaluate the performance of \gls{algff} and \gls{algrp} in terms of the parameters
from~\Cref{sec:setup}.

\ifnottechreport{
\begin{figure}[t]
    \centering \vspace{-8mm}
    \begin{subfigure}[t]{0.98\linewidth}
        \centering
        \includegraphics[width=\linewidth]{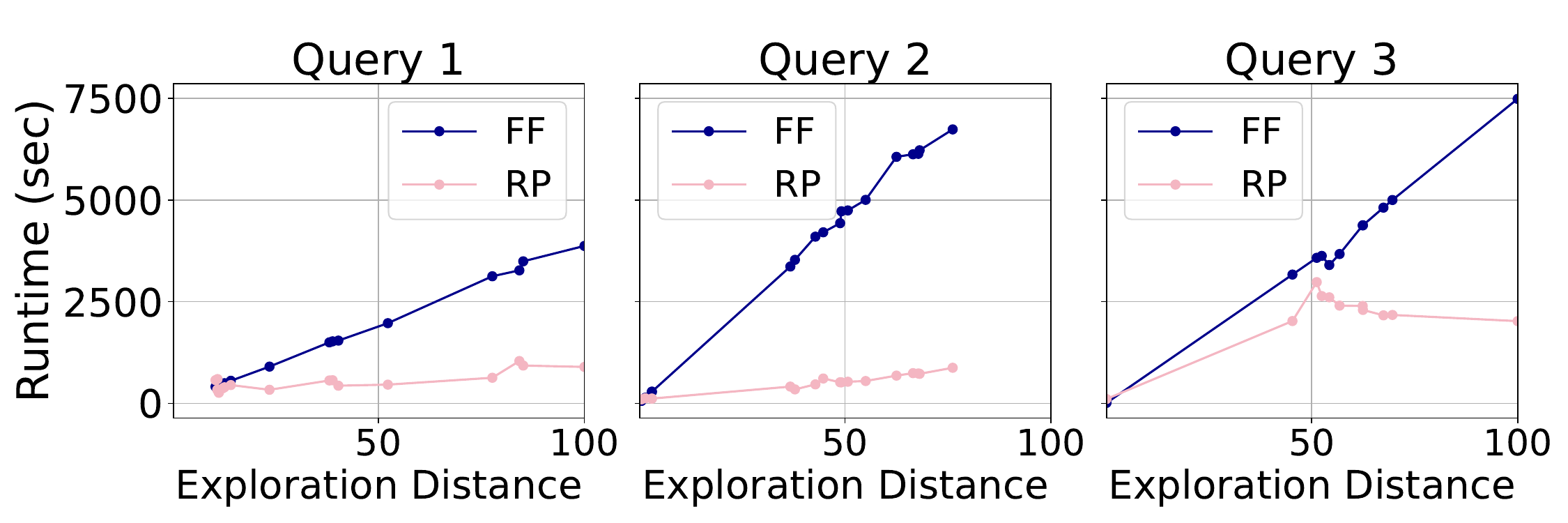}
    \end{subfigure}
    \vspace{-2mm}
    \caption{Runtime of \gls{algff} and \gls{algrp}, varying \gls{ED}.
    }
    \label{fig:healthcare-proximity}
\end{figure}
}
\iftechreport{
\begin{figure*}[htbp]
    \centering
    \begin{minipage}{1\linewidth}
    \begin{subfigure}[t]{0.49\linewidth}
        \centering
        \includegraphics[width=\linewidth, trim=0 30 0 0]{Time_vs_constraint_distance_combined_Healthcare.pdf}
        \caption{Runtime of \gls{algff} and \gls{algrp}, varying \gls{ED} - \gls{dshealth} dataset.}
        \label{fig:healthcare-proximity-Runtime-Healthcare}
    \end{subfigure}
    \begin{subfigure}[t]{0.49\linewidth}
        \centering
        \includegraphics[width=\linewidth, trim=0 30 0 0]{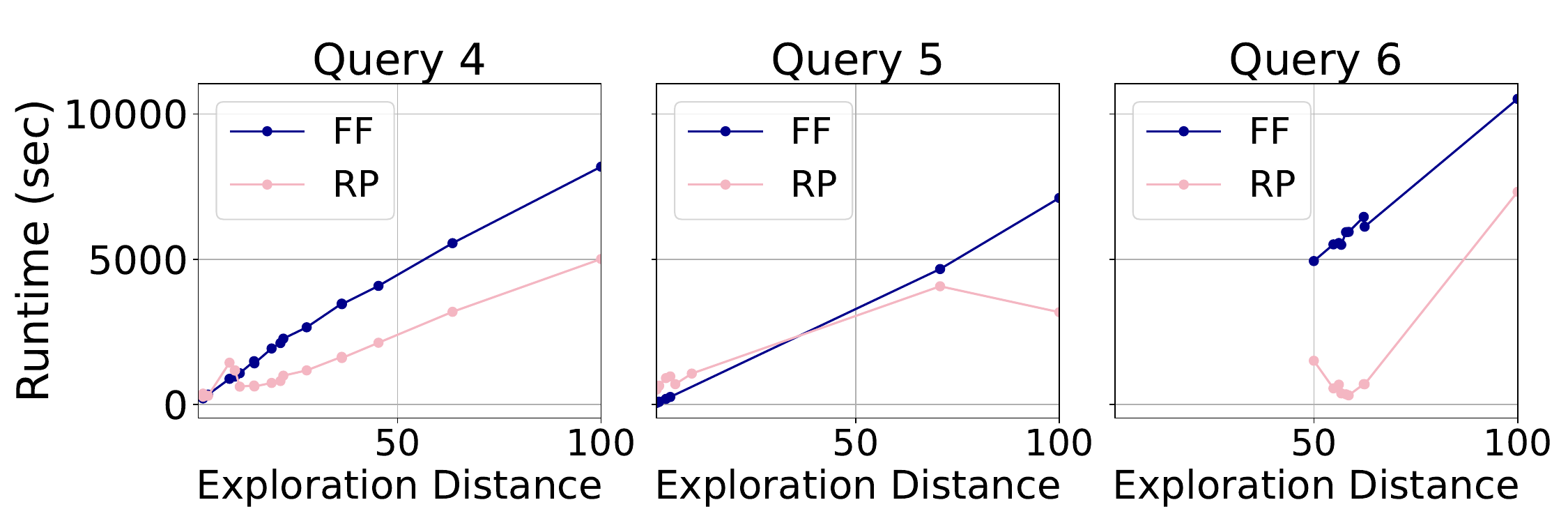}
        \caption{Runtime of \gls{algff} and \gls{algrp}, varying \gls{ED} - \gls{dsacs} dataset.}
        \label{fig:healthcare-proximity-Runtime-ACSIncome}
    \end{subfigure}
    \vspace{-2mm}
    \end{minipage}\\
    \begin{minipage}{1\linewidth}
    \begin{subfigure}[t]{0.49\linewidth}
        \centering
        \includegraphics[width=\linewidth, trim=0 30 0 0]{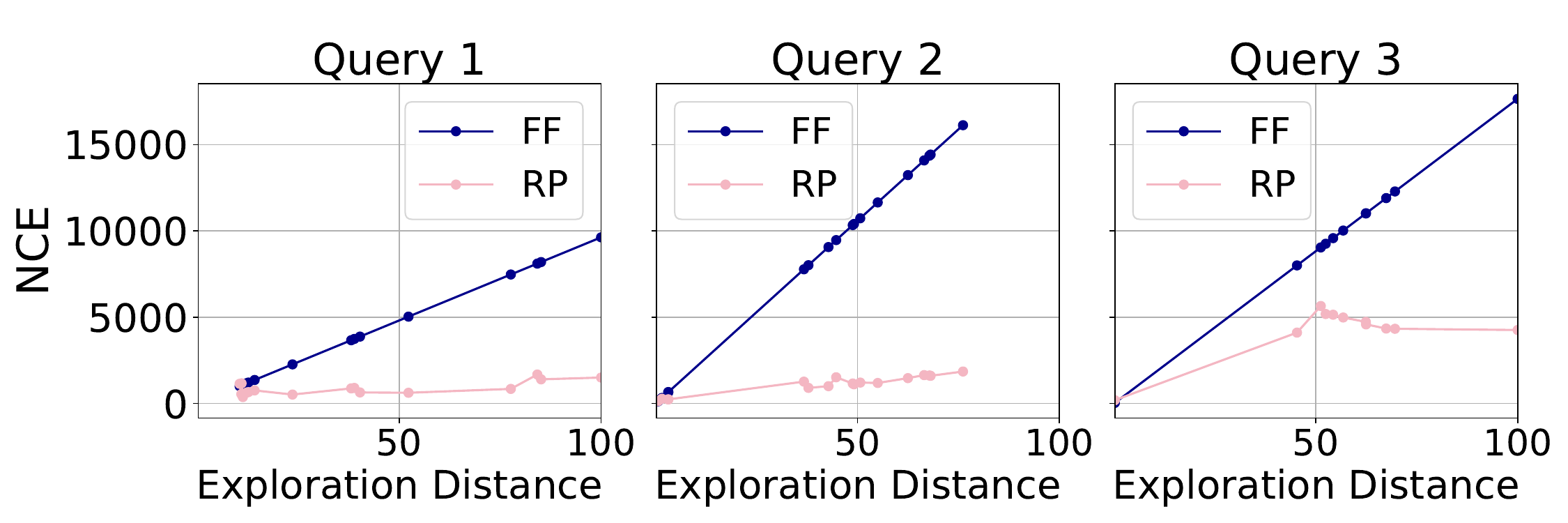}
        \caption{Total number of constraints evaluated (NCE) - \gls{dshealth} dataset.}
        \label{fig:healthcare-proximity-NCE-Healthcare}
    \end{subfigure}
    \hfill
    \begin{subfigure}[t]{0.49\linewidth}
        \centering
        \includegraphics[width=\linewidth, trim=0 30 0 0]{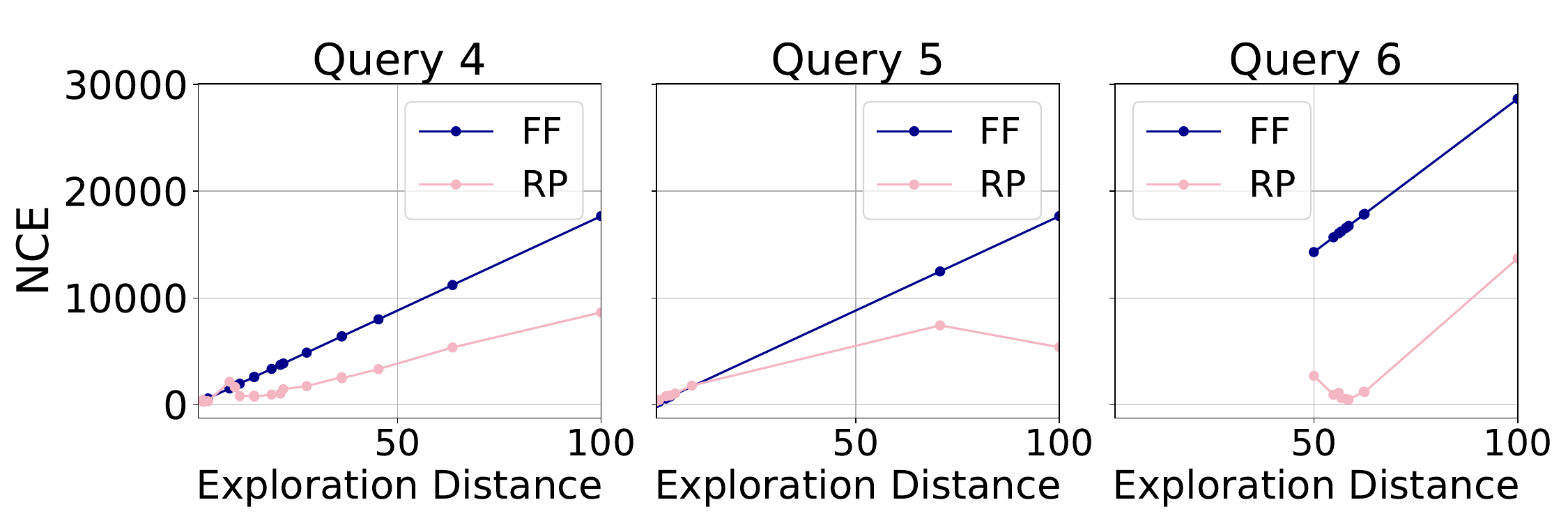}
        \caption{Total number of constraints evaluated (NCE) - \gls{dsacs} dataset.}
        \label{fig:healthcare-proximity-NCE-ACSIncome}
    \end{subfigure}
    \vspace{-2mm}
    \end{minipage}\\
    \begin{minipage}{1\linewidth}
    \begin{subfigure}[t]{0.49\linewidth}
        \centering
        \includegraphics[width=\linewidth, trim=0 30 0 0]{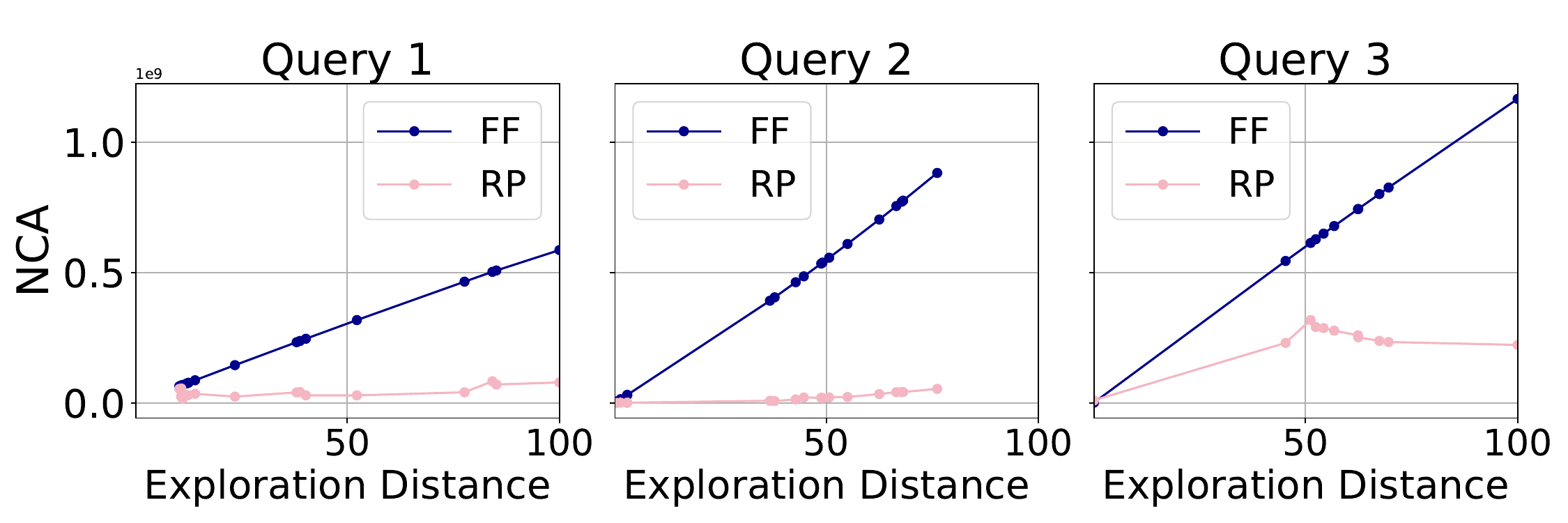}
        \caption{Total number of cluster accessed (NCA) - \gls{dshealth} dataset.}
        \label{fig:healthcare-proximity-NCA-Healthcare}
    \end{subfigure}
    \hfill
    \begin{subfigure}[t]{0.49\linewidth}
        \centering
        \includegraphics[width=\linewidth, trim=0 30 0 0]{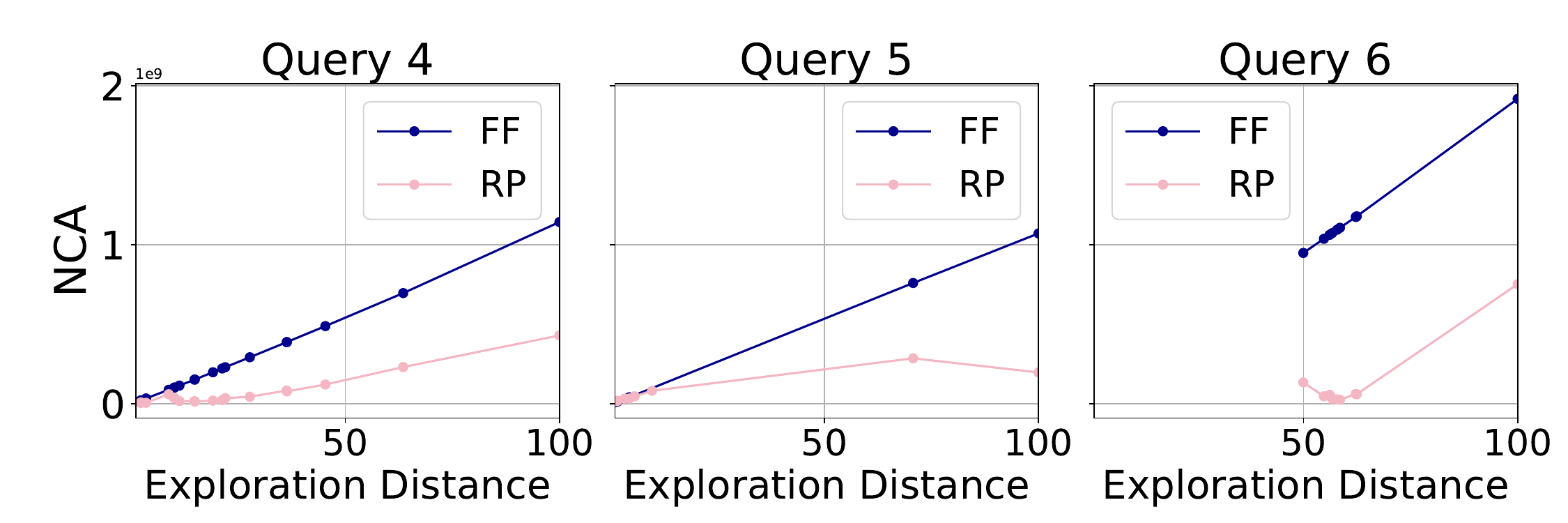}
        \caption{Total number of cluster accessed (NCA) - \gls{dsacs} dataset.}
        \label{fig:healthcare-proximity-NCA-ACSIncome}
    \end{subfigure}
    \end{minipage}
    \vspace{-2mm}
    \caption{
    Runtime\iftechreport{, \gls{NCE}, and \gls{NCA} for \gls{algff} and \gls{algrp} over the \gls{dshealth} and \gls{dsacs} datasets,} varying \gls{ED}.
    }
    \label{fig:checked_comparison_proposed}
\end{figure*}
}
\mypar{Effect of \Glsfmtlong{ED}}
\captionsetup[figure]{skip=5pt} 
We use queries $Q_1$–$Q_3$ and the constraint $\propconstr_{1}$ on \gls{dshealth} \iftechreport{and $Q_4$–$Q_6$ and the constraint $\propconstr_{3}$ on \gls{dsacs}} and vary the bounds to control for \gls{ED}.
The result is shown in\ifnottechreport{~\Cref{fig:healthcare-proximity}}\iftechreport{~\Cref{fig:healthcare-proximity-Runtime-Healthcare} for \gls{dshealth}} \iftechreport{and in~\Cref{fig:healthcare-proximity-Runtime-ACSIncome} for \gls{dsacs}}.
\ifnottechreport{We present results for \gls{dsacs} in~\cite{AG25techrepo}.}
For $Q_1$ and $Q_2$, when \gls{ED} $10\%$ or less, \gls{algff} and \gls{algrp} exhibit comparable performance.
A similar pattern is seen for $Q_3$, where \gls{algff} performs better than \gls{algrp} for very low \gls{ED} 
\iftechreport{as shown in~\Cref{fig:healthcare-proximity-Runtime-Healthcare}}.
\iftechreport{The same trend holds for $Q_4$ and $Q_5$, while for $Q_6$, \gls{algrp} consistently outperforms \gls{algff} for higher \gls{ED}, as illustrated in~\Cref{fig:healthcare-proximity-Runtime-ACSIncome}.}
\iftechreport{The \gls{NCE} and \gls{NCA} follow similar patterns to runtime. For \gls{ED} $> 50\%$, \gls{algrp} significantly reduces both \gls{NCE} and \gls{NCA}. However, when \gls{ED} $< 10\%$, the difference between the two algorithms diminishes, with both performing similarly. These trends are shown in~\Cref{fig:healthcare-proximity-NCA-Healthcare} and~\Cref{fig:healthcare-proximity-NCA-ACSIncome} for \gls{NCA}, and in~\Cref{fig:healthcare-proximity-NCE-Healthcare} and~\Cref{fig:healthcare-proximity-NCE-ACSIncome} for \gls{NCE}.}
The reason behind this trend is that when solutions are closed to the user query (smaller \gls{ED}), then there is less opportunity for pruning for \gls{algrp}.

\ifnottechreport{
\begin{figure}[t]
    \centering \vspace{-6.5mm}
    \begin{subfigure}[t]{0.47\linewidth}
        \centering
        \includegraphics[width=\linewidth, trim=0 40 0 0]{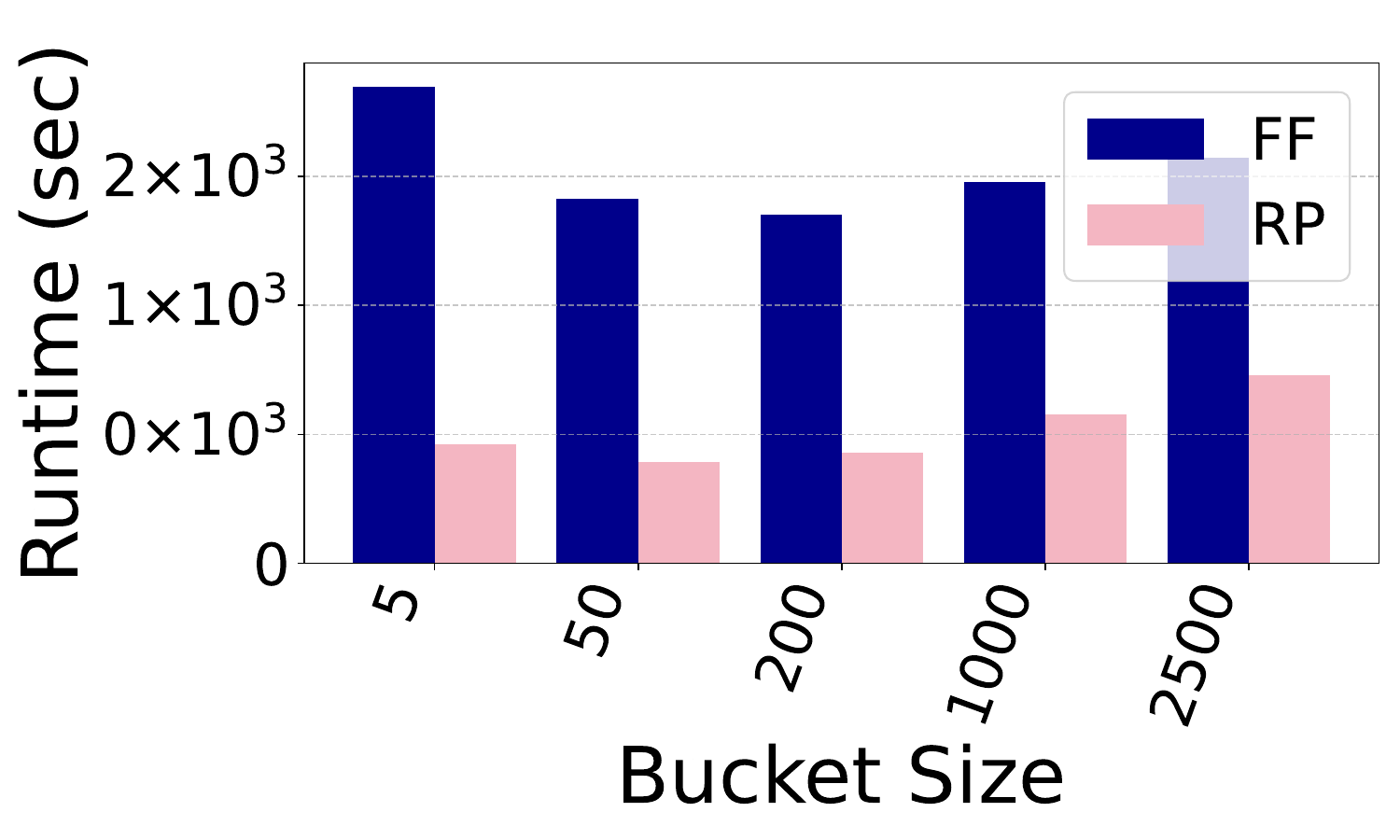}
        \caption{\gls{dshealth} dataset}
        \label{fig:BucketSize_Time_Healthcare}
    \end{subfigure}
    \vspace{-3mm}
    \begin{subfigure}[t]{0.47\linewidth}
        \centering
        \includegraphics[width=\linewidth, trim=0 40 0 0]{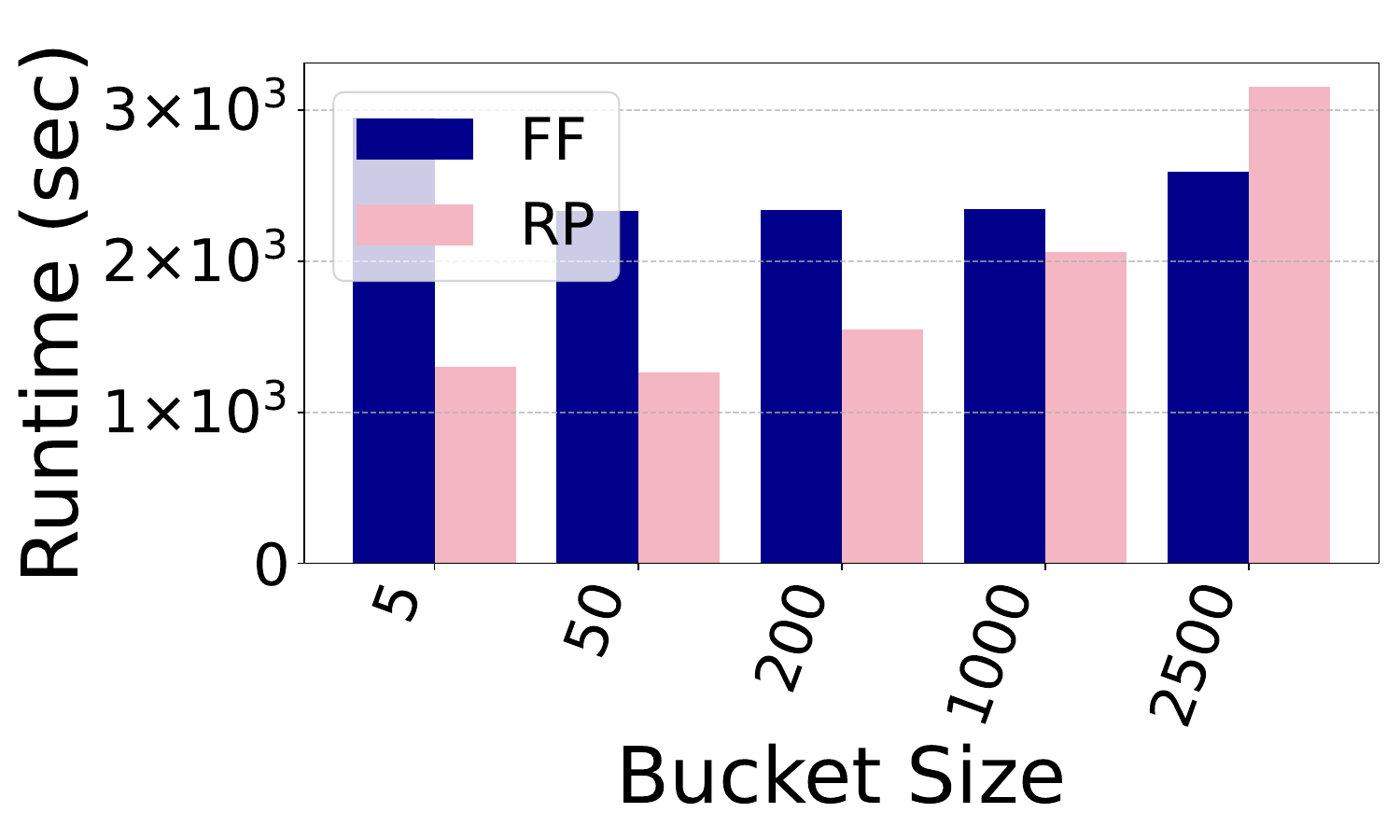}
        \caption{\gls{dsacs} dataset}
        \label{fig:BucketSize_Time_ACSIncome}
    \end{subfigure}
    \vspace{-1mm}
    \caption{Runtime, varying bucket size \bucketsize.
    }
    \label{fig:BucketSize_comparison}
  \end{figure}
}

\iftechreport{
\begin{figure}[t]
    \centering \vspace{-6.5mm}
    \begin{subfigure}[t]{0.47\linewidth}
        \centering
        \includegraphics[width=\linewidth, trim=0 40 0 0]{comparison_fully_vs_ranges_Time_Healthcare_BucketSize.pdf}
        \caption{Runtime - \gls{dshealth} dataset.}
        \label{fig:BucketSize_Time_Healthcare}
    \end{subfigure}
    \begin{subfigure}[t]{0.47\linewidth}
        \centering
        \includegraphics[width=\linewidth, trim=0 40 0 0]{comparison_fully_vs_ranges_Time_ACSIncome_BucketSize.pdf}
        \caption{Runtime - \gls{dsacs} dataset.}
        \label{fig:BucketSize_Time_ACSIncome}
    \end{subfigure}
    \vspace{-3mm}
    \begin{subfigure}[t]{0.47\linewidth}
        \centering
        \includegraphics[width=\linewidth, trim=0 40 0 0]{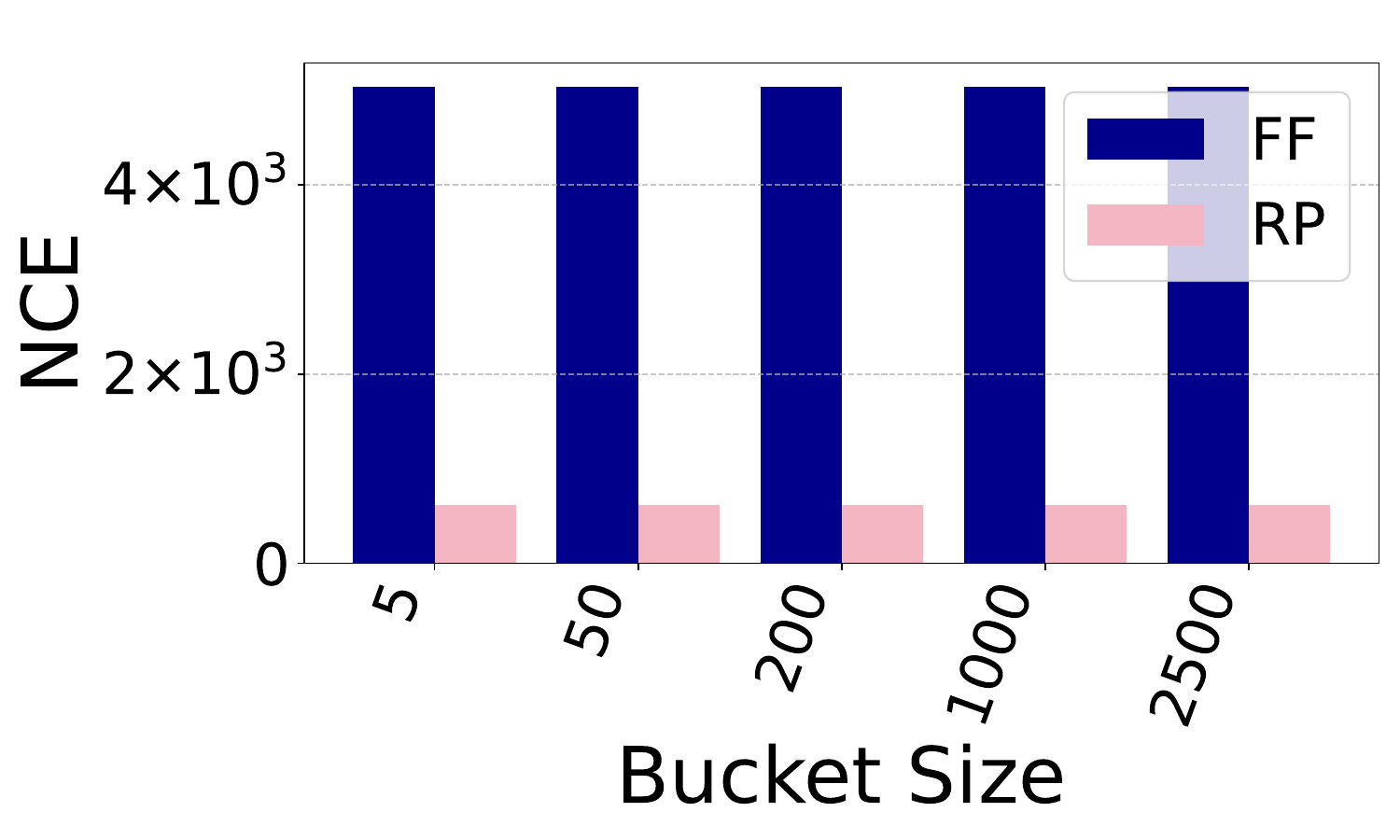}
        \caption{Total number of constraints evaluated (NCE) - \gls{dshealth} dataset.}
        \label{fig:BucketSize_NCE_Healthcare}
    \end{subfigure}
    \begin{subfigure}[t]{0.47\linewidth}
        \centering
        \includegraphics[width=\linewidth, trim=0 40 0 0]{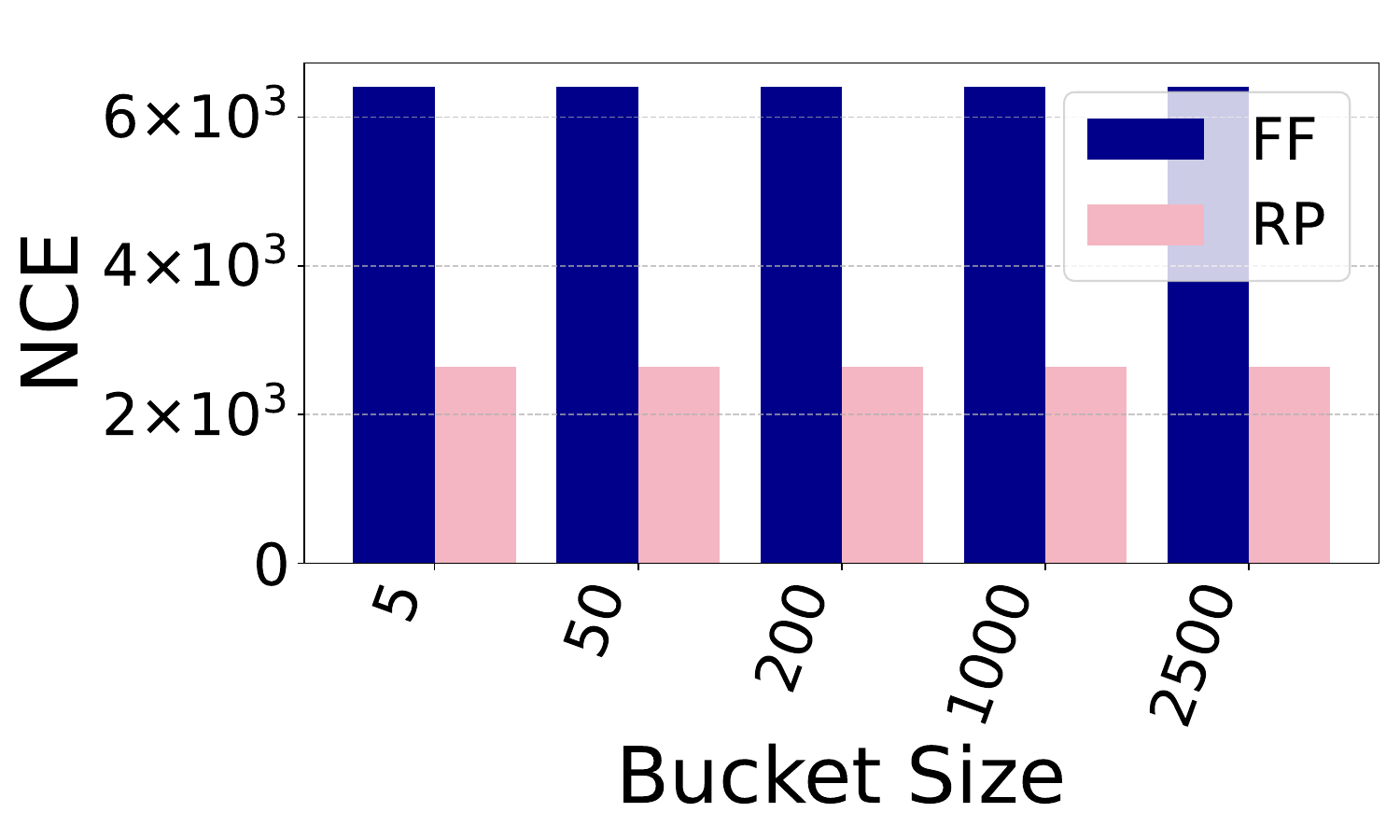}
        \caption{Total number of constraints evaluated (NCE) - \gls{dsacs} dataset.}
        \label{fig:BucketSize_NCE_ACSIncome}
    \end{subfigure}
    \vspace{-3mm}
    \begin{subfigure}[t]{0.47\linewidth}
        \centering
        \includegraphics[width=\linewidth, trim=0 40 0 0]{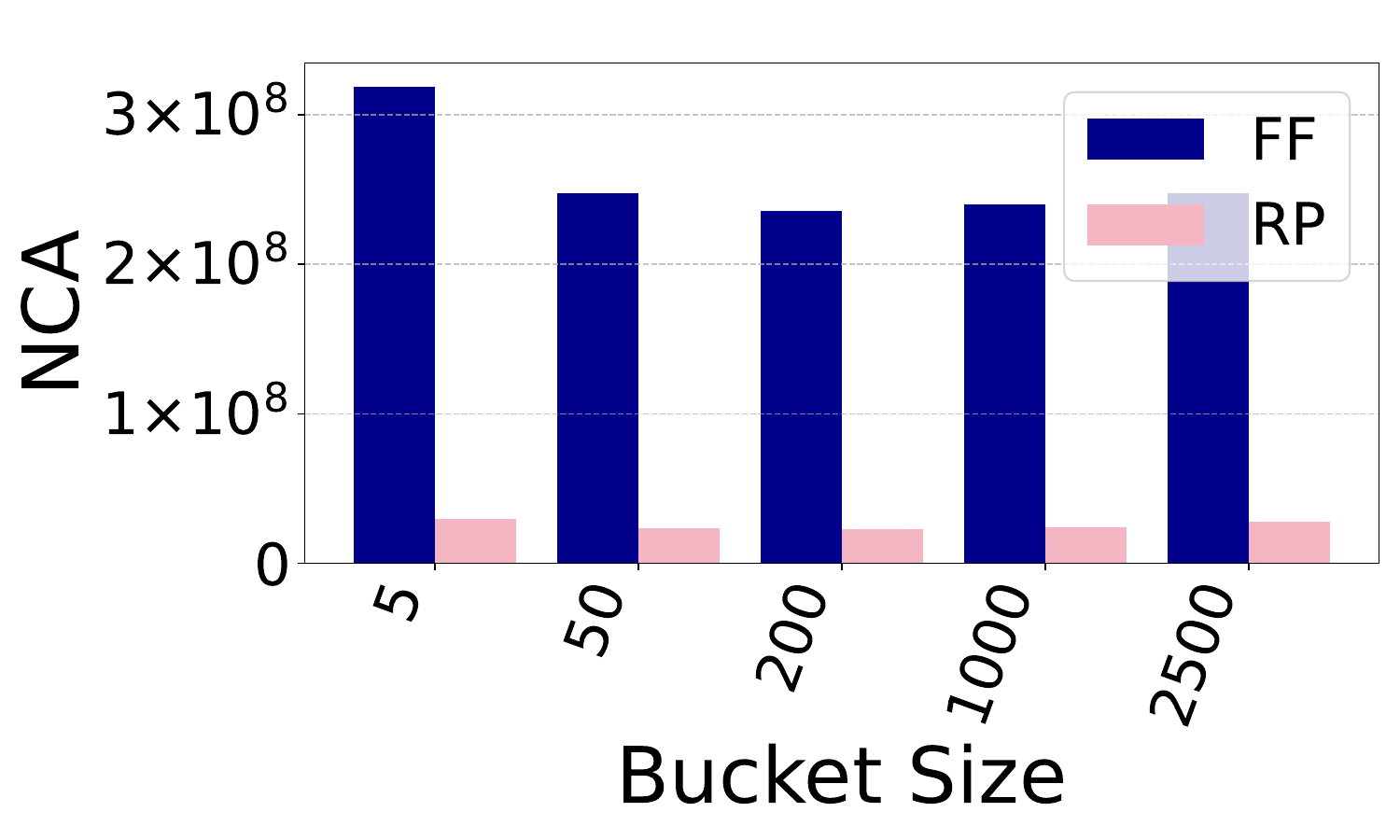}
        \caption{Total number of cluster accessed (NCA) - \gls{dshealth} dataset.}
        \label{fig:BucketSize_NCA_Healthcare}
    \end{subfigure}
    \begin{subfigure}[t]{0.47\linewidth}
        \centering
        \includegraphics[width=\linewidth, trim=0 40 0 0]{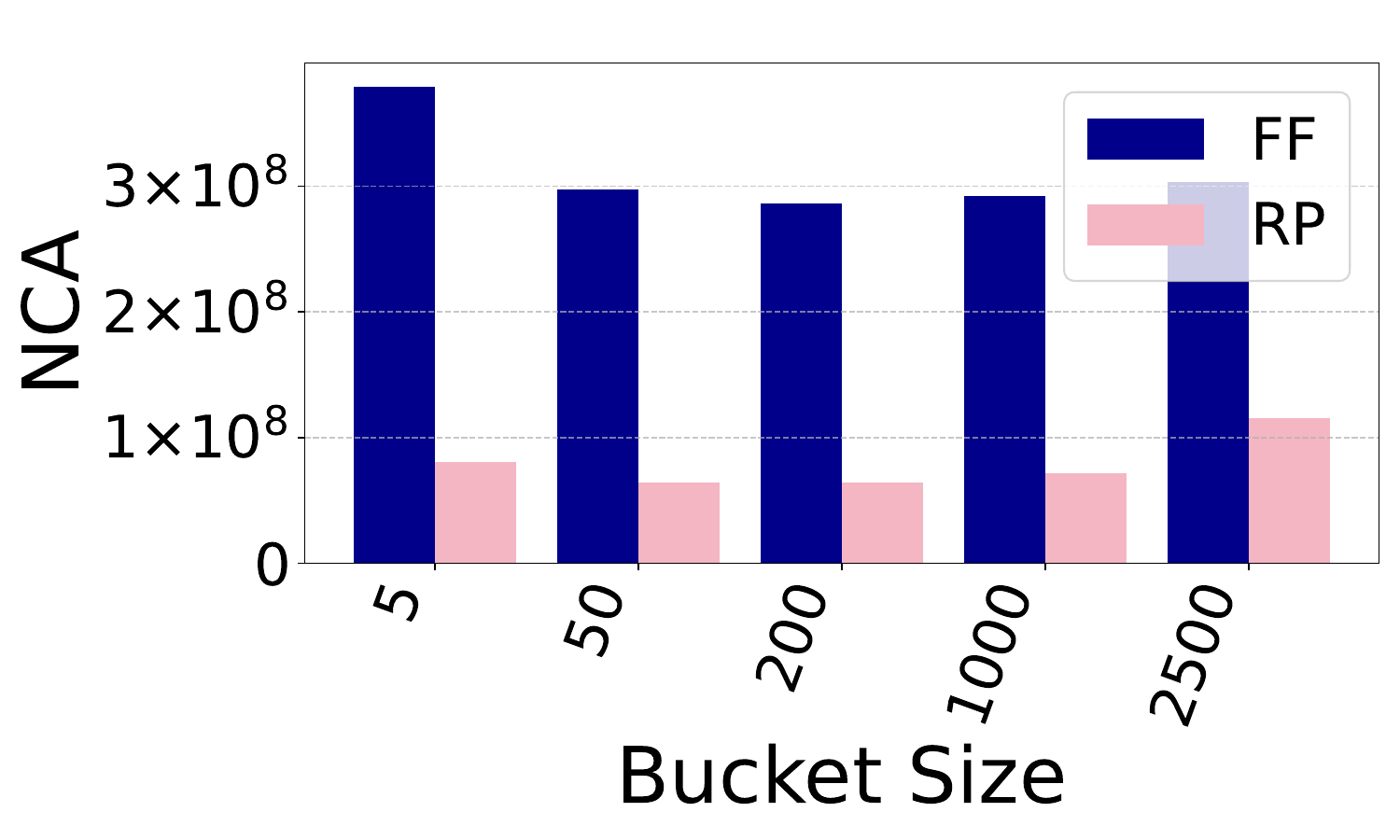}
        \caption{Total number of cluster accessed (NCA) - \gls{dsacs} dataset.}
        \label{fig:BucketSize_NCA_ACSIncome}
    \end{subfigure}
    \vspace{-1mm}
    \caption{Runtime, \gls{NCE}, and \gls{NCA} for \gls{algff} and \gls{algrp} over the \gls{dshealth} and \gls{dsacs} datasets, varying bucket size \bucketsize.}
    \label{fig:BucketSize_comparison}
  \end{figure}
}
\mypar{Effect of Bucket Size}
\label{subsub: effect of bucketSize}
We now evaluate the runtime of \gls{algff} and \gls{algrp} varying the bucket size \bucketsize using $Q_1$ with $\propconstr_{1}$ with bounds $[0.44, 0.5]$ for the \gls{dshealth} dataset and $Q_4$ with $\propconstr_{3}$ using bounds $[0.34, 0.39]$ for the \gls{dsacs} dataset.
We vary the \bucketsize from 5 to 2500.
We use the default branching factor $\branchfactor=5$.
\iftechreport{For this branching factor the
the structure of the \tree for 50k tuples is
as follows:
(i) Level 1: 5 clusters, each with 10,000 data points;
(ii) Level 2: 25 clusters, each with 2,000 data points;
(iii) Level 3: 125 clusters, each with 400 data points;
(iv) Level 4: 525 clusters, each with 80 data points;
(v) Level 5: 3,125 clusters, each with 16 data points;
(vi) Level 6: 15,625 clusters, each with 3 or 4 data points.}
Our algorithm chooses the number of levels to ensure that the size of leaf clusters is $\leq \bucketsize$.
For example, for $\bucketsize= 200$, the tree will have 4 levels.
%
The results of \ifnottechreport{this experiment}\iftechreport{the runtime} are shown in~\ifnottechreport{\Cref{fig:BucketSize_comparison}} \iftechreport{\Cref{fig:BucketSize_Time_Healthcare} and \Cref{fig:BucketSize_Time_ACSIncome}}.\iftechreport{Similarly, the \gls{NCA} as shown in~\Cref{fig:BucketSize_NCA_Healthcare} and~\Cref{fig:BucketSize_NCA_ACSIncome} exhibit the same trend as the runtime.} The advantage of smaller bucket sizes is that it is more likely that we can find a cluster that is fully covered / not covered at all. However, this comes at the cost of having to explore more clusters.
\iftechreport{For \gls{NCE}, as shown in~\Cref{fig:BucketSize_NCE_Healthcare} and~\Cref{fig:BucketSize_NCE_ACSIncome}, the number of constraints evaluated remains constant across different bucket sizes \bucketsize. This is because the underlying data remains the same, and varying \bucketsize does not affect the set of constraints that need to be evaluated.}
In preliminary experiments, we have identified $\bucketsize=15$ to yield robust performance for a wide variety of settings and use this as the default.

\begin{table}[H]
    \centering
    \vspace{-1mm}
    \caption{Branching Configuration and Data Distribution}
    \vspace{-2mm}
    \begin{tabular}{|c|c||c|c|}
        \hline
        \textbf{\# of Branches} & \textbf{\# of Leaves} & \textbf{\# of Branches} & \textbf{\# of Leaves} \\
        \hline
        5  & 15625  & 20 & 8000 \\
        10 & 10000  & 25 & 15625   \\
        15 & 3375   & 30 & 27000 \\
        \hline
    \end{tabular}
    \label{tab:branching_data}
\end{table}
\ifnottechreport{
\begin{figure}[t]
    \centering
    \vspace{-3mm}
    \begin{subfigure}[t]{0.47\linewidth}
        \centering
        \includegraphics[width=\linewidth, trim=0 30 0 0]{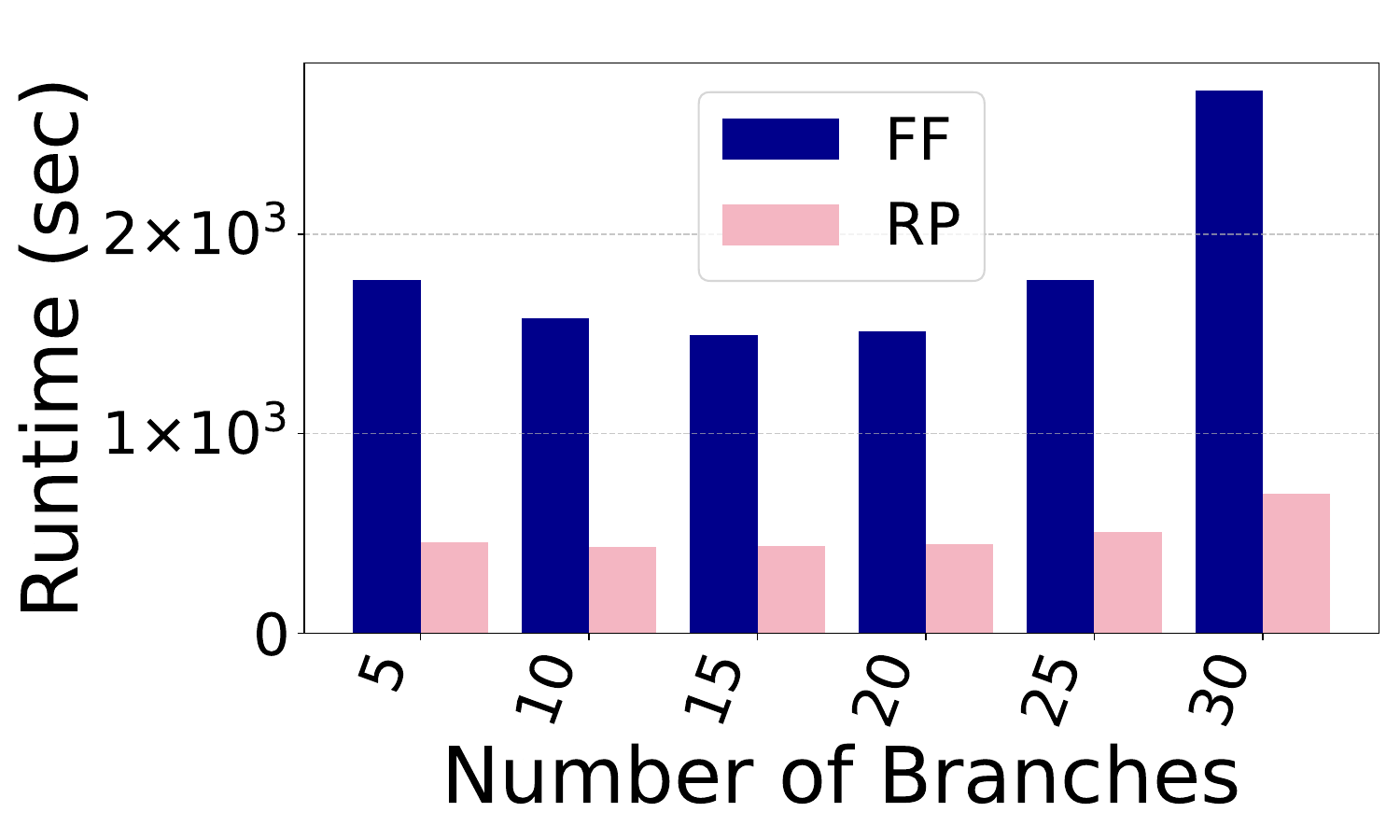}
        \caption{\gls{dshealth} dataset}
        \label{fig:BranchNum_Time_ACSIncome_Healthcare}
    \end{subfigure}
    \begin{subfigure}[t]{0.47\linewidth}
        \centering
        \includegraphics[width=\linewidth, trim=0 30 0 0]{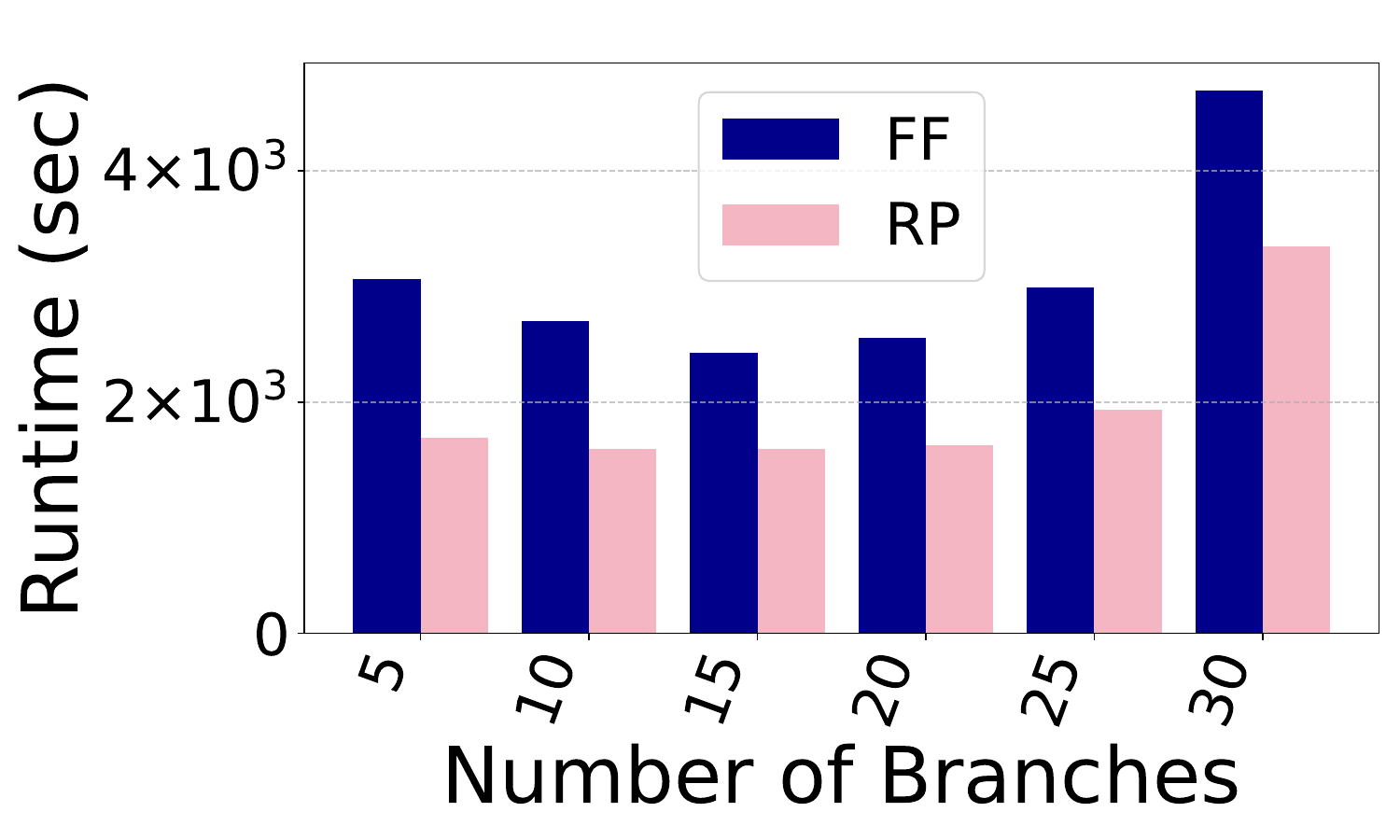}
        \caption{\gls{dsacs} dataset}
        \label{fig:BranchNum_Time_ACSIncome}
    \end{subfigure}
    \vspace{-3.5mm}
    \caption{
    Runtime, varying the number of branches \branchfactor.
    }
    \label{fig:num-of-branch-perf}
\end{figure}
}
\iftechreport{
\begin{figure}[t]
    \centering
    \vspace{-7mm}
    \begin{subfigure}[t]{0.47\linewidth}
        \centering
        \includegraphics[width=\linewidth, trim=0 30 0 0]{comparison_fully_vs_ranges_Time_Healthcare_BranchNum.pdf}
        \caption{Runtime - \gls{dshealth} dataset}
        \label{fig:BranchNum_Time_ACSIncome_Healthcare}
    \end{subfigure}
    \begin{subfigure}[t]{0.47\linewidth}
        \centering
        \includegraphics[width=\linewidth, trim=0 30 0 0]{comparison_fully_vs_ranges_Time_ACSIncome_BranchNum.pdf}
        \caption{Runtime - \gls{dsacs} dataset}
        \label{fig:BranchNum_Time_ACSIncome}
    \end{subfigure}
    \vspace{-3mm}
    \begin{subfigure}[t]{0.47\linewidth}
        \centering
        \includegraphics[width=\linewidth, trim=0 30 0 0]{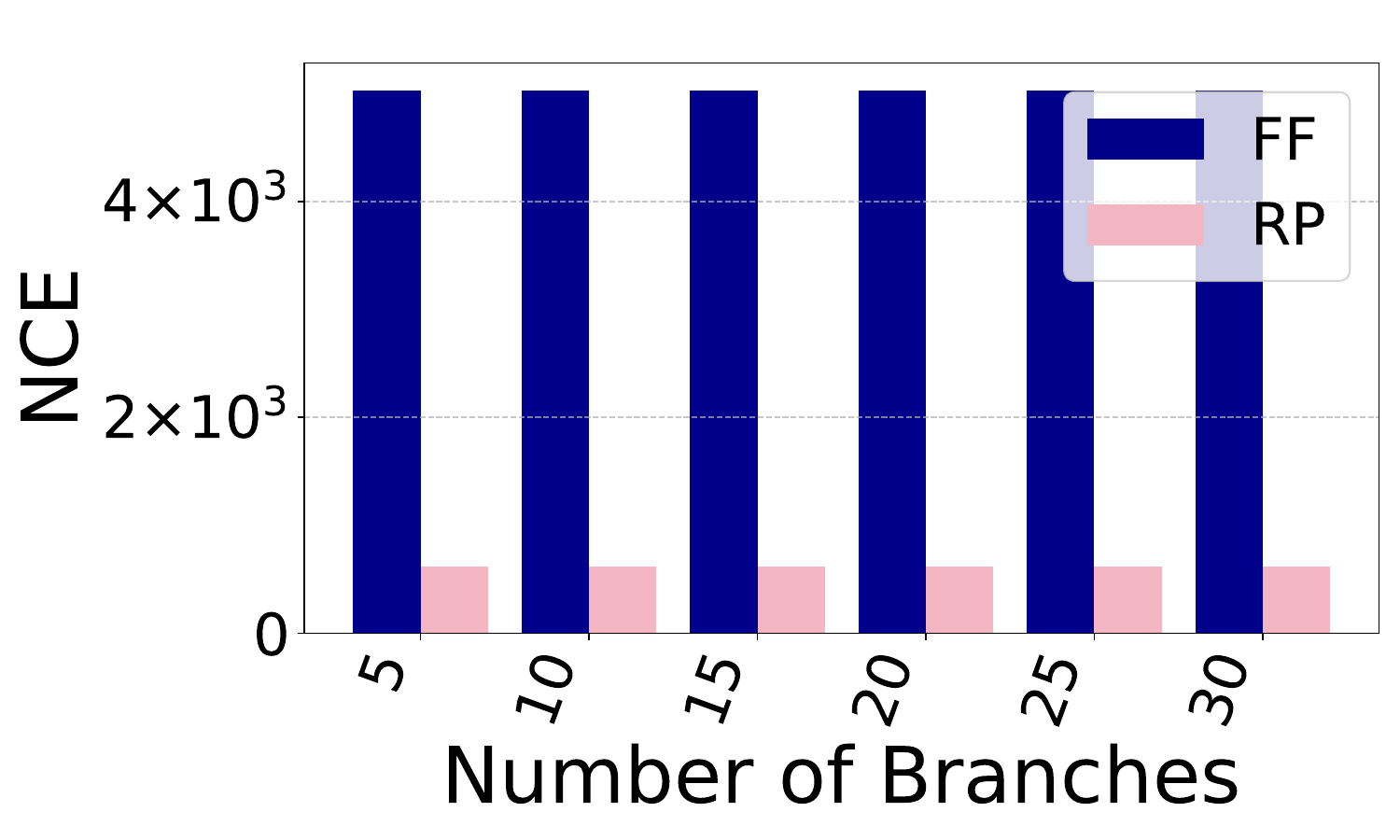}
        \caption{Total number of constraints evaluated (NCE) - \gls{dshealth} dataset}
        \label{fig:BranchNum_NCE_Healthcare}
    \end{subfigure}
    \begin{subfigure}[t]{0.47\linewidth}
        \centering
        \includegraphics[width=\linewidth, trim=0 30 0 0]{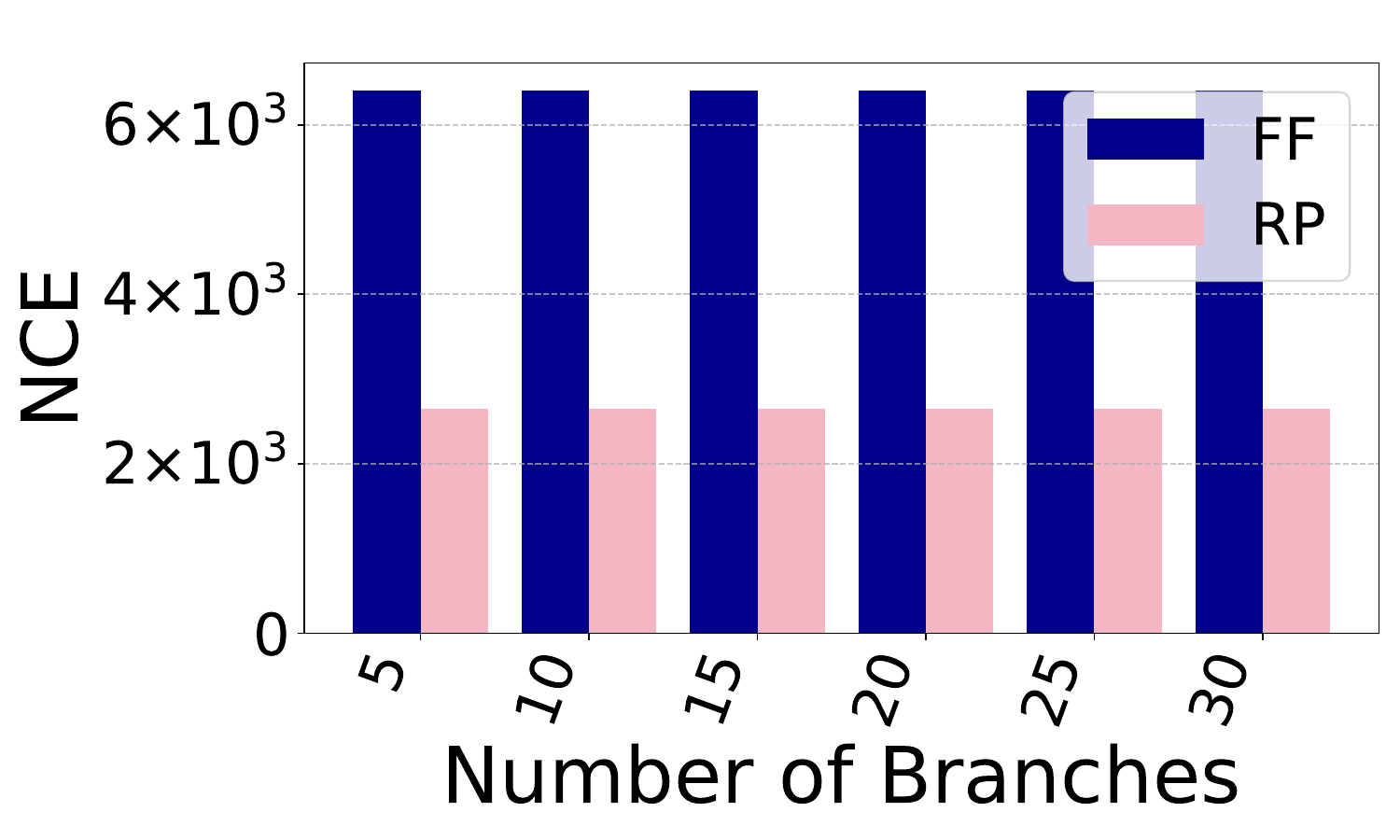}
        \caption{Total number of constraints evaluated (NCE) - \gls{dsacs} dataset}
        \label{fig:BranchNum_NCE_ACSIncome}
    \end{subfigure}
    \vspace{-3mm}
    \begin{subfigure}[t]{0.47\linewidth}
        \centering
        \includegraphics[width=\linewidth, trim=0 30 0 0]{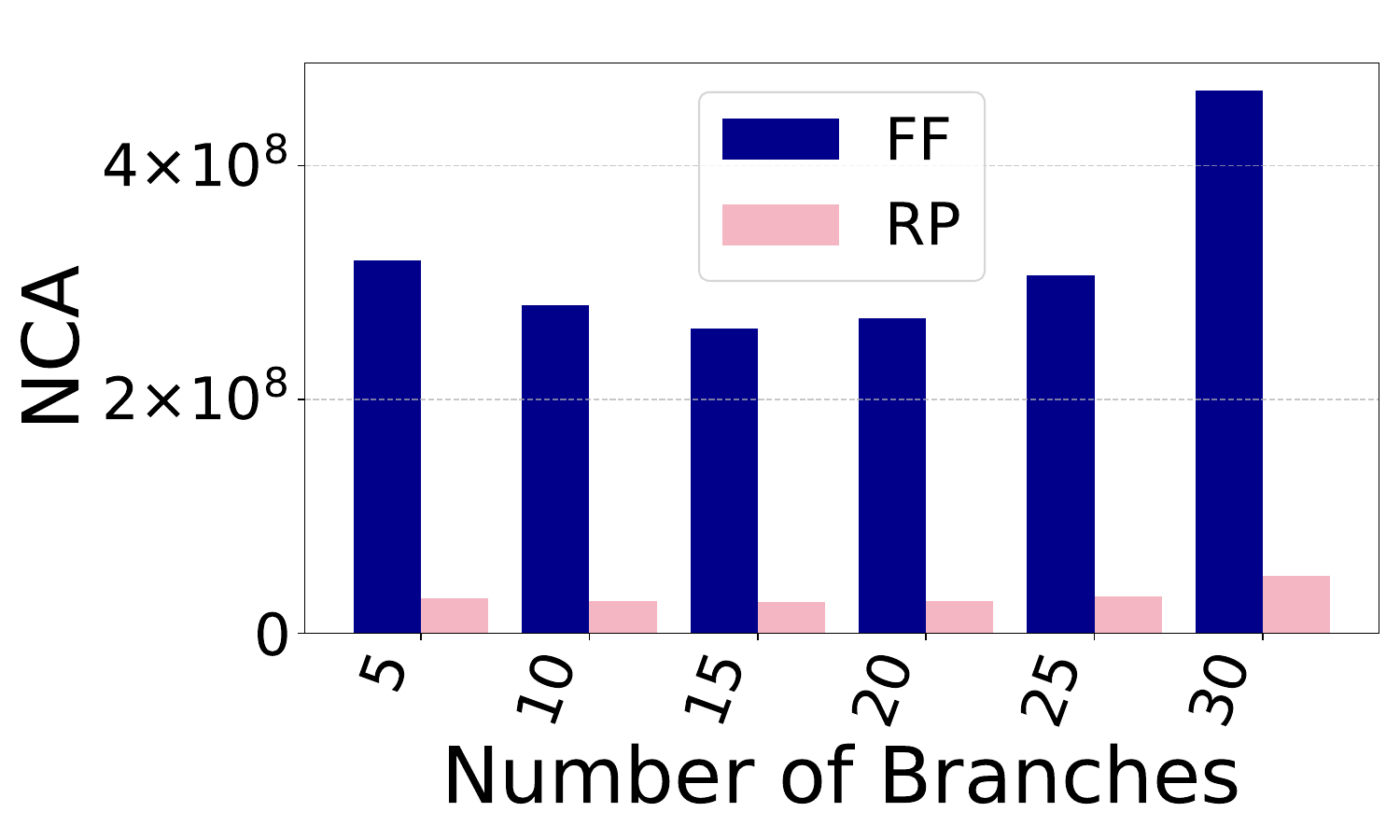}
        \caption{Total number of cluster accessed (NCA) - \gls{dshealth} dataset}
        \label{fig:BranchNum_NCA_Healthcare}
    \end{subfigure}
    \begin{subfigure}[t]{0.47\linewidth}
        \centering
        \includegraphics[width=\linewidth, trim=0 30 0 0]{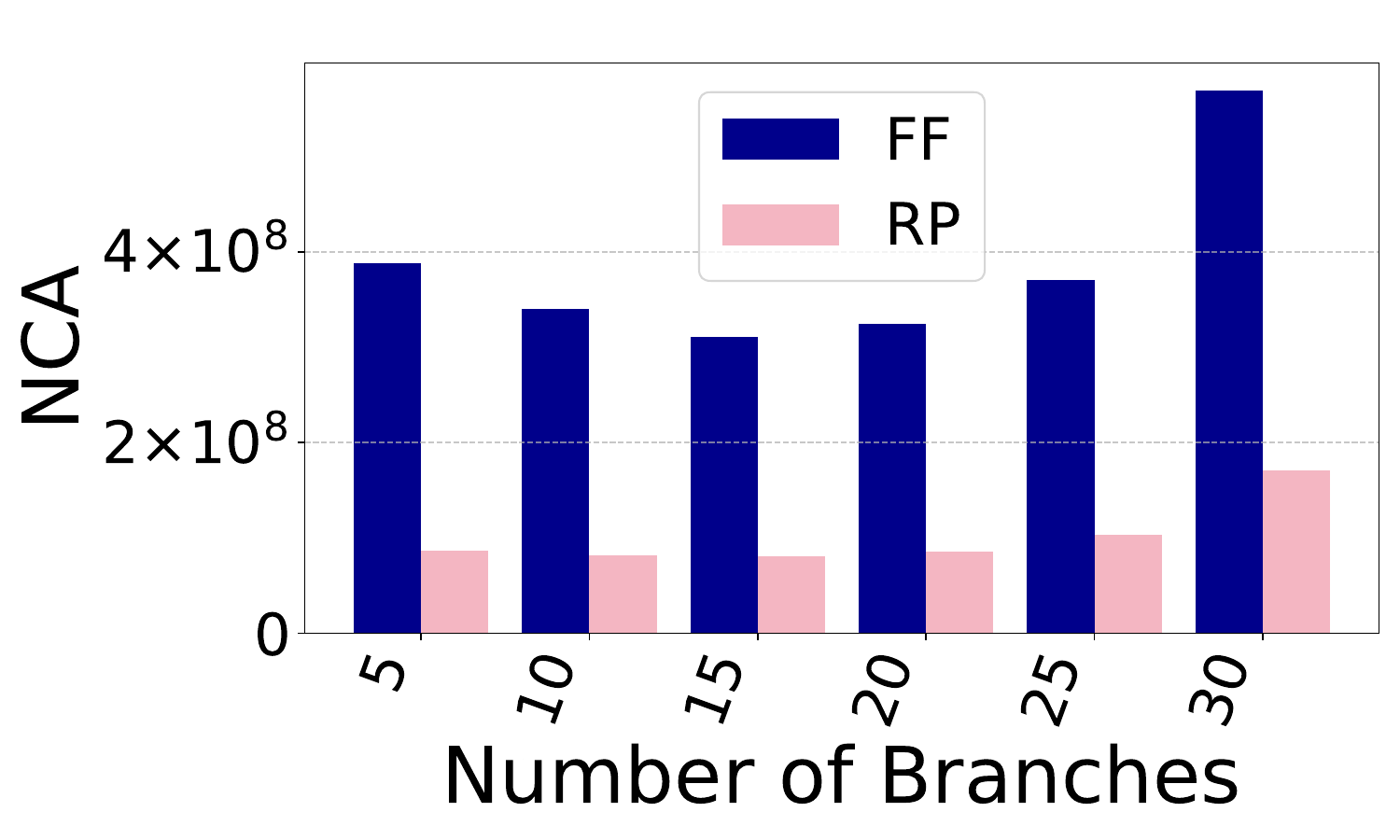}
        \caption{Total number of cluster accessed (NCA) - \gls{dsacs} dataset}
        \label{fig:BranchNum_NCA_ACSIncome}
    \end{subfigure}
    \caption{
    Runtime, \gls{NCE}, and \gls{NCA} for \gls{algff} and \gls{algrp} over the \gls{dshealth} and \gls{dsacs} datasets, varying the number of branches \branchfactor.
    }
    \label{fig:num-of-branch-perf}
\end{figure}
}

\mypar{Effect of the Branching Factor} \label{subsub:Effect of branches}
We now examines the relationship between the branching factor \branchfactor and the runtime of \gls{algff} and \gls{algrp}.
We use the same queries, constraints, bounds, and datasets as in the previous evaluation and vary the \branchfactor from 5 to 30. The corresponding number of leaf nodes in the \tree is shown in~\Cref{tab:branching_data}.
As we use the default bucket size $\bucketsize = 15$, the branching factor determines the depth of the tree.
The result shown in\ifnottechreport{~\Cref{fig:num-of-branch-perf}} \iftechreport{~\Cref{fig:BranchNum_Time_ACSIncome_Healthcare} and~\Cref{fig:BranchNum_Time_ACSIncome}} confirms that, as expected, the performance of \gls{algff} and \gls{algrp} correlates with the number of clusters at the leaf level. For \gls{algff}, 
branching factors of 5 and 25 yield nearly identical runtime because both have the same number of leaf nodes (15,625). A similar pattern can be observed for $\branchfactor=10$ and $\branchfactor=20$. At $\branchfactor=15$, \gls{algff} achieves the lowest runtime, as it involves the smallest number of leaves (3,375). For $\branchfactor=30$, the number of leaf clusters significantly increases, leading to a substantial rise in the runtime of \gls{algff}. 
\iftechreport{Similarly, the \gls{NCA} as shown in~\Cref{fig:BranchNum_NCA_Healthcare} and~\Cref{fig:BranchNum_NCA_ACSIncome} exhibit the same trend as the runtime.}
\iftechreport{For \gls{NCE}, as shown in~\Cref{fig:BranchNum_NCE_Healthcare} and~\Cref{fig:BranchNum_NCE_ACSIncome}, the number of constraints evaluated remains constant across different branching factors \branchfactor. This is because the underlying data remains the same, and varying \branchfactor does not affect the set of constraints that need to be evaluated.}
For \gls{algrp}, overall performance trends align with those of \gls{algff}.
However, \gls{algrp} is less influenced by the branching factor as for smaller clusters it may be possible to prune / confirm larger candidate sets at once. 
Both bucket size \(S\) and branching factor \(B\) impact performance, and the optimal values depend on the characteristics of the dataset and queries. The intuition is as follows: when \(S\) is too small, the resulting tree becomes too deep, leading to an excessive number of leaf clusters; when \(S\) is too large, the ability to prune effectively diminishes because clusters encompass too many data points. Likewise, if \(B\) is too large, data is distributed across many child nodes, making it harder to prune entire sub-trees; if \(B\) is too small, the tree again becomes too deep with many leaf clusters. For new datasets, we suggest starting with moderate values for both \(S\) and \(B\), then adjusting based on the number of leaf clusters observed: if the tree has too many leaf clusters, consider increasing \(S\) or decreasing \(B\); if pruning is insufficient, consider decreasing \(S\) or increasing \(B\). There are several additional factors that can affect optimal choices for these parameters: (i) strong correlations between attributes used in conditions lead to more homogeneous clusters which in turn means that larger clusters can be tolerated without significantly impacting pruning power, (ii) if attributes in user query conditions are correlated with attributes of filter-aggregate queries, then aggregation results vary widely for clusters potentially leading to a stronger separation between repairs and more pruning potential even with larger clusters.
We leave automatic parameter tuning, e.g., based on measuring correlations between  attributes over a sample, to future work.

\ifnottechreport{
\begin{figure}[t]
    \centering
    \vspace{-3mm}
    \begin{subfigure}{0.47\linewidth}
        \centering
        \includegraphics[width=\linewidth, trim=0 40 0 0]{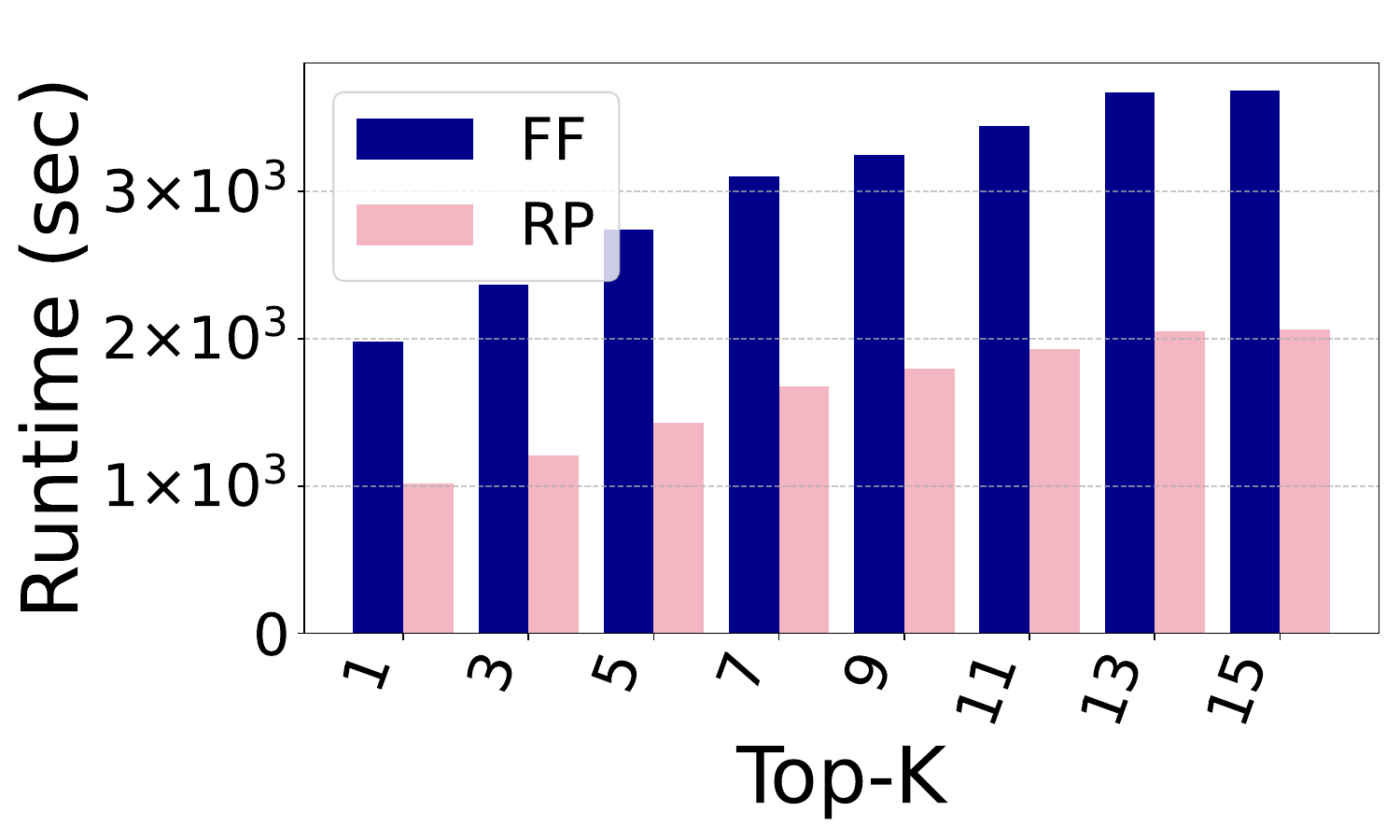}
        \caption{Runtime (sec)\newline}
        \label{fig:topk_Time}
    \end{subfigure}
    \begin{subfigure}{0.47\linewidth}
        \centering
        \includegraphics[width=\linewidth, trim=0 40 0 0]{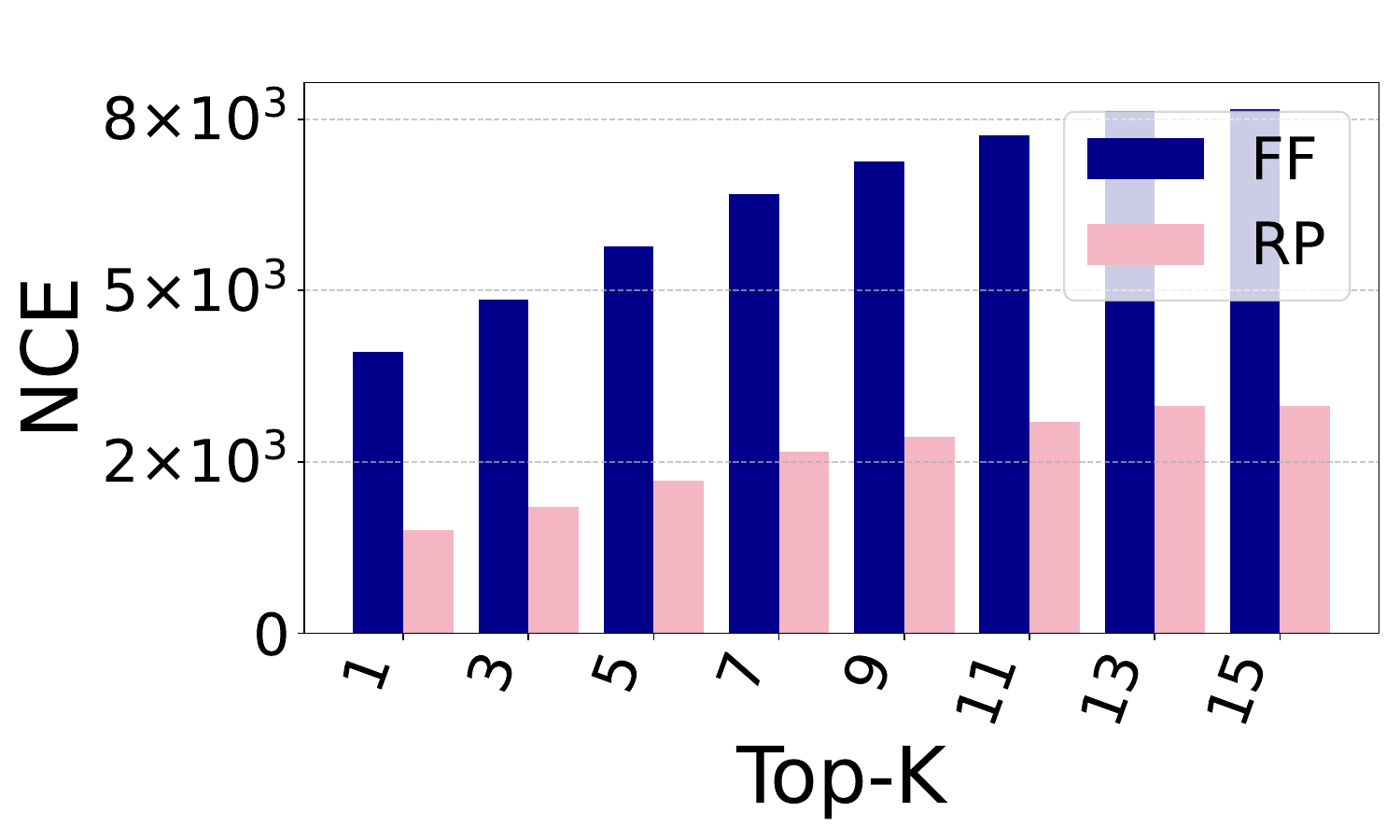}
        \caption{Number of constraints evaluated (NCE)}
        \label{fig:topk_NCE}
    \end{subfigure}
    \vspace{-3.5mm}
    \caption{Performance of \gls{algff} and \gls{algrp} varying $k$
    }
    \label{fig:TopK_comparison}
\end{figure}
}
\iftechreport{
\begin{figure*}[t]
    \centering
    \vspace{-3mm}
    \begin{subfigure}{0.3\linewidth}
        \centering
        \includegraphics[width=\linewidth, trim=0 20 0 0]{comparison_fully_vs_ranges_Time_ACSIncome_TopK.pdf}
        \caption{Runtime (sec).\newline}
        \label{fig:topk_Time}
    \end{subfigure}
    \hfill
    \begin{subfigure}{0.3\linewidth}
        \centering
        \includegraphics[width=\linewidth, trim=0 20 0 0]{comparison_fully_vs_ranges_CheckedNum_ACSIncome_TopK.pdf}
        \caption{Total number of constraints evaluated (NCE).}
        \label{fig:topk_NCE}
    \end{subfigure}
    \hfill
    \begin{subfigure}{0.3\linewidth}
        \centering
        \includegraphics[width=\linewidth, trim=0 20 0 0]{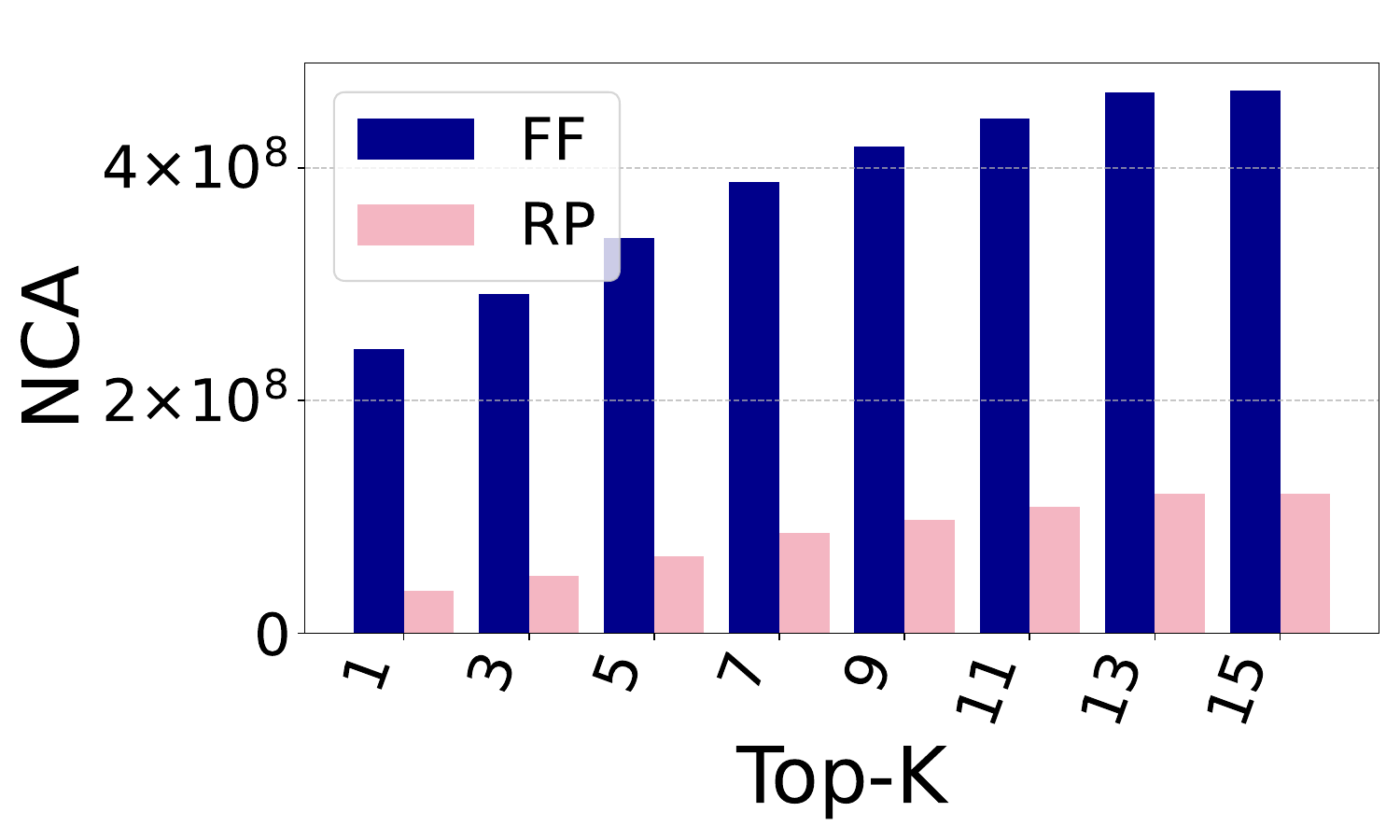}
        \caption{Total number of clusters accessed (NCA).}
        \label{fig:topk_NCA}
    \end{subfigure}
    \vspace{-2mm}
    \caption{Runtime, \gls{NCE}, and \gls{NCA} for \gls{algff} and \gls{algrp} over top-$k$.}
    \label{fig:TopK_comparison}
\end{figure*}
}

\mypar{Effect of $k$}
\label{subsub: effect of top-k}
In this experiment, we vary the parameter $k$ from 1 to 15.
For both \gls{algff} and \gls{algrp}, as $k$ increases, the runtime also increases, as shown in~\Cref{fig:topk_Time}.
This behavior is expected since finding a single repair ($k$=1) requires less computation than identifying multiple repairs. When $k$ is larger, the algorithm must explore a larger fraction of the search space to find additional repairs.
\ifnottechreport{Similarly, the \gls{NCE} as shown in~\Cref{fig:topk_NCE} exhibits the same increasing trend.}
\iftechreport{Similarly, both the \gls{NCE} as shown in~\Cref{fig:topk_NCE} and \gls{NCA} as shown in~\Cref{fig:topk_NCA} exhibit the same increasing trend.}
 \gls{algrp} consistently outperforms \gls{algff}.

\ifnottechreport{
\begin{figure}[t]
    \centering
    \vspace{-3mm}
    \begin{subfigure}{0.47\linewidth}
        \centering
        \includegraphics[width=\linewidth, trim=0 40 0 0]{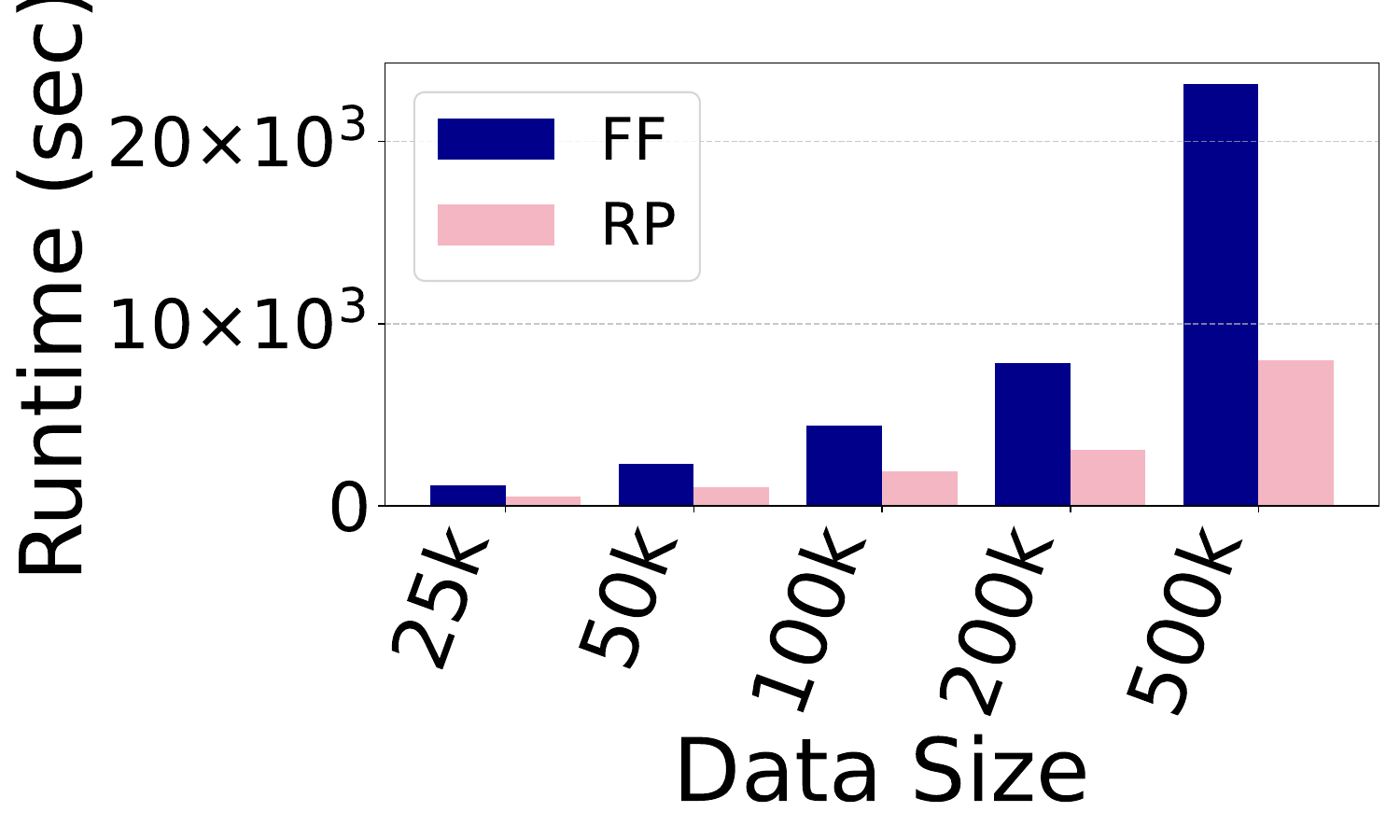}
        \caption{Runtime (sec)}
        \label{fig:DataSize_Time}
    \end{subfigure}
    \hfill
    \begin{subfigure}{0.47\linewidth}
        \centering
        \includegraphics[width=\linewidth, trim=0 40 0 0]{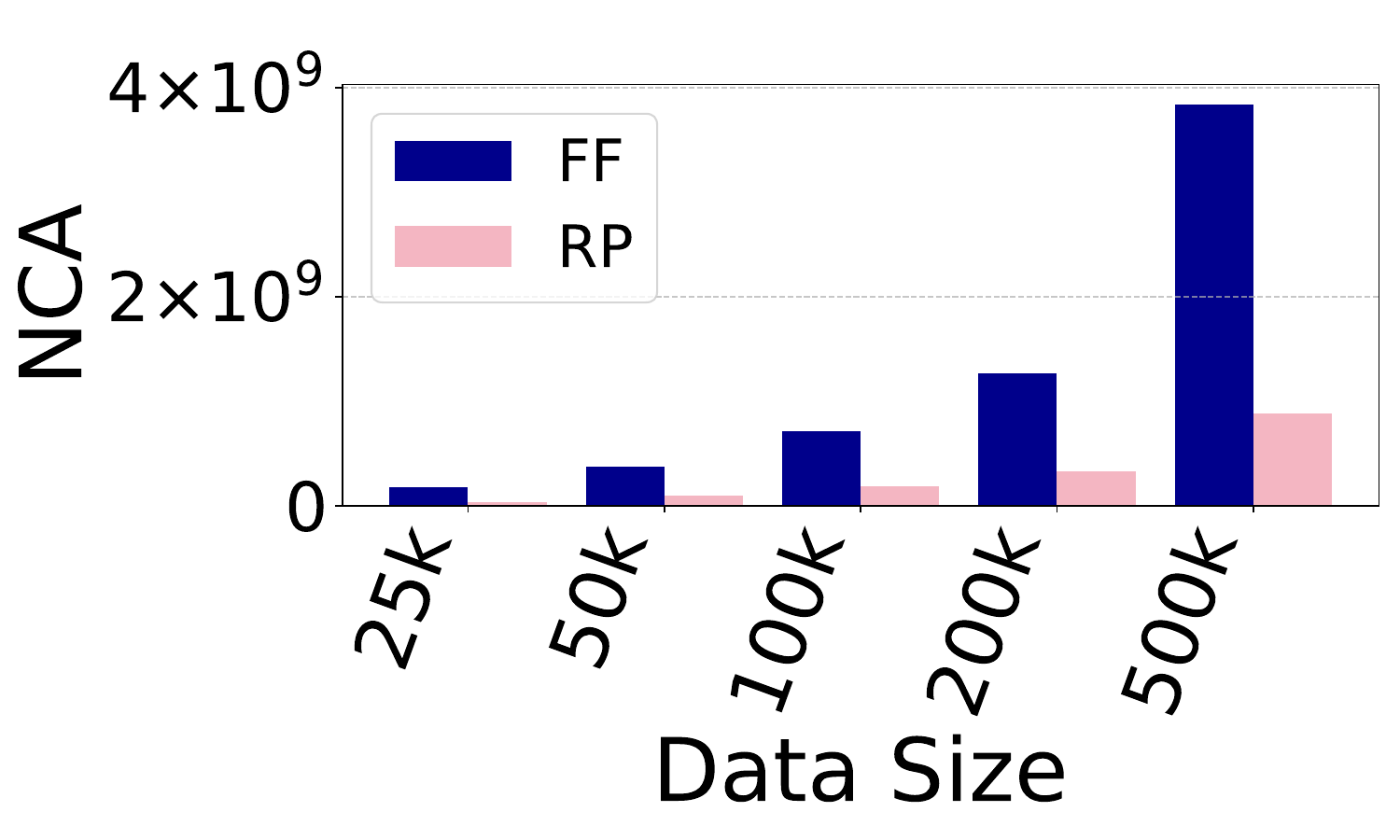}
        \caption{Number of clusters accesses}
        \label{fig:DataSize_Clusters_Accesses}
    \end{subfigure}
    \vspace{-4mm}
    \caption{Evaluation result over TPC-H dataset
    }
    \label{fig:BranchNum_comparison}
\end{figure}
}
\iftechreport{
\begin{figure*}[t]
    \centering
    \vspace{-3mm}
    \begin{subfigure}{0.3\linewidth}
        \centering
        \includegraphics[width=\linewidth, trim=0 40 0 0]{comparison_fully_vs_ranges_Time_TPCH_DataSize.pdf}
        \caption{Runtime (sec).\newline}
        \label{fig:DataSize_Time}
    \end{subfigure}
    \hfill
    \begin{subfigure}{0.3\linewidth}
        \centering
        \includegraphics[width=\linewidth, trim=0 40 0 0]{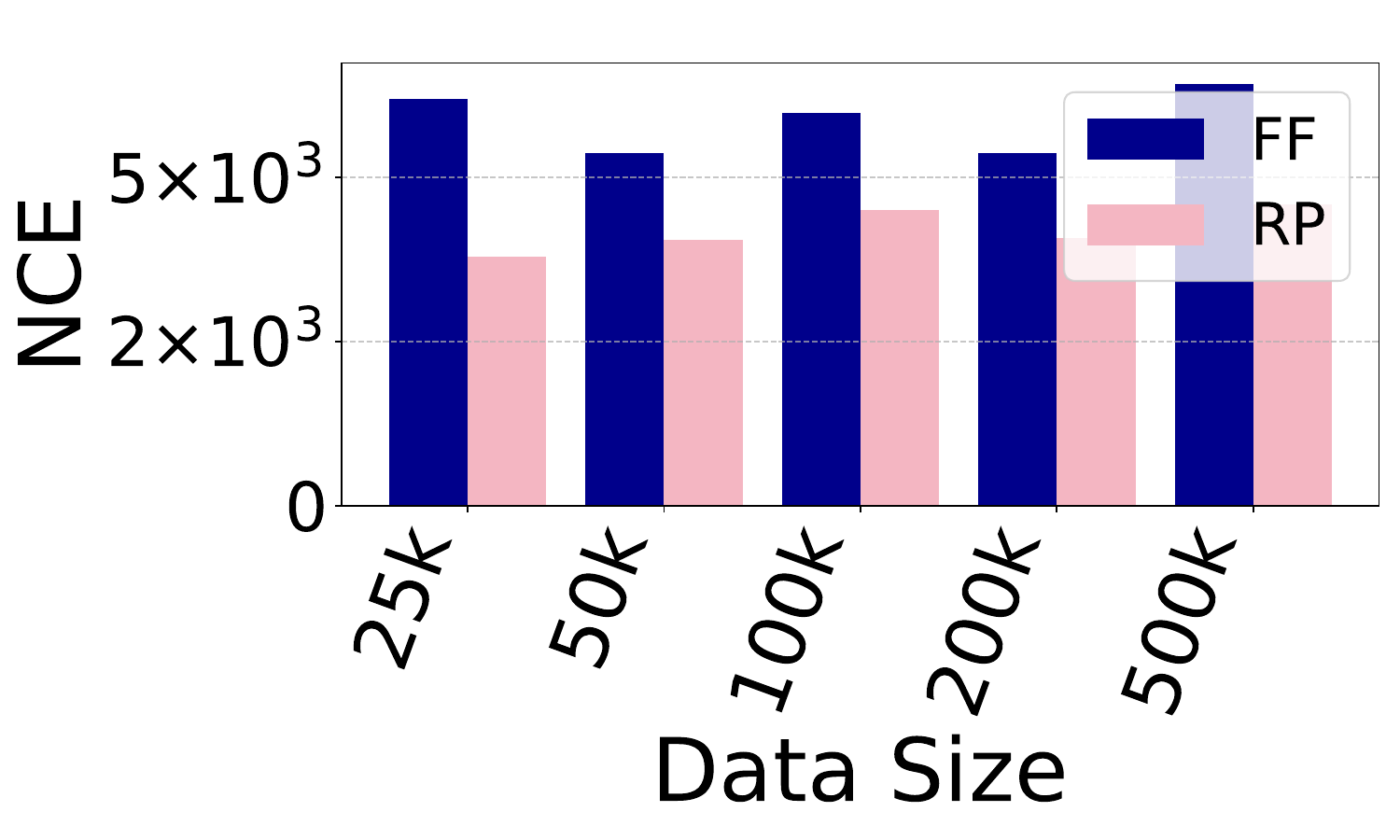}
        \caption{Total number of constraints evaluated (NCE).}
        \label{fig:DataSize_NCE}
    \end{subfigure}
    \hfill
    \begin{subfigure}{0.3\linewidth}
        \centering
        \includegraphics[width=\linewidth, trim=0 40 0 0]{comparison_fully_vs_ranges_AccessNum_TPCH_DataSize.pdf}
        \caption{Number of clusters accessed (NCA).\newline}
        \label{fig:DataSize_NCA}
    \end{subfigure}
    \vspace{-2mm}
    \caption{Runtime, \gls{NCE}, and \gls{NCA} for \gls{algff} and \gls{algrp} over the TPC-H dataset, varying data size.}
    \label{fig:DataSize_comparison}
\end{figure*}
}

\mypar{Effect of Dataset Size}
\label{subsub:Effect of data size}
Next, we vary the dataset size and measure the runtime\iftechreport{, \gls{NCE} }and the \gls{NCA} for the \gls{tpch} dataset using $Q_7$ with $\propconstr_{5}$. 
Dataset size impacts both the size of the search space and the size of the kd-tree. Nonetheless, as shown in~\ifnottechreport{\Cref{fig:DataSize_Time}} \iftechreport{\Cref{fig:DataSize_Time}}, our algorithms scale roughly linearly in dataset size demonstrating the effectiveness of reusing aggregation results for clusters and range-based pruning.
This is further supported by the \gls{NCA} measurements
shown in~\ifnottechreport{\Cref{fig:DataSize_Clusters_Accesses}}, \iftechreport{\Cref{fig:DataSize_NCE}}
which exhibit the same trend as the runtime.\iftechreport{ For \gls{NCE}, as shown in~\Cref{fig:DataSize_NCA}, the number of constraints evaluated varies across different dataset sizes. This variation occurs because the underlying data itself changes with the dataset size. This contrasts with the observations in~\Cref{fig:BucketSize_comparison} and~\Cref{fig:num-of-branch-perf}, where the number of evaluated constraints remained constant due to the data being fixed across configurations. These results confirm that the number of evaluated constraints is influenced by changes in the dataset content, rather than by variations in the branching factor \branchfactor or bucket size \bucketsize alone.}
%

\subsection{Comparison with Related Work} \label{subsec: compare with related work}
\BGI{\textbf{Moved from beginning of section, merge in here} Note that \gls{erica} uses a different optimization criterion than we do: as in skyline queries,  \gls{erica} returns all \emph{minimal refinements} which are repairs where none of the predicates can be further refined and still yield a repair. In contrast, we return the top-k repairs based on a weighted sum of distances between constants in predicates.}  We will explain in \Cref{subsec: compare with related work} how we achieve a fair comparison.
We compare our approach with \gls{erica}~\cite{LM23}, which solves the related problem of finding all minimal refinements of a given query that satisfy a set of cardinality constraints
for groups within the result set. Such constraints are special cases on the \glspl{AC} we support.  \gls{erica} returns all repairs that are not \emph{dominated} (the skyline~\cite{BK01b}) by any other repair where a repair dominates another repair if it is at least as close to the user query for every condition $\theta_i$ and strictly closer in at least one condition.  Thus, different from our approach, the number of returned repairs is not an input parameter in \gls{erica}. For a fair comparison, we compute the minimal repairs and then set $k$ such that our methods returns a superset of the repairs returned by \gls{erica}.
Our algorithms, like \gls{erica}, operate by modifying constants in predicates on attributes already present in the query and do not introduce new predicates. A key difference is that \gls{erica} supports adding constants to set membership predicates for categorical attributes, e.g., replacing $A \in \{c_1\}$ with $A \in \{c_1,c_2,c_3\}$, while our approach maps categorical values to numeric codes and adjusts thresholds. 
  As we will discuss in \Cref{sec:relwork}, both our approach and \gls{erica} can model addition of new predicates by refining dummy predicates that evaluate to true on all inputs.
%

To conduct the evaluation for \gls{erica},
we used
the available Python implementation~\cite{ericacode}\footnote{{We replaced Erica’s DataFrame filter checks and constraints evaluation (which run in C) with equivalent pure-Python loops over lists just as in our own code, so that both implementations are using the same programming language. This change ensures our comparison highlights algorithmic differences rather than language speed. 
}}.

%
We adopt the queries, constraints, and the dataset from~\cite{LM23}.
We compare the generated refinements and runtime of our techniques with \gls{erica} using $Q_1$ and $Q_2$ (\Cref{tab:queries}) on the \gls{dshealth} dataset (50K tuples) with constraints $\propconstrset_{6}$ and $\propconstrset_{7}$ (\Cref{tab:constraints}), respectively.


\mypar{ Generated Repairs Comparison}
We first compare the generated repairs by our approach and \gls{erica}.
As mentioned above, we did adjusted $k$ per query and constraint set to ensure that our approach returns a superset of the repairs returned by \gls{erica}.
For $Q_1$ with $\propconstrset_{6}$ ($Q_2$ with $\propconstrset_{7}$), \gls{erica} generates 7 (9) minimal repairs whereas our technique generates 356 (1035), including those produced by \gls{erica}. 
The top-1 repair returned by our approach is guaranteed to be minimal. However, the remaining minimal repairs returned by \gls{erica} may have a significantly higher distance to the user query than the remaining top-k answers returned by our approach.
For example, in $Q_2$, given the condition \texttt{num-children >= 4} of the user query, our solution includes a refined condition \texttt{num-children >= 3} whereas \gls{erica} provides a refinement \texttt{num-children >= 1} which is dissimilar to the user query.

\begin{figure}[t]
    \centering
    \begin{subfigure}[t]{0.46\linewidth}
        \centering
        \includegraphics[width=\linewidth,trim=0 30 0 0]{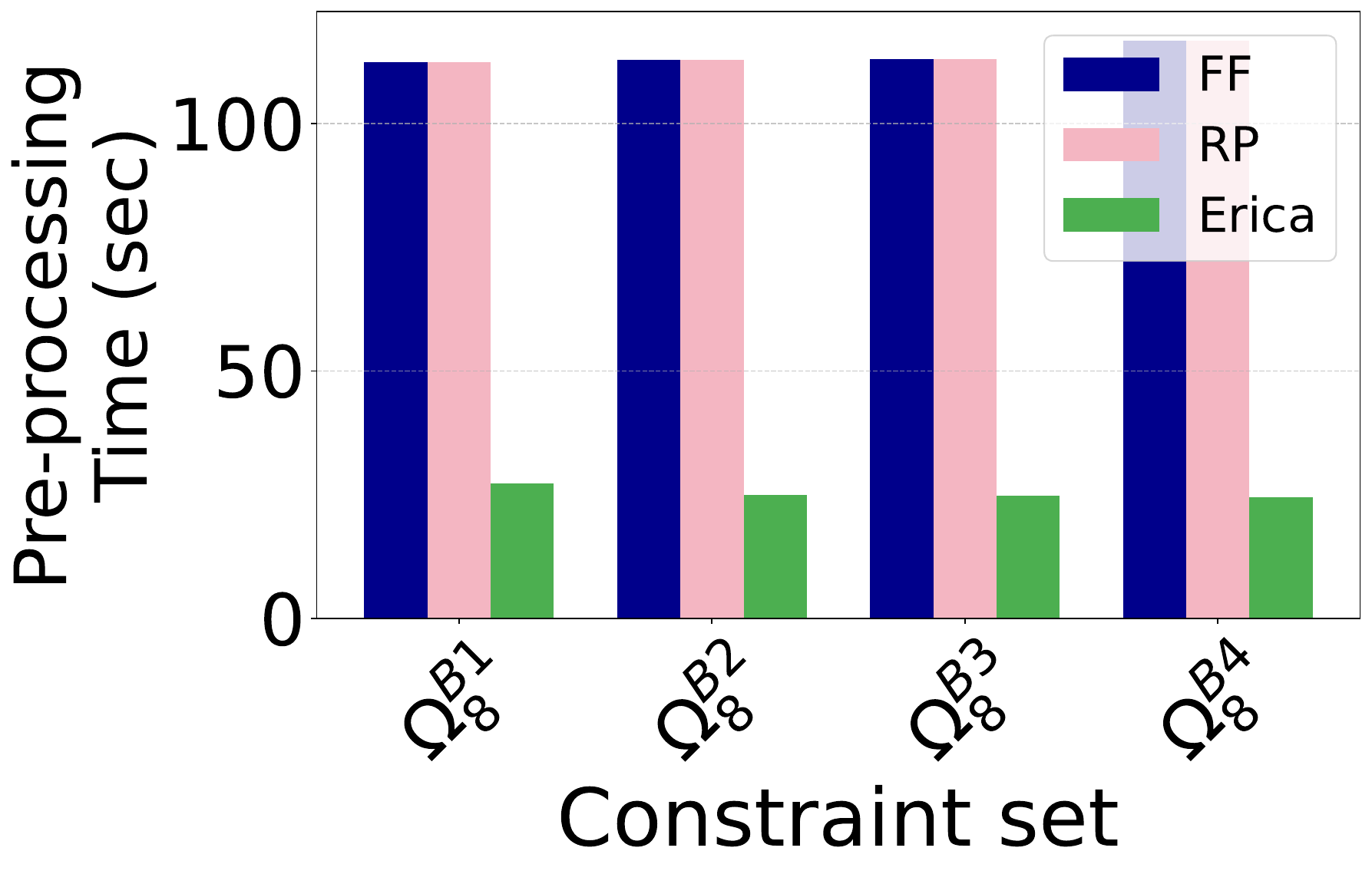}
        \caption{Pre-processing}
        \label{fig:Erica_ProvTime}
    \end{subfigure}
    \hfill
    \begin{subfigure}[t]{0.48\linewidth}
        \centering
        \includegraphics[width=\linewidth,trim=0 30 0 0]{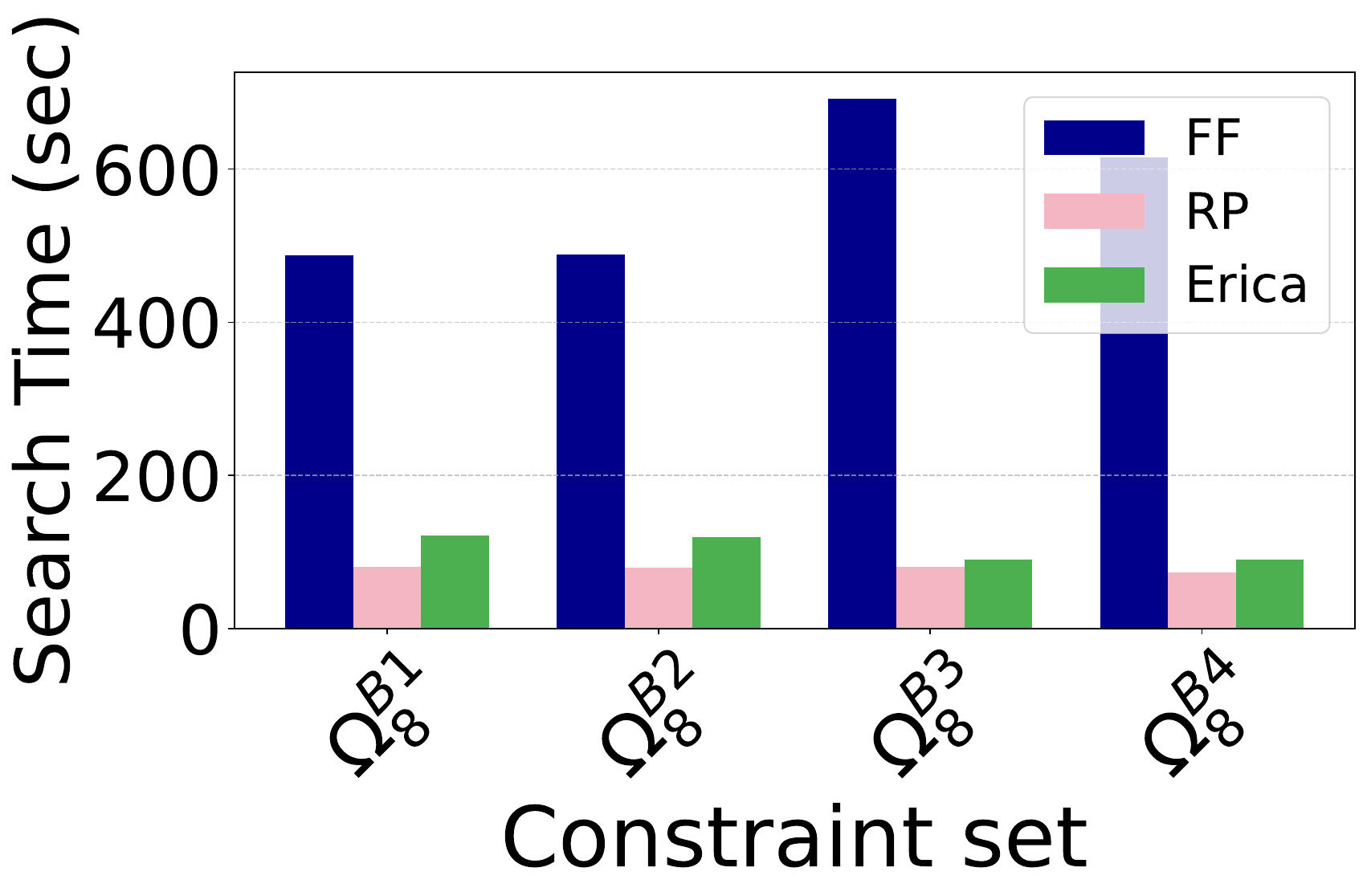}
        \caption{Search}
        \label{fig:Erica_SearchTime}
    \end{subfigure}
    \vspace{-3mm}
    \caption{Runtimes of \gls{algff}, \gls{algrp}, and \gls{erica}.}
    \label{fig:Eric_comparison}
\end{figure}

\mypar{Runtime Comparison}
The experiment utilizes $Q_4$ with $\propconstrset_{8}$
on the 50K \gls{dsacs} dataset, which is derived from \gls{erica}’s dataset, query, and constraint.
We use the same bounds in the constraints for both \gls{erica} and our algorithms:
$B1 \defas (B_{u_1} = 30, B_{u_2} =150, B_{u_3} =10)$ and $B2 \defas (B_{u_1} =30,B_{u_2} =300,B_{u_3} =25)$, $B3 \defas (B_{u_1} =10, B_{u_2} =650, B_{u_3} =50)$, and $B4 \defas (B_{u_1} =15,B_{u_2} = 200, B_{u_3} =15)$.
To ensure a fair comparison of execution time, we fix the number of generated repairs (i.e., top-$k$) in our approach to equal to the number of repairs produced by \gls{erica}. We set $k$=17 for constraint sets $\propconstrset_{8}^{B1}$ and $\propconstrset_{8}^{B2}$, $k$=11  for $\propconstrset_{8}^{B3}$,  and $k$=13 for $\propconstrset_{8}^{B4}$.
Due to the different optimization criterions, variations in the generated repairs between our approach and \gls{erica} are expected.
The results in~\Cref{fig:Erica_SearchTime} reveal an advantage of the \gls{algrp} algorithm, which outperforms \gls{erica} in  the time exploring the search space to generate a repair.
However, as shown in~\Cref{fig:Erica_ProvTime}, in pre-processing time which is the time of materializing aggregates and constructing the kd-tree for our methods and generating provenance expressions for \gls{erica}, \gls{erica} outperforms both \gls{algrp} and \gls{algff}. Erica’s pre-processing is faster because it only computes provenance expressions and, for each predicate, builds a list of candidate constants sorted by their distance to the original query. In contrast, our methods require clustering the data, indexing the clusters, and materializing summaries for each cluster, which is more computationally intensive. However, this extra work enables us to reason about complex, non-monotone constraints, which Erica’s simpler list-based approach cannot. Furthermore, we argue that it can be beneficial to decrease search time at the cost of higher preprocessing time as some of the preprocessing results could be shared across user requests.
Overall the total runtime of \gls{algrp} and \gls{erica} are comparable, even through our approach does not apply any specialized optimizations that exploit monotonicity as in \gls{erica}. 
These results also highlight the need for our range-based optimizations in \gls{algrp}, as \gls{algff} is significantly slower than \gls{erica}.


\section{Related Work} \label{sec:relwork}



\newcommand{\dmlin}{\textit{lin. comb.}}
\newcommand{\dmpareto}{\textit{skyline}}
\newcommand{\dmedit}{\textit{edit-distance}}
\newcommand{\dmjaccard}{\textit{Jacc. (result)}}
\newcommand{\dmresulteuclid}{\textit{L2 (result)}}

\begin{table}[t]
\centering
\caption{Comparison of query repair techniques}
\label{tab:techniques_comparison}
\scriptsize
\begin{tabular}{|p{1.2cm}|p{1cm}|p{1.1cm}|p{1.2cm}|p{0.6cm}|p{1.2cm}|}
\hline
\textbf{Approach}               & \textbf{Supports $\f{sum}$, $\f{min}$, $\f{max}$, $\f{avg}$ } & \textbf{Distance Metric} & \textbf{Constrains Result Subsets?} & \textbf{Repairs Joins?} & \textbf{Arithmetic Expressions Supported} \\ \hline
HC~\cite{1717427}               & \redx                                                         & \redx                    & \redx                               & \redx                   & \redx                                     \\ \hline
TQGen~\cite{MishraKZ08}         & \redx                                                         & \redx                    & \redx                               & \redx                   & \redx                                     \\ \hline
SnS~\cite{MK09}                 & \redx                                                         & \redx                    & \redx                               & \redx                   & \redx                                     \\ \hline
EAGER~\cite{AlbarrakS17}        & \greenok                                                      & \dmlin                   & \redx                               & \redx                   & \redx                                     \\ \hline
  SAUNA~\cite{KadlagWFH04}      & \redx                                                         & \dmresulteuclid          & \redx                               & \greenok                & \redx                                     \\ \hline
ConQueR~\cite{tran2010conquer}  & \redx                                                         & \dmedit                  & \redx                               & \greenok                & \redx                                     \\ \hline
FixTed~\cite{BidoitHT16}        & \redx                                                         & \dmpareto                & \redx                               & \greenok                & \redx                                     \\ \hline
FARQ~\cite{shetiya2022fairness} & \redx                                                         & \dmjaccard               & \greenok                            & \redx                   & \redx                                     \\ \hline
Erica~\cite{LM23}               & \redx                                                         & \dmpareto                & \greenok                            & \redx                   & \redx                                     \\ \hline \hline
\textbf{\gls{algrp} (ours)}     & \greenok                                                      & \dmlin                   & \greenok                            & \redx                   & \greenok                                  \\ \hline
\end{tabular}
\end{table}
\mypar{Query refinement \& relaxation}
Table~\ref{tab:techniques_comparison} summarizes several query refinement techniques for aggregate constraints, and compares their capabilities in terms of supported  aggregates (only $\f{count}$ or also other aggregates), distance metric use to compare repairs to the original query based on distances between predicates (e.g., \dmlin: linear combination of predicate-level distances, \dmpareto: skyline over predicate-level distances), whether the method allows constraints that apply only to a subset of the result (some methods only constraint the whole query result), whether join conditions can be repaired, and whether they support arithmetic expressions.
Li et al.~\cite{LM23} determine all minimal refinements of a conjunctive query by changing constants in selection conditions such that the refined query fulfills a conjunction of cardinality constraints, e.g., the query should return at least 5 answers where \texttt{gender} = female. A refinement is minimal if it fulfills the constraints and there does not exist any refinement that is closer to the original query in terms of similarity of constants used in predicates (\dmpareto).
However, ~\cite{LM23} only supports cardinality constraints ($\f{count}$) and does not allow for arithmetic combinations of the results of such queries as shown in Table~\ref{tab:techniques_comparison}.
Mishra et al.~\cite{MK09} refine a query to return a given number $k$ of results with interactive user feedback. Koudas et al.~\cite{KL06a} refine a query that returns an empty result to produce at least one answer. In ~\cite{tran2010conquer, BidoitHT16} a query is repaired to return missing results of interest provided by the user. 
Campbell et al.~\cite{campbell24} repair top-k queries, supporting non-monotone constraints through the use of constraint solvers.
\cite{1717427,MishraKZ08} refine queries for database testing such that subqueries of the repaired query approximately fulfill cardinality constraints. \cite{1717427} demonstrated that the problem is NP-hard in the number of predicates. Both approaches do not optimize for similarity to the user query. \cite{KadlagWFH04} relaxes a query to return approximately $N$ results preferring repairs based on the difference between the result of the user query and repair.
Most work on query refinement has limited the scope to constraints that are monotone in the size of the query answer.
Monotonicity is then exploited to prune the search space~\cite{bruno-06-gqwccdt, MK09, MishraKZ08, KadlagWFH04, WM13}.
To the best of our knowledge, our approach is the only one that supports arithmetic constraints which is necessary to express complex real world constraints, e.g., standard fairness measures, but requires novel pruning techniques that can handle such non-monotone constraints. While some approaches explicitly support adding and deletion of predicates, any approach that can both relax or refine predicates and support deletion by relaxing a predicate until it evaluate to true on all inputs and addition by adding dummy predicates that evaluate to true on all inputs and then either refine them (adding a new predicate) or not (decide to not add this predicate).


\mypar{How-to queries} Like in query repair~\cite{LM23}, the goal of how-to queries~\cite{MeliouS12}  is to achieve a desired change to a query's result. However, how-to queries change the database to achieve this result instead of repairing the query. Wang et al.~\cite{wang-17-q} study the problem of deleting operations from an update history to fulfill a constraint over the current database. 
However, this approach does not consider query repair (changing predicates) nor aggregate constraints.

\mypar{Explanations for Missing Answers}
Query-based explanations for missing answers~\cite{ChapmanJ09, DeutchFGH20, DiestelkamperLH21} are sets of operators that are responsible for the failure of a query to return a result of interest. However, this line of work
does not generate query repairs.

\mypar{Bounds with Interval Arithmetic}
Prior work has highlighted the effectiveness of interval arithmetic across various database applications \cite{zhang2007efficient, de2004affine, stolfi2003introduction, feng2021efficient}. For instance, \cite{feng2021efficient} determines bounds on query results over uncertain database. 
Similarly, the work in \cite{zhang2007efficient} introduced a bounding technique for iceberg cubes, establishing an early foundation for leveraging interval arithmetic to constrain aggregates. Interval arithmetic has been used extensively in \emph{abstract interpretation}~\cite{cousot-77-ab,stolfi2003introduction,de2004affine} to bound the result of computations. \iftechreport{For example, see \cite{stolfi2003introduction, de2004affine} for introductions to interval arithmetic and more advanced numerical abstract domain.}


\section{Conclusions and Future Work} \label{sec:conclusion}

We introduce a novel approach for repairing a query to satisfy a constraint on the query's result. We support a significantly larger class of constraints than prior work, including common fairness metrics like SPD. We avoid redundant work by reusing aggregate results when evaluating repair candidates and present techniques for evaluating multiple repair candidates at once by bounding their results. Our approach works best if there is homogeneity among similar repair candidates that can be exploited.
Interesting directions for future work include (i) the study of more general types of repairs, e.g., repairs that add or remove joins or change the structure of the query, (ii) considering other optimization criteria, e.g., computing a skyline as in some work on query refinement, (iii) employing more expressive domains than intervals for computing tighter bounds, e.g., zonotopes~\cite{de2004affine}, and (iv) supporting dynamic settings where the table, predicates, constraints, or distance metrics may change. In this regard, we may exploit efficient incremental maintenance kd-trees and aggregate summaries~\cite{bentley1980decomposable, bentley1975multidimensional}. However, our setting is more challenging as small changes to aggregation results can affect the validity of large sets of repair candidates.



\bibliographystyle{ACM-Reference-Format}
\bibliography{references}

\clearpage

\end{document}
\endinput